\documentclass[11pt,a4paper,leqno]{article}
\usepackage{subfig}
\usepackage{graphicx}
\usepackage{makeidx}
\usepackage{pst-node}
\usepackage{amsmath,amsthm}
\usepackage{enumerate}
\usepackage{verbatim}
\usepackage{framed}
\usepackage{ulem}
\usepackage{bbold}
\usepackage{natbib}
\usepackage{mdframed}

\newtheorem{theorem}{Theorem}
\newtheorem{lemma}[theorem]{Lemma} 
\newtheorem{proposition}[theorem]{Proposition} 
\newtheorem{corollary}[theorem]{Corollary}

\theoremstyle{definition}

\title{
Inheritance of Convexity for the $\tilde{\mathcal{P}}_{\min}$-Restricted Game
}

\author{A. Skoda\thanks
{Corresponding
    author. Universit\'e de Paris I, Centre d'Economie de la Sorbonne, 106-112
    Bd de l'H\^opital, 75013 Paris, France. E-mail: 
		{\tt
      alexandre.skoda@univ-paris1.fr}} }
			
\begin{document}
\def\R{{\bf \mbox{I\hspace{-.17em}R}}}
\def\N{\mbox{I\hspace{-.15em}N}}

\maketitle


\begin{abstract}
We consider a restricted game on weighted graphs
associated with minimum partitions.
We replace in the classical definition of Myerson restricted game the connected components of any subgraph
by the sub-components obtained with a specific partition $\tilde{\mathcal{P}}_{\min}$.
This partition relies on the same principle
as the partition $\mathcal{P}_{\min}$ introduced by \cite{GrabischSkoda2012}
but restricted to connected coalitions.
More precisely,
this new partition $\tilde{\mathcal{P}}_{\min}$ is induced
by the deletion of the minimum weight edges in each connected component associated with a coalition.
We provide a characterization of the graphs satisfying inheritance of convexity
from the underlying game to 
the restricted game associated with $\tilde{\mathcal{P}}_{\min}$.
\end{abstract}



\section{Introduction}
We consider
a finite set $N$ of players
with $|N| = n$.
Let $\mathcal{P}$
be a correspondence on $N$ 
associating to every subset $A \subseteq N$
a partition $\mathcal{P}(A)$ of~$A$.
For any game $(N,v)$,
\cite{Skoda201702} defined the
restricted game
$(N,\overline{v})$ associated with $\mathcal{P}$ by:
\begin{equation}
  \overline {v} (A) = \sum_{F \in \mathcal{P}(A)} v(F),  \; \; \textrm{for all}\; A \subseteq N.
\end{equation}
We refer to this game as the $\mathcal{P}$-restricted game.
$v$ is the characteristic function of the underlying game,
$v: 2^{N} \rightarrow \R$, $A \mapsto v(A)$
and satisfies $v(\emptyset) = 0$.
Many correspondences have been considered in the literature
to take into account communication or social restrictions.
The first founding example is the Myerson's correspondence
$\mathcal{P}_M $ associated with communication games \citep{Myerson77}.
Communication games are
cooperative games $(N, v)$ defined on the set of vertices
$N$ of an undirected graph $G=(N,E)$,
where $E$ is the set of edges.
$\mathcal{P}_{M}$ associates to every coalition $A \subseteq N$
the partition of $A$ into connected components.
The $\mathcal{P}_{M}$-restricted game $(N, \overline{v})$,
known as Myerson restricted game,
takes into account
the connectivity between players in $G$.
Only connected coalitions are able to cooperate and get their initial
values.
Many other correspondences have been considered
to define restricted games
(see, e.g.,
\citet{AlgabaBilbaoLopez2001,Bilbao2000,Bilbao2003,Faigle89,GrabischSkoda2012,Grabisch2013}).

For a given correspondence,
a classical problem is to study inheritance of a property
from the initial game $(N,v)$ to the restricted game $(N,\overline{v})$.
Inheritance of properties has been thoroughly studied for the Myerson correspondence $\mathcal{P}_{M}$
(see \cite{Owen86,NouwelandandBorm1991,Slikker2000}).
Inheritance of convexity is of special interest
as it implies that lots of appealing properties are also inherited,
for instance superadditivity,
non-emptiness of the core,
and that the Shapley value lies in the core.
\cite{Skoda201702} obtained an abstract characterization of inheritance of convexity
for an arbitrary correspondence
using a cyclic intersecting sequence condition on coalitions.
This result implies some of the characterizations obtained for specific correspondences.
In particular,
it implies the characterization of inheritance of convexity for
the Myerson correspondence by cycle-completeness of the underlying graph,
established by \cite{NouwelandandBorm1991}.
For correspondences associated with graphs,
it is of course more interesting to find characterizations in terms of graph structures
as for the Myerson correspondence.

\cite{GrabischSkoda2012} introduced
the correspondence $\mathcal{P}_{\min}$ for communication games on weighted graphs.
A communication game on a weighted
graph is a combination of a cooperative game $(N, v)$ and a weighted graph
$G = (N,E,w)$ which has the set of players as its vertices and where $w$ is
a weight function defined on the set $E$ of edges of $G$.
In this context,
it is likely that players belonging to a given coalition 
are more or less prone to cooperate 
depending on the weights of their links.
The correspondence $\mathcal{P}_{\min}$
takes into account connectedness of the players
but a coalition gets its initial value under a stronger requirement.
There must be some privileged relation
between players to activate their cooperation.
For a given coalition $A \subseteq N$,
we denote by $E(A)$ the set of edges with both end-vertices in $A$,
and by $\Sigma(A)$ the set of minimum weight edges in $E(A)$.
It is assumed that two players have a privileged relation in a coalition $A \subseteq N$
if they are linked by an edge with weight strictly greater
than the minimum edge-weight in the subgraph $G_A=(A, E(A))$.
More precisely,
the correspondence $\mathcal{P}_{\min}$
associates with any coalition $A\subseteq N$
the partition $\mathcal{P}_{\min}(A)$ of $A$ into the connected components
of the subgraph $(A, E(A) \setminus \Sigma(A))$.
Then,
the $\mathcal{P}_{\min}$-restricted game $(N,\overline{v})$ is defined by:
\[
\overline{v}(A)= \sum_{F \in \mathcal{P}_{\min}(A)} v(F),  \textrm{ for all } A \subseteq N.
\]
\cite{GrabischSkoda2012}
established three necessary conditions on the underlying graph $G$
to have inheritance of convexity with the correspondence~$\mathcal{P}_{\min}$.
To establish these conditions,
they only had to consider connected subsets.
Hence,
these conditions are valid assuming only $\mathcal{F}$-convexity
which is a weaker condition than convexity introduced by~\cite{GrabischSkoda2012},
obtained by restricting convexity to connected subsets.
\cite{Skoda2017} presented a characterization
of inheritance of $\mathcal{F}$-convexity for $\mathcal{P}_{\min}$
by five necessary and sufficient conditions on the edge-weights of specific subgraphs.
These subgraphs correspond to stars, paths,
cycles, pans, and adjacent cycles
of the underlying graph $G$.
Finally,
\cite{Skoda2019b} obtained a characterization of inheritance of convexity for $\mathcal{P}_{\min}$.
As convexity implies $\mathcal{F}$-convexity,
the conditions established by \cite{Skoda2017} are necessary
but they do not appear in the characterization of inheritance of convexity.
Indeed,
this last characterization relies on more straightforward conditions.
This is in part due
to the fact that inheritance of convexity restricts the
edge-weights to at most three different values.

In this paper,
we establish a characterization of inheritance of convexity for a new correspondence $\tilde{\mathcal{P}}_{\min}$.
This correspondence
is close to the correspondence $\mathcal{P}_{\min}$
as they coincide on connected coalitions of players.
The correspondence $\tilde{\mathcal{P}}_{\min}$ follows the same pattern
as the correspondence $\mathcal{P}_{\min}$ 
but it requires that players cooperate only if they are in a privileged relation
relatively to the connected component they belong to.
More precisely,
let $A \subseteq N$ be a coalition of players
and let $A_{1}, A_{2}, \ldots, A_{p}$ with $p \geq 1$ be the connected components of $G_A$.
Then,
$\tilde{\mathcal{P}}_{\min}(A) = \lbrace \mathcal{P}_{\min}(A_1), \ldots, \mathcal{P}_{\min}(A_p) \rbrace$
and the $\tilde{\mathcal{P}}_{\min}$-restricted game $(N, \hat{v})$ is defined by
\begin{equation}
\hat{v}(A)=
\sum_{F \in \tilde{\mathcal{P}}_{\min}(A)} v(F), \; \textrm{for all} \; A \subseteq N.
\end{equation}
By definition of $\tilde{\mathcal{P}}_{\min}$,
we also have the following relation:
\begin{equation}
\hat{v}(A)=
\sum_{l = 1}^{p} \overline{v}(A_{l}), \; \textrm{for all} \; A \subseteq N.
\end{equation}
Hence,
the $\tilde{\mathcal{P}}_{\min}$-restricted game
also
corresponds to the Myerson restricted game 
associated with the $\mathcal{P}_{\min}$-restricted game.
As a consequence,
the $\tilde{\mathcal{P}}_{\min}$-restricted game $(N, \hat{v})$
and the $\mathcal{P}_{\min}$-restricted game $(N, \overline{v})$
assign the same values to connected coalitions.
In particular,
this implies that 
$\mathcal{F}$-convexity of the $\tilde{\mathcal{P}}_{\min}$-restricted game
is equivalent to
$\mathcal{F}$-convexity of the $\mathcal{P}_{\min}$-restricted game.
As a result,
the characterization of inheritance of $\mathcal{F}$-convexity
obtained by~\cite{Skoda2017}
for the correspondence $\mathcal{P}_{\min}$
is also valid for
the correspondence $\tilde{\mathcal{P}}_{\min}$.

In the present paper,
we establish a characterization of inheritance of convexity 
for the correspondence $\tilde{\mathcal{P}}_{\min}$.
This characterization includes the five necessary
and sufficient conditions on the edge-weights characterizing inheritance of $\mathcal{F}$-convexity. 
We have to add three new conditions to these five previous conditions
to obtain a characterization of inheritance of classical convexity.
These new conditions can be seen as reinforcements
of the preceding conditions established on cycles, pans, and adjacent cycles 
in the characterization of $\mathcal{F}$-convexity.
We prefer to state them separately
as they are more specific than the previous ones.
In particular,
their necessity is obtained considering non-connected coalitions,
whereas these last coalitions are not considered with $\mathcal{F}$-convexity.
Moreover,
our result implies that,
in the case of cycle-free graphs,
the two conditions on stars and paths
are necessary and sufficient for inheritance of convexity
with the correspondence $\tilde{\mathcal{P}}_{\min}$.
A similar result holds for inheritance of $\mathcal{F}$-convexity 
with the correspondence $\mathcal{P}_{\min}$
as proved by~\citet{GrabischSkoda2012}.
Hence,
there is equivalence between
inheritance of $\mathcal{F}$-convexity 
with the correspondence $\mathcal{P}_{\min}$
and inheritance of convexity with $\tilde{\mathcal{P}}_{\min}$
in the case of cycle-free graphs.
There is no such equivalence
with inheritance of convexity with $\mathcal{P}_{\min}$.
Indeed,
\cite{Skoda2019b}
proved that inheritance of convexity with $\mathcal{P}_{\min}$
restricts the number of different edge-weights
to at most three
even in the case of cycle-free graphs
and
there is no such limitation
for inheritance of convexity with $\tilde{\mathcal{P}}_{\min}$.
We also obtain that inheritance of convexity and inheritance of convexity restricted to the class of unanimity games
are equivalent for the correspondence $\tilde{\mathcal{P}}_{\min}$.
This last result also holds for $\mathcal{P}_{\min}$ (\cite{Skoda2017})
and was already observed
in a more general setting by~\cite{Skoda201702}.

The article is organized as follows.
In Section~\ref{SectionPreliminaryDefinitionsResults} we give preliminary definitions and
results established by~\cite{GrabischSkoda2012}.
In particular,
we recall
the definition of convexity,
$\mathcal{F}$-convexity
and general conditions on a correspondence
to have inheritance of superadditivity,
convexity
or $\mathcal{F}$-convexity.
We recall in Section~\ref{SectionInheritanceOfF-Convexity}
the necessary and sufficient conditions on the graph and the weight vector $w$
established by~\cite{Skoda2017}
for inheritance of $\mathcal{F}$-convexity
from the original game $(N,v)$ 
to the $\mathcal{P}_{\min}$-restricted game $(N,\overline{v})$.
Section~\ref{SubsectionInheritanceOfConvexityForTheTildePmin-RestrictedGame}
includes the main results ot the paper.
We establish three supplementary necessary conditions
on the edge-weights to have inheritance of convexity for the correspondence $\tilde{\mathcal{P}}_{\min}$.
Then,
we prove that these three conditions appended to the five previous conditions
necessary for inheritance of $\mathcal{F}$-convexity are also sufficient
to have inheritance of convexity for $\tilde{\mathcal{P}}_{\min}$.
We conclude with some remarks and suggestions in Section~\ref{Conclusion}.

\section{Preliminary definitions and results}
\label{SectionPreliminaryDefinitionsResults}
Let $N$ be a given set
with $|N|= n$.
We denote by $2^{N}$ the set of all subsets of $N$.
A game $(N,v)$ is \textit{zero-normalized}\label{definitionZeroNormalizedGame}
if $v(\lbrace i \rbrace) = 0$ for all $i \in N$.
We recall that a game $(N,v)$ is \emph{superadditive} if,
for all $A,B \in 2^{N}$ 
such that $A \cap B = \emptyset$,
$v(A \cup B) \geq v(A) + v(B)$.
For any given subset $\emptyset \not= S \subseteq N$, the unanimity game $(N, u_{S})$ is defined by:
\begin{equation}
u_{S}(A) =
\left \lbrace
\begin{array}{cl}
1 &  \textrm{ if } A \supseteq S,\\
0 & \textrm{ otherwise.}
\end{array}
\right.
\end{equation}
We note that $u_{S}$ is superadditive for all $S \not= \emptyset$.
Let $X$ and $Y$ be two given sets.
A \textit{correspondence} $f$
with domain $X$ and range $Y$
is a map that associates to every element $x \in X$
a subset $f(x)$ of $Y$,
\emph{i.e.},
a map from $X$ to $2^{Y}$.
In this work we consider specific correspondences $\mathcal{P}$
with domain and range $2^{N}$,
such that for every subset $\emptyset \not= A \subseteq N$,
the family $\mathcal{P}(A)$ of subsets of $N$
corresponds to a partition of $A$.
We set $\mathcal{P}(\emptyset) = \lbrace \emptyset \rbrace$.

For a given correspondence $\mathcal{P}$ on $N$
and subsets $A \subseteq B \subseteq N$,
we denote by $\mathcal{P}(B)_{|A}$ the restriction of the partition $\mathcal{P}(B)$ to $A$.
For two given subsets $A$ and $B$ of $N$,
$\mathcal{P}(A)$ is a refinement of $\mathcal{P}(B)$
if every block of $\mathcal{P}(A)$ is a subset of some block of $\mathcal{P}(B)$.

We recall the following results
established by~\cite{GrabischSkoda2012}.
The first one gives general conditions on a correspondence $\mathcal{P}$
to have inheritance of superadditivity.

\begin{theorem}
\label{thLGNEPNFaSNuSsFaSNuSin}
Let $N$ be an arbitrary set and $\mathcal{P}$ a correspondence on $2^N$.
The following conditions are equivalent:
\begin{enumerate}[1)]
\item
The $\mathcal{P}$-restricted game $(N, \overline{u_{S}})$ is superadditive
for all $\emptyset \not= S \subseteq N$.
\item
\label{Claim3ThLGNEPNFaSNuSsFaSNuSin}
$\mathcal{P}(A)$ is a refinement of $\mathcal{P}(B)_{|A}$
for all subsets $A \subseteq B \subseteq N$.
\item
The $\mathcal{P}$-restricted game $(N,\overline{v})$ is superadditive
for all superadditive game $(N,v)$.
\end{enumerate}
\end{theorem}
As $\mathcal{P}_{\min}(A)$
(resp. $\tilde{\mathcal{P}}_{\min}(A)$)
is a refinement of $\mathcal{P}_{\min}(B)_{|A}$
(resp. $\tilde{\mathcal{P}}_{\min}(B)_{|A}$)
for all subsets $A \subseteq B \subseteq N$,
Theorem~\ref{thLGNEPNFaSNuSsFaSNuSin} implies the following result.

\begin{corollary}
\label{corIGNEiafttsihfv}
Let $G=(N,E,w)$ be an arbitrary weighted graph.
The $\mathcal{P}_{\min}$-restricted game $(N,\overline{v})$
(resp. the $\tilde{\mathcal{P}}_{\min}$-restricted game $(N,\hat{v})$)
is superadditive
for every superadditive game $(N,v)$.
\end{corollary}

Let us consider a game $(N,v)$.
For arbitrary subsets $A$ and $B$ of $N$, we define the value:
\[
\Delta v(A,B):=v(A\cup B)+v(A\cap B)-v(A)-v(B).
\]
A game $(N,v)$ is \textit{convex}
if its characteristic function $v$ is supermodular,
\emph{i.e.},
$\Delta v(A,B)\geq 0$
for all $A,B \in 2^{N}$.
We note that $u_{S}$ is supermodular for all $S \not= \emptyset$.
Let $\mathcal{F}$
be a
\emph{weakly union-closed family}\footnote{
$\mathcal{F}$ 
is weakly union-closed if $A \cup B \in \mathcal{F}$ for all $A$, $B \in \mathcal{F}$
such that $A \cap B \not= \emptyset $ \citep{FaigleGrabisch2010}.
Weakly union-closed families were introduced
and analysed by~\cite{Algaba98}
(see also \cite{AlgabaBilbaoBormLopez2000}) and called union stable systems.
}
of subsets of $N$
such that $\emptyset \notin \mathcal{F}$.
A game $v$ on $2^{N}$ is said to be \emph{$\mathcal{F}$-convex}
if $\Delta v(A,B) \geq 0$,
for all $A,B \in \mathcal{F}$
with $A \cap B \in \mathcal{F}$.
Let us note that
a game $(N,v)$ is convex if and only if
it is superadditive and $\mathcal{F}$-convex
with $\mathcal{F} = 2^{N} \setminus \lbrace \emptyset \rbrace$.
Of course,
convexity implies $\mathcal{F}$-convexity.
For any $i \in N$
and any subset $A \subseteq N \setminus \lbrace i \rbrace$,
the derivative of $v$ at $A$ w.r.t $i$ is defined by
\[
\Delta_i v(A) := v(A \cup \lbrace i \rbrace) - v(A).
\]
$\Delta_i v(A)$ is also known as the marginal contribution of player $i$ w.r.t coalition~$A$. 
If a game $v$ on $2^{N}$ is $\mathcal{F}$-convex then,
for all $i \in N$ and all $A \subseteq B \subseteq N \setminus \lbrace i \rbrace$
with $A, B$ and $A \cup \lbrace i \rbrace \in \mathcal{F}$ 
we have:
\begin{equation}
\Delta_i v(B) \geq \Delta_i v(A).
\end{equation}

For a given graph $G= (N,E)$,
we say that a subset $A \subseteq N$ is connected
if the induced graph $G_{A} = (A,E(A))$ is connected.
The family of connected subsets of $N$
is obviously weakly union-closed.
For this last family,
the following result holds. 

\begin{theorem}\label{theoremLGNEbafalFbtfocsoN}
Let $G=(N,E)$ be an arbitrary graph and let $\mathcal{F}$ be the family of connected subsets of $N$.
The following conditions are equivalent:
\begin{equation}\label{eqv(AB)+v(AB)v(A)+v(B)}
 v \textrm{ is } \mathcal{F}\textrm{-convex}.
\end{equation}
\begin{equation}\label{eqvBi-vB=vAi-vA}
\begin{array}{c}
\Delta_i v(B) \geq \Delta_i v(A),\\
\forall i \in N,\,
\forall A \subseteq  B \subseteq  N\setminus  \lbrace i \rbrace
\textrm{ with } A, B, \textrm{ and } A \cup \lbrace i \rbrace \in \mathcal{F}.
\end{array}
\end{equation}
\end{theorem}

The next theorem gives general conditions on a correspondence $\mathcal{P}$
to have inheritance of convexity for unanimity games.

\begin{theorem}
\label{thLGNuPNFNFSNuSSNNuSFABFABF}
Let $N$ be an arbitrary set and $\mathcal{P}$ a correspondence on $2^N$.
Let $\mathcal{F}$ be a weakly union-closed family of subsets of $N$
with $\emptyset \notin \mathcal{F}$.
If the $\mathcal{P}$-restricted game $(N, \overline{u_{S}})$ is superadditive
for all $\emptyset \not= S \subseteq N$,
then the following conditions are equivalent.
\begin{enumerate}[1)]
\item
The $\mathcal{P}$-restricted game $(N, \overline{u_{S}})$ is $\mathcal{F}$-convex
for all $\emptyset \not= S \subseteq N $.
\item
\label{itemThForAllA,BInFWithACapBInF,P(ACapB)=A_jCapB_k;A_jInP(A)B_kInP(B),A_jCapB_kNot=Emptyset}
For all $A, B \in \mathcal{F}$
with $A \cap B \in \mathcal{F}$,
$\mathcal{P}(A \cap B) = \lbrace A_{j} \cap B_{k} \, ;\,  A_{j} \in \mathcal{P}(A), B_{k} \in \mathcal{P}(B), A_{j} \cap B_{k} \not= \emptyset  \rbrace$.
\end{enumerate}
Moreover if $\mathcal{F} = 2^{N} \setminus \lbrace \emptyset \rbrace$
or if $\mathcal{F}$ corresponds to the set of connected subsets of a graph
then 1) and 2) are equivalent to:
\begin{enumerate}[3)]
\item
\label{thClaim3'}
For all $i \in N$,
for all $ A \subseteq B \subseteq N \setminus \lbrace i \rbrace$
with $A, B$, and $A \cup \lbrace i \rbrace  \in \mathcal{F}$,
and for all $A' \in \mathcal{P}(A \cup \lbrace i \rbrace)_{|A}$, $\mathcal{P}(A)_{|A'} = \mathcal{P}(B)_{|A'}$.
\end{enumerate}
\end{theorem}

We finally recall the following lemma.

\begin{lemma}
\label{lemmavB+i=1pvBiA}
Let us consider $A, B \subseteq N$
and a partition $\lbrace B_{1}, B_{2}, \ldots, B_{p} \rbrace$ of~$B$.
Let $\mathcal{F}$ be a weakly union-closed family of subsets of~$N$
with $\emptyset \notin \mathcal{F}$.
If $A, B_{i}$, and $ A\cap B_{i} \in \mathcal{F}$
for all $i \in \lbrace 1, \ldots, p \rbrace$,
then for every $\mathcal{F}$-convex game $(N,v)$ we have
\begin{equation}
\label{eqlemsupermodularpropvAB+i=1pvABi}
v(A \cup B) + \sum_{i=1}^{p} v(A \cap B_{i}) \geq v(A) + \sum_{i=1}^{p} v(B_{i}).
\end{equation}
\end{lemma}

\section{Inheritance of $\mathcal{F}$-convexity}
\label{SectionInheritanceOfF-Convexity}
Let
$G = (N,E,w)$ be a weighted graph
and let $\mathcal{F}$ be the family of connected subsets of $N$.
We recall in this section necessary and sufficient conditions on the weight vector $w$
established by~\cite{Skoda2017}
for inheritance of $\mathcal{F}$-convexity
from the original game $(N,v)$ 
to the $\mathcal{P}_{\min}$-restricted game $(N,\overline{v})$.
We denote by $w_{k}$ or $w_{ij}$
the weight of an edge $e_{k} = \lbrace i,j \rbrace$ in $E$.

A star $S_k$ corresponds to a tree
with one internal vertex and $k$ leaves.
We consider a star $S_{3}$
with vertices ${1, 2, 3, 4}$ and edges $e_{1} = \lbrace 1, 2 \rbrace$,
$e_{2} = \lbrace 1, 3 \rbrace$ and $e_{3} = \lbrace 1, 4 \rbrace$.
Let us note that the edges $\lbrace 2, 3 \rbrace$, $\lbrace 3, 4 \rbrace$,
or $\lbrace 2, 4 \rbrace$ may exist in $G$.
\vspace{\baselineskip}

\begin{mdframed}
\textbf{Star Condition.}
\it
For every star of type $S_{3}$ in $G$,
the edge-weights satisfy
\[
w_{1} \leq w_{2} = w_{3},
\]
after renumbering the edges if necessary.
\end{mdframed}

\vspace{\baselineskip}

\begin{mdframed}
\textbf{Path Condition.}
\it
For every path $\gamma = \lbrace1, e_{1}, 2, e_{2}, 3, \ldots, m, e_{m}, m+1 \rbrace$ in $G$
and for all $i,j,k$ with $1 \leq i < j < k \leq m$,
the edge-weights satisfy
\[
w_{j} \leq \max (w_{i}, w_{k}).
\]
\end{mdframed}

\vspace{\baselineskip}

For a given cycle
$C = \lbrace 1, e_{1}, 2, e_{2}, \ldots, m, e_{m}, 1 \rbrace$
with $m \geq 3$,
we denote by $E(C)$ the set of edges
$\lbrace e_{1}, e_{2}, \ldots, e_{m} \rbrace$ of~$C$
and by $\hat{E}(C)$ the set composed of $E(C)$
and of the chords of $C$ in $G$.\\

\begin{mdframed}
\textbf{Cycle Condition.}
\it
For every cycle $C = \lbrace 1, e_{1}, 2, e_{2}, \ldots, m, e_{m}, 1 \rbrace$ in $G$ with $m \geq 3$,
the edge-weights satisfy
\[
w_{1} \leq w_{2} \leq w_{3} = \cdots = w_{m}= \hat{M},
\]
after renumbering the edges if necessary,
where $\hat{M} = \max_{e \in \hat{E}(C)} w(e)$.
Moreover,
$w(e) = w_{2}$ for all chord incident to $2$,
and $w(e) = \hat{M}$ for all $e \in \hat{E}(C)$
non-incident to~$2$.
\end{mdframed}

\vspace{\baselineskip}

For a given cycle $C$,
an edge $e$ in $\hat{E}(C)$ is a \emph{maximum weight edge} of~$C$
if $w(e) = \max_{e \in \hat{E}(C)} w(e)$.
Otherwise,
$e$ is a non-maximum weight edge of~$C$.
Moreover,
we call maximum
(resp. non-maximum)
weight chord of $C$
any maximum
(resp. non-maximum)
weight edge in $\hat{E}(C) \setminus E(C)$.

A \textit{pan graph} is a connected graph corresponding to the union of a cycle
and a path.
\vspace{\baselineskip}

\begin{mdframed}
\textbf{Pan Condition.}
\it
For every subgraph of $G$ corresponding to the union of a cycle
$C = \lbrace 1, e_{1}, 2, e_{2}, \ldots, e_{m}, 1 \rbrace$
with $m \geq 3$,
and a path $P$ such that there is an edge $e$ in $P$
with $w(e) \leq \min_{1 \leq k \leq m} w_{k}$
and $|V(C) \cap V(P)| = 1$,
the edge-weights satisfy
\begin{enumerate}[(a)]
\item
\label{eqPanConditionEither}
either $w_{1} = w_{2} = w_{3} = \cdots = w_{m}= \hat{M}$,
\item
\label{eqPanConditionOr}
or $w_{1} = w_{2} < w_{3} = \cdots = w_{m} = \hat{M}$,
\end{enumerate}
\noindent
where $\hat{M} = \max_{e \in \hat{E}(C)} w(e)$.
If Condition~(\ref{eqPanConditionOr}) is satisfied
then $V(C) \cap V(P) = \lbrace 2 \rbrace$,
and if moreover $w(e) < w_{1}$
then $\lbrace 1, 3 \rbrace$
is a maximum weight chord of $C$.
\end{mdframed}

\vspace{\baselineskip}

Two cycles are said \emph{adjacent}
if they share at least one common edge.

\vspace{\baselineskip}

\begin{mdframed}
\textbf{Adjacent Cycles Condition.}
\it
For all pairs $\lbrace C, C' \rbrace$ of adjacent cycles in $G$ such that
\begin{enumerate}[(a)]
\item
\label{enumPropV(C)-V(C)^{'}nonempty}
$V(C) \setminus V(C') \not= \emptyset$ and $V(C') \setminus V(C) \not= \emptyset$,
\item
\label{enumPropCatmost1non-maxweightchord}
$C$ has at most one non-maximum weight chord,
\item
\label{enumPropCC^{'}nomaxweightchord}
$C$ and $C'$ have no maximum weight chord,
\item
$C$ and $C'$ have no common chord,
\end{enumerate}
\noindent
$C$ and $C'$ cannot have two common non-maximum weight edges.
Moreover,
$C$ and $C'$ have a unique common non-maximum weight edge $e_{1}$
if and only if 
there are non-maximum weight edges
$e_{2} \in E(C) \setminus E(C')$
and $e_{2}' \in E(C') \setminus E(C)$
such that $e_{1}, e_{2}, e_{2}'$ are adjacent and
\begin{itemize}
\item
$w_{1} = w_{2} = w_{2}'$ if $|E(C)| \geq 4$ and $|E(C')| \geq 4$.
\item
$w_{1} = w_{2} \geq w_{2}'$
or $w_{1} = w_{2}' \geq w_{2}$
if $|E(C)| = 3$ or $|E(C')| = 3$.
\end{itemize}
\end{mdframed}

\vspace{\baselineskip}

The following characterization of inheritance of $\mathcal{F}$-convexity
was established by~\cite{Skoda2017}
for the correspondence $\mathcal{P}_{\min}$.
This result is also valid for the correspondence $\tilde{\mathcal{P}}_{\min}$
as the $\mathcal{P}_{\min}$-restricted game coincides with
the $\tilde{\mathcal{P}}_{\min}$-restricted game on connected subsets. 

\begin{theorem}
\label{thPathBranchCyclePanAdjacentCond}
Let $\mathcal{F}$ be the family of connected subsets of $N$.
The $\mathcal{P}_{\min}$-restricted game $(N, \overline{v})$
(resp. the $\tilde{\mathcal{P}}_{\min}$-restricted game $(N, \hat{v})$) is $\mathcal{F}$-convex
for every superadditive and $\mathcal{F}$-convex game $(N,v)$
if and only if
the Star, Path, Cycle, Pan, and Adjacent cycles conditions are satisfied.
\end{theorem}

We finally recall two lemmas proved by~\cite{Skoda2017}
valid when some of the previous necessary conditions are satisfied.
The first one gives simple conditions ensuring that
$\mathcal{P}_{\min}(A)$ is induced by $\mathcal{P}_{\min}(B)$ for any subsets $A \subseteq B \subseteq N$.
The second one gives
restrictions on the minimum edge-weights of subsets.

We say that an edge $e \in E$ is connected to a subset $A \subseteq N$,
if there is a path in $G$ joining $e$ to $A$.
\begin{lemma}
\label{lemPminvB-vB>=vA-vA2}
Let $\mathcal{F}$ be the family of connected subsets of $N$.
Let us consider $A, B \in \mathcal{F}$
with $A \subseteq B \subseteq N$, $|A| \geq 2$, and $\sigma(A) = \sigma(B)$.
Let us assume that
\begin{enumerate}
\item
The Pan condition is satisfied.
\item
$G_{B} = (B, E(B))$ is cycle-free
or there exists an edge $e \in E$ connected to $B$ with $w(e) < \sigma(B)$.
\end{enumerate}
Then
\begin{enumerate}
\item
$\mathcal{P}_{\min}(A) = \mathcal{P}_{\min}(B)_{|A}$.
\item
For every $\mathcal{F}$-convex game $(N,v)$, we have
\begin{equation}
\label{eqlemvB-vB>=vA-vA2}
v(B) - \overline{v}(B) \geq v(A) - \overline{v}(A).
\end{equation}
\end{enumerate}
\end{lemma}

\begin{lemma}
\label{lemminsAi>=sA>=sBorsA=sB>sAi}
Let $\mathcal{F}$ be the family of connected subsets of $N$.
Let us assume that the Path and Star conditions are satisfied.
For all $i \in N$,
for all $A \subseteq B \subseteq N \setminus \lbrace i \rbrace$
with $A$ and $B$ in $\mathcal{F}$,
$|A| \geq 2$,
and $E(A,i) \not= \emptyset$,
we have
\begin{enumerate}
\item
either $\sigma(A,i) \geq \sigma (A) \geq \sigma (B)$,
\item
or $\sigma (A) = \sigma(B) > \sigma (A,i) $,
\end{enumerate}
where $\sigma(A,i) = \min_{e \in E(A,i)} w(e)$.
\end{lemma}

\section{Inheritance of convexity for the $\tilde{\mathcal{P}}_{\min}$-restricted game}
\label{SubsectionInheritanceOfConvexityForTheTildePmin-RestrictedGame}

Throughout this section,
$\mathcal{F}$ is the family of connected subsets.

We get in this section
a characterization of inheritance of convexity
for the correspondence $\tilde{\mathcal{P}}_{\min}$.
Of course,
as the class of convex games is contained
in the class of $\mathcal{F}$-convex games,
the necessary conditions
for inheritance of $\mathcal{F}$-convexity
recalled in Section~\ref{SectionInheritanceOfF-Convexity}
are also necessary
for inheritance of convexity.
But these conditions are not sufficient
and we establish three new necessary conditions.
\vspace{\baselineskip}
\begin{mdframed}
\textbf{Reinforced Cycle Condition.}
\it
For every cycle
$C_{m} = \lbrace 1, e_{1}, 2, e_{2}, \ldots,\\ m, e_{m}, 1\rbrace$
with $m \geq 4$
and $w_{3} = \ldots = w_{m} = \hat{M} = \max_{e \in \hat{E}(C_{m})} w(e)$:
\begin{enumerate}
\item
\label{itemReinforcedCycleCondition1}
If $w_{1} < w_{3}$
(resp. $w_{2} < w_{3}$),
then any edge incident to $j$ with $4 \leq j \leq m-1$
(resp. $5 \leq j \leq m$)
is linked to $e_1$
(resp. $e_2$)
by an edge.
\item
\label{itemReinforcedCycleCondition2a}
If $\max(w_1, w_2) < w_3$,
then any edge incident to $j$ with $4 \leq j \leq m$ is linked to $2$ by an edge.
\end{enumerate}
\end{mdframed}

\begin{proposition}
If for all $\emptyset \not= S \subseteq N$
the $\tilde{\mathcal{P}}_{\min}$-restricted game $(N, \overline{u_{S}})$ is convex,
then the Reinforced Cycle Condition is satisfied.
\end{proposition}

\begin{proof}
Let us assume $w_1 < w_3$
(resp. $\max(w_1,w_2) < w_3$).
Let us note that the proof with $w_2 < w_3$ is similar.
Let us consider $e = \lbrace j, j' \rbrace$
with $4 \leq j \leq m$.
If $e \in E(C_m)$,
then $w(e) = \hat{M}$.
Otherwise,
the Star condition applied to $\lbrace e, e_{j-1}, e_{j} \rbrace$
implies $w(e) \leq \hat{M}$.
The Path condition applied
to $\lbrace 1, e_{1}, 2, e_{2}, \ldots e_{j-1}, j, e \rbrace$
implies $\hat{M} = w_{j-1} \leq \max (w_{1}, w(e))$.
As $w_{1} < \hat{M}$,
we get $w(e) = \hat{M}$.  
Let us assume $j=4$
(the other cases are similar)
and $\lbrace 1, 2, j, j' \rbrace$
(resp. $\lbrace 2, j, j' \rbrace$)
non-connected.
Let us consider
$A = \lbrace 1, 2, 3, j, j' \rbrace$
(resp. $A = \lbrace 2, 3, j, j' \rbrace$)
and $B= (V(C_{m}) \setminus \lbrace 3 \rbrace) \cup \lbrace j'\rbrace$
as represented in Figure~\ref{fig2ProofReinforcedCycleCondition}
(resp.  Figure~\ref{fig2ProofReinforcedCycleCondition-2a})
with $m=6$ and $j' \not= 5$.
\begin{figure}[!h]
\centering
\subfloat[$w_1 < w_3$.]{
\begin{pspicture}(-.5,-.3)(1,2)
\tiny
\begin{psmatrix}[mnode=circle,colsep=0.4,rowsep=0.2]
 	  & {$3$}	 	& {$4$}  & {$j'$}\\
{$2$}	&	&		&{$5$}\\
 	  & {$1$}& 	  {$6$}
\psset{arrows=-, shortput=nab,labelsep={0.08}}
\tiny
\ncline{2,1}{3,2}_{$w_{1}$}
\ncline{2,1}{1,2}^{$w_{2}$}
\ncline{3,2}{3,3}_{$\hat{M}$}
\ncline{3,3}{2,4}_{$\hat{M}$}
\ncline{1,2}{1,3}^{$\hat{M}$}
\ncline{1,3}{2,4}_{$\hat{M}$}
\ncline{1,3}{1,4}^{$w(e)$}
\normalsize
\uput[0](.9,0){\textcolor{blue}{$B$}}
\pspolygon[framearc=1,linecolor=blue,linearc=.4,linewidth=.02]
(-2.4,.6)(-2.2,1.5)(-1.2,.8)(-.5,.5)(-.4,1)(-.4,2)(1.4,2)(1.4,.6)(.1,-.35)(-1.2,-.35)
\uput[0](-2,1.8){\textcolor{red}{$A$}}
\pspolygon[framearc=1,linecolor=red,linearc=.4,linewidth=.02]
(-2.4,.8)(-1.1,1.9)(1.3,1.9)(1.3,1.1)(-1,1.1)(-1.1,.7)(-.4,.2)(-.9,-.45)
\end{psmatrix}
\end{pspicture}
\label{fig2ProofReinforcedCycleCondition}
}
\subfloat[$\max(w_1,w_2)<w_3$.]{
\begin{pspicture}(-.5,-.3)(1,2)
\tiny
\begin{psmatrix}[mnode=circle,colsep=0.4,rowsep=0.2]
 	  & {$3$}	 	& {$4$}  & {$j'$}\\
{$2$}	&	&		&{$5$}\\
 	  & {$1$}& 	  {$6$}
\psset{arrows=-, shortput=nab,labelsep={0.08}}
\tiny
\ncline{2,1}{3,2}_{$w_{1}$}
\ncline{2,1}{1,2}^{$w_{2}$}
\ncline{3,2}{3,3}_{$\hat{M}$}
\ncline{3,3}{2,4}_{$\hat{M}$}
\ncline{1,2}{1,3}^{$\hat{M}$}
\ncline{1,3}{2,4}_{$\hat{M}$}
\ncline{1,3}{1,4}^{$w(e)$}
\normalsize
\uput[0](.9,0){\textcolor{blue}{$B$}}
\pspolygon[framearc=1,linecolor=blue,linearc=.4,linewidth=.02]
(-2.4,.6)(-2.2,1.5)(-1.2,.8)(-.5,.5)(-.4,1)(-.4,2)(1.4,2)(1.4,.6)(.1,-.35)(-1.2,-.35)
\uput[0](-2,1.8){\textcolor{red}{$A$}}
\pspolygon[framearc=1,linecolor=red,linearc=.4,linewidth=.02]
(-2.3,.9)(-1.1,1.9)(1.3,1.9)(1.3,1.1)(-.9,1.1)(-1.9,0)
\end{psmatrix}
\end{pspicture}
\label{fig2ProofReinforcedCycleCondition-2a}
}
\caption{$e$ incident to $4$.}
\label{fig2ProofReinforcedCycleCondition-2}
\end{figure}
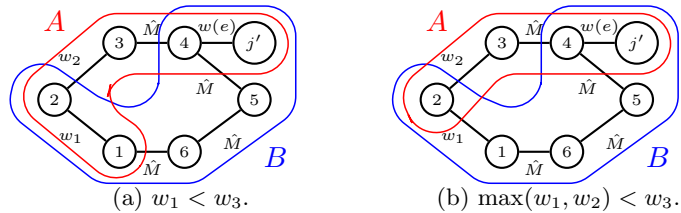
As $A \in \mathcal{F}$ and $B \in \mathcal{F}$,
we have $\tilde{\mathcal{P}}_{\min}(A) = \mathcal{P}_{\min}(A)$
and $\tilde{\mathcal{P}}_{\min}(B) = \mathcal{P}_{\min}(B)$.
As $w_1 < \hat{M}$
(resp.  $\max(w_1,w_2) < \hat{M}$),
there are components $A'$
in $\mathcal{P}_{\min}(A)$
and  $B'$ in $\mathcal{P}_{\min}(B)$
both containing $\lbrace j, j' \rbrace$.
By Theorem~\ref{thLGNuPNFNFSNuSSNNuSFABFABF} (applied with $\mathcal{F}=2^{N} \setminus \lbrace \emptyset \rbrace$),
$A' \cap B'$ is a component of $\tilde{\mathcal{P}}_{\min}(A \cap B)$
containing $\lbrace j, j' \rbrace$.
But we have $\tilde{\mathcal{P}}_{\min}(A \cap B)=
\lbrace \lbrace 1 \rbrace, \lbrace 2 \rbrace, \lbrace j \rbrace, \lbrace j' \rbrace \rbrace$
(resp. $\tilde{\mathcal{P}}_{\min}(A \cap B)=
\lbrace \lbrace 2 \rbrace, \lbrace j \rbrace, \lbrace j' \rbrace \rbrace$),
a contradiction.
\end{proof}

\begin{mdframed}
\textbf{Reinforced Pan Condition}
\it
For all connected subgraphs corresponding to the union of a cycle
$C_m= \lbrace 1, e_{1}, 2,  e_{2}, 3, \ldots, m, e_{m}, 1 \rbrace$
with $m \geq 4$,
satisfying $w_1 \leq w_2 \leq w_3 = \ldots = w_m = \hat{M} = \max_{e \in \hat{E}(C)} w(e)$,
and a path $P$ containing an edge $e$ with $w(e) < \hat{M}$
and $V(C_m) \cap V(P) = \lbrace 2 \rbrace$:
\begin{enumerate}[(a)]
\item
\label{itemReinforcedPanConditionb}
If $w(e) < w_1$,
then any vertex $j$ with $4 \leq j \leq m$
is linked to $P$ by an edge in $E$.
\item
\label{itemReinforcedPanConditionc}
If $w(e) < w_1 < \hat{M}$,
then ($w_1 = w_2$ and) any vertex $j$ with $4 \leq j \leq m$
is linked to $2$ by an edge in $E$.
\end{enumerate}
\end{mdframed}

\begin{proposition}
If for all $\emptyset \not= S \subseteq N$
the $\tilde{\mathcal{P}}_{\min}$-restricted game $(N, \overline{u_{S}})$ is convex,
then the Reinforced Pan Condition is satisfied.
\end{proposition}

\begin{proof}
We first prove that (\ref{itemReinforcedPanConditionb}) is satisfied.
Let us assume $j = 4$ (the other cases are similar).
Let us assume $V (P) \cup \lbrace 4 \rbrace$ non-connected.
Let us consider $A = V(P) \cup \lbrace 3, 4 \rbrace$
and $B = V(P) \cup (V(C_m)\setminus \lbrace 3 \rbrace)$,
as represented in Figure~\ref{figReinforcedCycleConditionConstantCycleb} with $m = 5$.
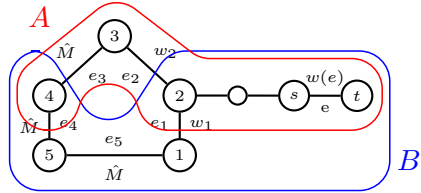
\begin{figure}[!h]
\centering
\begin{pspicture}(-.6,-.4)(0,2.2)
\tiny
\begin{psmatrix}[mnode=circle,colsep=0.4,rowsep=0.3]
& {$3$}\\
{$4$}	& & {$2$} & {} & {$s$} & {$t$} \\
{$5$}	&	& {$1$}
\psset{arrows=-, shortput=nab,labelsep={0.1}}
\tiny
\ncline{2,1}{1,2}^{$\hat{M}$}_{$e_3$}
\ncline{2,1}{3,1}_{$\hat{M}$}^{$e_4$}
\ncline{3,1}{3,3}_{$\hat{M}$}^{$e_5$}
\ncline{1,2}{2,3}^{$w_2$}_{$e_2$}
\ncline{2,3}{3,3}^{$w_1$}_{$e_1$}
\ncline{2,3}{2,4}
\ncline{2,4}{2,5}
\ncline{2,5}{2,6}^{$w(e)$}_{e}
\normalsize
\uput[0](2.7,0){\textcolor{blue}{$B$}}
\pspolygon[framearc=1,linecolor=blue,linearc=.4,linewidth=.02]
(-2.3,-.4)(-2.3,1.45)(-1.8,1.45)(-1,.2)(-.2,1.45)(2.7,1.45)(2.7,-.4)(-1.1,-.4)
\uput[0](-2.15,1.9){\textcolor{red}{$A$}}
\pspolygon[framearc=1,linecolor=red,linearc=.4,linewidth=.02]
(-2.2,.4)(-2.2,1.15)(-.9,2.2)(.2,1.35)(2.6,1.35)(2.6,.4)(-.5,.4)(-1,1.8)(-1.5,.4)
\end{psmatrix}
\end{pspicture}
\caption{$w(e) < w_1 \leq w_2 \leq \hat{M}$.}
\label{figReinforcedCycleConditionConstantCycleb}
\end{figure}
As $A \in \mathcal{F}$ and $B \in \mathcal{F}$,
we have $\tilde{\mathcal{P}}_{\min}(A) = \mathcal{P}_{\min}(A)$
and $\tilde{\mathcal{P}}_{\min}(B) = \mathcal{P}_{\min}(B)$.
As $w(e) < w_1 $,
there are components $A'$ in $\mathcal{P}_{\min}(A)$
and $B'$ in $\mathcal{P}_{\min}(B)$ containing $\lbrace 2, 4 \rbrace$.
By Theorem~\ref{thLGNuPNFNFSNuSSNNuSFABFABF}
(applied with $\mathcal{F}=2^{N} \setminus \lbrace \emptyset \rbrace$),
$A' \cap B'$ is a component of $\tilde{\mathcal{P}}_{\min}(A \cap B)$
containing $\lbrace 2, 4 \rbrace$.
But $\tilde{\mathcal{P}}_{\min}(A \cap B)=\lbrace \lbrace 4 \rbrace, \mathcal{P}_{\min}(V(P)) \rbrace$,
a contradiction.

We now prove that (\ref{itemReinforcedPanConditionc}) is satisfied.
As $w(e) < w_1$,
the Pan condition implies
\begin{equation}
w_1 = w_2 < w_3 = \cdots = w_m = \hat{M}.
\end{equation}
Let us assume $j = 4$ (the other cases are similar)
and $\lbrace 2, 4 \rbrace \notin E$.
By~(\ref{itemReinforcedPanConditionb}),
there exists at least one edge linking $4$ to $P$.
Let us consider an edge $e'_4 = \lbrace 4, j \rbrace$ with $j \in V(P) \setminus \lbrace 2 \rbrace$. 
Let $t$ be the end-vertex of $P$ different from $2$
and let $e'$ be the edge of $P$ incident to $2$
($e'$ may coincide with $e$).
By the Star condition,
we have $w(e') \leq w_1 = w_2$.
Let $C'$ be the cycle formed by $\lbrace 2,  e_2, 3, e_3, 4, e'_4, j \rbrace \cup P_{2,j}$.
As $w(e') \leq w_2 < \hat{M} = w_3$,
the Cycle condition implies
\begin{equation}
\label{eqw(e)=HatM,ForalleInE(C')Setminuse',e2}
w(e) = \hat{M}, \; \forall e \in E(C') \setminus \lbrace e', e_2 \rbrace.
\end{equation}

Let us first assume $e' \not= e$.
If $w(e') < w_1 = w_2$,
then we can replace $e$ by~$e'$.
Therefore,
we can assume $w(e')=w_1 = w_2$.
Let us assume $j = t$ as represented in Figure~\ref{figReinforcedCycleConditionConstantCycle2a}.
Then,
$e$ belongs to $C'$ and has weight $w(e) < \hat{M}$,
contradicting (\ref{eqw(e)=HatM,ForalleInE(C')Setminuse',e2}). 
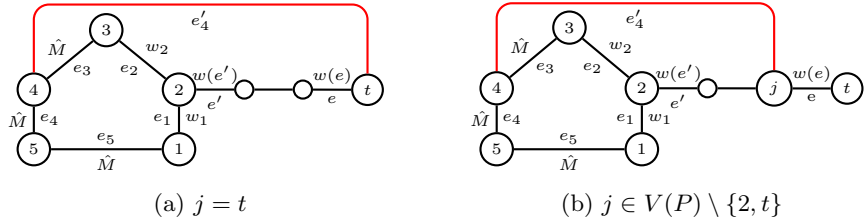
\begin{figure}[!h]
\centering
\subfloat[$j=t$]{
\begin{pspicture}(-.5,-.4)(.5,2)
\tiny
\begin{psmatrix}[mnode=circle,colsep=0.5,rowsep=0.3]
& {$3$}\\
{$4$}	& & {$2$} & {} & {} & {$t$} \\
{$5$}	&	& {$1$}
\psset{arrows=-, shortput=nab,labelsep={0.05}}
\tiny
\ncline{2,1}{1,2}^{$\hat{M}$}_{$e_3$}
\ncline{2,1}{3,1}_{$\hat{M}$}^{$e_4$}
\ncline{3,1}{3,3}_{$\hat{M}$}^{$e_5$}
\ncline{1,2}{2,3}^{$w_2$}_{$e_2$}
\ncline{2,3}{3,3}^{$w_1$}_{$e_1$}
\ncline{2,3}{2,4}_{$e'$}^{$w(e')$}
\ncline{2,4}{2,5}
\ncline{2,5}{2,6}^{$w(e)$}_{$e$}
\ncangles[angleA=-270,angleB=90,armA=.9,armB=.7,linearc=.2,linecolor=red]{2,1}{2,6}_{$e'_4$}
\end{psmatrix}
\end{pspicture}
\label{figReinforcedCycleConditionConstantCycle2a}
}
\subfloat[$j \in V(P) \setminus \lbrace 2, t \rbrace$]{
\begin{pspicture}(-.5,-.4)(.5,2)
\tiny
\begin{psmatrix}[mnode=circle,colsep=0.5,rowsep=0.3]
& {$3$}\\
{$4$}	& & {$2$} & {} & {$j$} & {$t$} \\
{$5$}	&	& {$1$}
\psset{arrows=-, shortput=nab,labelsep={0.05}}
\tiny
\ncline{2,1}{1,2}^{$\hat{M}$}_{$e_3$}
\ncline{2,1}{3,1}_{$\hat{M}$}^{$e_4$}
\ncline{3,1}{3,3}_{$\hat{M}$}^{$e_5$}
\ncline{1,2}{2,3}^{$w_2$}_{$e_2$}
\ncline{2,3}{3,3}^{$w_1$}_{$e_1$}
\ncline{2,3}{2,4}_{$e'$}^{$w(e')$}
\ncline{2,4}{2,5}
\ncline{2,5}{2,6}^{$w(e)$}_{e}
\ncangles[angleA=-270,angleB=90,armA=.9,armB=.1,linearc=.2,linecolor=red]{2,1}{2,5}_{$e'_4$}
\end{psmatrix}
\end{pspicture}
\label{figReinforcedCycleConditionConstantCycle2b}
}
\caption{$w(e) < w(e') = w_1 = w_2 < \hat{M}$.}
\label{figReinforcedCycleConditionConstantCycle2}
\end{figure}
Let us now assume $j \in V(P) \setminus \lbrace 2, t \rbrace$,
as represented in Figure~\ref{figReinforcedCycleConditionConstantCycle2b}.
As $w(e) < w(e') = w_2 < \hat{M}$,
the Pan condition applied to the pan formed by $C'$ and $P_{j,t}$ implies $j=2$,
a contradiction.

Let us now assume $e' = e$.
We necessarily have $j=t$.
Let us consider
$A = \lbrace 2, 3, 4, t \rbrace$
and $B= (V(C_{m}) \setminus \lbrace 3 \rbrace) \cup \lbrace t \rbrace$
as represented in Figure~\ref{figReinforcedCycleConditionConstantCycle3-bis}
with $m=5$.
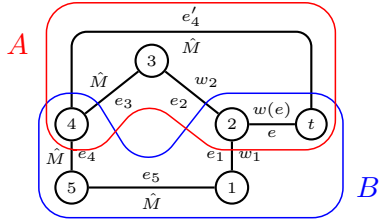
\begin{figure}[!h]
\centering
\begin{pspicture}(-.5,-.1)(.5,2.5)
\tiny
\begin{psmatrix}[mnode=circle,colsep=0.6,rowsep=0.35]
& {$3$}\\
{$4$}	& & {$2$} & {$t$}\\
{$5$}	&	& {$1$}
\psset{arrows=-, shortput=nab,labelsep={0.05}}
\tiny
\ncline{2,1}{1,2}^{$\hat{M}$}_{$e_3$}
\ncline{2,1}{3,1}_{$\hat{M}$}^{$e_4$}
\ncline{3,1}{3,3}_{$\hat{M}$}^{$e_5$}
\ncline{1,2}{2,3}^{$w_2$}_{$e_2$}
\ncline{2,3}{3,3}^{$w_1$}_{$e_1$}
\ncline{2,3}{2,4}_{$e$}^{$w(e)$}
\ncangles[angleA=-270,angleB=90,armA=1,armB=.7,linearc=.2]{2,1}{2,4}^{$e'_4$}_{$\hat{M}$}
\end{psmatrix}
\normalsize
\uput[0](.2,0){\textcolor{blue}{$B$}}
\pspolygon[framearc=1,linecolor=blue,linearc=.4,linewidth=.02]
(-3.8,-.4)(-3.8,1.25)(-3.1,1.25)(-2.4,.1)(-1.5,1.25)(.2,1.25)(.2,-.4)(-1.1,-.4)
\uput[0](-4.4,1.9){\textcolor{red}{$A$}}
\pspolygon[framearc=1,linecolor=red,linearc=.4,linewidth=.02]
(-3.7,.5)(-3.7,2.45)(.1,2.45)(.1,.5)(-1.5,.5)(-2.4,1.2)(-3,.5)
\end{pspicture}
\caption{$w(e) < w_1 = w_2 < \hat{M}$.}
\label{figReinforcedCycleConditionConstantCycle3-bis}
\end{figure}
As $A \in \mathcal{F}$ and $B \in \mathcal{F}$,
we have $\tilde{\mathcal{P}}_{\min}(A) = \mathcal{P}_{\min}(A)$
and $\tilde{\mathcal{P}}_{\min}(B) = \mathcal{P}_{\min}(B)$.
As $w(e) < w_1 = w_2$,
there are components $A'$ in $\mathcal{P}_{\min}(A)$
and $B'$ in $\mathcal{P}_{\min}(B)$ containing $\lbrace 2, 4 \rbrace$.
By Theorem~\ref{thLGNuPNFNFSNuSSNNuSFABFABF}
(applied with $\mathcal{F}=2^{N} \setminus \lbrace \emptyset \rbrace$),
$A' \cap B'$ is a component of $\tilde{\mathcal{P}}_{\min}(A \cap B)$
containing $\lbrace 2, 4 \rbrace$.
But $\tilde{\mathcal{P}}_{\min}(A \cap B)=
\lbrace \lbrace 2 \rbrace, \lbrace 4, t \rbrace \rbrace$,
a contradiction.
\end{proof}

\begin{mdframed}
\textbf{Reinforced Adjacent Cycles Condition}
\it
For all pairs $\lbrace C,C' \rbrace$ of adjacent cycles in $G$ such that
one of the following conditions is satisfied
\begin{enumerate}
\item
\label{enumProp|E(C)|=|E(C')|=4}
$|E(C)| = |E(C')| = 4$ and $C$ and $C'$ have two common non-maximum weight edges.
\item
\label{enumProp|E(C)|=|E(C')|=5}
\begin{enumerate}
\item
$|E(C)| = |E(C')| = 5$,
$|E(C) \cap E(C')| = 3$,
and $C$ and $C'$ have only one non-maximum weight-edge
which is common to $C$ and $C'$.
\item
Setting $C = \lbrace 1, e_1, 2, e_2, 3, e_3, 4, e_4, 5, e_5, 1 \rbrace$,
we have $C' = \lbrace 1, e_1, 2, e_2, 3, e'_3, 4', e'_4, 5, e_5, 1 \rbrace$ with $4' \not= 4$,
$4$ and $4'$ are not both linked to $1$ or to $2$,
and $e_1$ is the unique non-maximum weight edge common to $C$ and $C'$.
\end{enumerate}
\end{enumerate}
\noindent
There exists an edge linking $V(C) \setminus V(C')$ to $V(C') \setminus V(C)$.
\end{mdframed}

\begin{proposition}
If for all $\emptyset \not= S \subseteq N$
the $\tilde{\mathcal{P}}_{\min}$-restricted game $(N, \overline{u_{S}})$ is convex,
then the Reinforced Adjacent Cycles Condition is satisfied.
\end{proposition}

\begin{proof}
Let us first assume Condition~\ref{enumProp|E(C)|=|E(C')|=4} satisfied.
Let us consider $C = \lbrace 1, e_1, 2, e_2, 3, e_3, 4, e_4, 1 \rbrace$
and $C' = \lbrace 1, e_1, 2, e_2, 3, e'_3, 4', e'_4, 1  \rbrace$
where $e_1$ and $e_2$ are non-maximum weight edges common to $C$ and $C'$.
By the Cycle condition,
we have
\begin{equation}
\max(w_1, w_2) < w_3 = w_4 = w(e'_3) = w(e'_4) = \hat{M},
\end{equation}
where $\hat{M} = \max_{e \in \hat{E}(C)} w(e) = \max_{e \in \hat{E}(C')} w(e)$.
Moreover,
any chord of $C$ or $C'$ incident to $2$
has a weight equal to $\max(w_1,w_2)$.
Let us consider
$A = \lbrace 1, 2, 4, 4' \rbrace$
and $B = \lbrace 2, 3, 4, 4' \rbrace$)
as represented in Figure~\ref{figReinforcedAdjcacentCyclesCondition}.
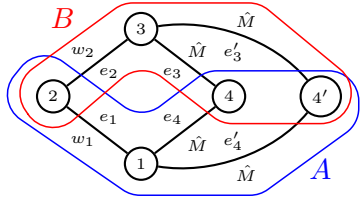
\begin{figure}[!h]
\centering
\begin{pspicture}(-.5,-.2)(.5,2.4)
\tiny
\begin{psmatrix}[mnode=circle,colsep=0.7,rowsep=0.35]
& {$3$}\\
{$2$}	& & {$4$} & {$4'$}\\
& {$1$}
\psset{arrows=-, shortput=nab,labelsep={0.05}}
\tiny
\ncline{2,1}{1,2}^{$w_2$}_{$e_2$}
\ncline{1,2}{2,3}_{$e_3$}^{$\hat{M}$}
\ncline{2,3}{3,2}_{$e_4$}^{$\hat{M}$}
\ncline{2,1}{3,2}_{$w_1$}^{$e_1$}
\ncarc[arcangle=30]{1,2}{2,4}_{$e'_3$}^{$\hat{M}$}
\ncarc[arcangle=-30]{3,2}{2,4}_{$\hat{M}$}^{$e'_4$}
\normalsize
\uput[0](2.1,0){\textcolor{blue}{$A$}}
\pspolygon[framearc=1,linecolor=blue,linearc=.4,linewidth=.02]
(-2,.9)(-1.4,1.65)(-.05,.65)(.8,1.3)(2.8,1.3)(2.8,.6)(1.5,-.35)(-.2,-.35)
\uput[0](-1.3,2){\textcolor{red}{$B$}}
\pspolygon[framearc=1,linecolor=red,linearc=.4,linewidth=.02]
(-1.8,1.1)(-.3,2.2)(1.5,2.2)(2.7,1.2)(2.7,.6)(1,.6)(.1,1.3)(-.2,1.3)(-1.3,.35)
\end{psmatrix}
\end{pspicture}
\caption{$\max(w_1, w_2) < \hat{M}$.}
\label{figReinforcedAdjcacentCyclesCondition}
\end{figure}
Let us assume $\lbrace 4, 4' \rbrace \notin E$.
Let us note that $\lbrace 1, 3 \rbrace$, $\lbrace 2, 4 \rbrace$, and $\lbrace 2, 4' \rbrace$ may exist.
As $A \in \mathcal{F}$ and $B \in \mathcal{F}$,
we have $\tilde{\mathcal{P}}_{\min}(A) = \mathcal{P}_{\min}(A)$
and $\tilde{\mathcal{P}}_{\min}(B) = \mathcal{P}_{\min}(B)$.
As $\max(w_1,w_2) < \hat{M}$,
there are components $A'$
in $\mathcal{P}_{\min}(A)$
and  $B'$ in $\mathcal{P}_{\min}(B)$
both containing $\lbrace 4, 4' \rbrace$.
By Theorem~\ref{thLGNuPNFNFSNuSSNNuSFABFABF} (applied with $\mathcal{F}=2^{N} \setminus \lbrace \emptyset \rbrace$),
$A' \cap B'$ is a component of $\tilde{\mathcal{P}}_{\min}(A \cap B)$
containing $\lbrace 4, 4' \rbrace$.
But we have $\tilde{\mathcal{P}}_{\min}(A \cap B)=
\lbrace \lbrace 2 \rbrace, \lbrace 4 \rbrace, \lbrace 4' \rbrace \rbrace$,
a contradiction.

Let us now assume Condition~\ref{enumProp|E(C)|=|E(C')|=5} satisfied.
Let us consider $C = \lbrace 1, e_1, 2,\\ e_2, 3, e_3, 4, e_4, 5, e_5, 1 \rbrace$
and $C' = \lbrace 1, e_1, 2, e_2, 3, e'_3, 4', e'_4, 5, e_5, 1  \rbrace$
where $e_1$ is the unique non-maximum weight edge common to $C$ and $C'$.
By the Cycle condition,
we have
\begin{equation}
w_1 < w_2 = w_3 = w_4 = w_5 = w(e'_3) = w(e'_4) = \hat{M}.
\end{equation}
Moreover,
any chord of $C$ or $C'$ has weight $\hat{M}$.
Let us consider
$A = \lbrace 1, 2, 4, 4', 5 \rbrace$
and $B = \lbrace 1, 2, 3, 4, 4' \rbrace$
as represented in Figure~\ref{figReinforcedAdjcacentCyclesCondition}.
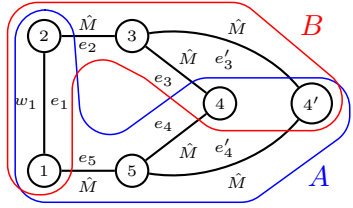
\begin{figure}[!h]
\centering
\begin{pspicture}(-.5,-.2)(.5,2.4)
\tiny
\begin{psmatrix}[mnode=circle,colsep=0.7,rowsep=0.35]
{$2$ }& {$3$}\\
 & & {$4$} & {$4'$}\\
{$1$}& {$5$}
\psset{arrows=-, shortput=nab,labelsep={0.05}}
\tiny
\ncline{1,1}{1,2}^{$\hat{M}$}_{$e_2$}
\ncline{1,2}{2,3}_{$e_3$}^{$\hat{M}$}
\ncline{2,3}{3,2}_{$e_4$}^{$\hat{M}$}
\ncline{1,1}{3,1}_{$w_1$}^{$e_1$}
\ncline{3,1}{3,2}_{$\hat{M}$}^{$e_5$}
\ncarc[arcangle=30]{1,2}{2,4}_{$e'_3$}^{$\hat{M}$}
\ncarc[arcangle=-30]{3,2}{2,4}_{$\hat{M}$}^{$e'_4$}
\normalsize
\uput[0](2.2,0){\textcolor{blue}{$A$}}
\pspolygon[framearc=1,linecolor=blue,linearc=.4,linewidth=.02]
(-1.6,2.2)(-.85,2.2)(-.7,.2)(.8,1.3)(2.8,1.3)(2.8,.6)(1.5,-.35)(-1.6,-.35)
\uput[0](2.1,2){\textcolor{red}{$B$}}
\pspolygon[framearc=1,linecolor=red,linearc=.4,linewidth=.02]
(-1.7,-.25)(-1.7,2.3)(1.4,2.3)(2.7,1.2)(2.7,.6)(1,.6)(.1,1.3)(-.85,1.8)(-.85,-.25)
\end{psmatrix}
\end{pspicture}
\caption{$w_1 < \hat{M}$.}
\label{figReinforcedAdjcacentCyclesCondition}
\end{figure}
Let us assume $\lbrace 4, 4' \rbrace \notin E$.
As $A \in \mathcal{F}$ and $B \in \mathcal{F}$,
we have $\tilde{\mathcal{P}}_{\min}(A) = \mathcal{P}_{\min}(A)$
and $\tilde{\mathcal{P}}_{\min}(B) = \mathcal{P}_{\min}(B)$.
As $w_1 < \hat{M}$,
there are components $A'$
in $\mathcal{P}_{\min}(A)$
and  $B'$ in $\mathcal{P}_{\min}(B)$
both containing $\lbrace 4, 4' \rbrace$.
By Theorem~\ref{thLGNuPNFNFSNuSSNNuSFABFABF} (applied with $\mathcal{F}=2^{N} \setminus \lbrace \emptyset \rbrace$),
$A' \cap B'$ is a component of $\tilde{\mathcal{P}}_{\min}(A \cap B)$
containing $\lbrace 4, 4' \rbrace$.
As $4$ and $4'$ are neither both linked to $1$
nor both linked to $2$,
they necessarily belong to distinct blocks of
$\tilde{\mathcal{P}}_{\min}(A \cap B)$,
a contradiction.
\end{proof}

\begin{theorem}
\label{thADDLG=NEbacfgaluEbawfvhat}
For every convex game $(N, v)$,
the $\tilde{\mathcal{P}}_{\min}$-restricted game $(N, \overline{v})$ is convex
if and only if the Star, Path, Cycle, Pan, Adjacent Cycles,
Reinforced Cycle, Reinforced Pan,
and Reinforced Adjacent Cycles conditions are satisfied.
\end{theorem}

We have already seen that these conditions are necessary.
To prove their sufficiency,
we will need the following propositions and lemmas.

\begin{proposition}
\label{Propositionpartitionin2set1}
Let us assume that the Path, Star, Cycle,
Reinforced cycle and Reinforced pan conditions are satisfied.
Let us consider $i \in N$
and $A \subseteq B \subseteq N \setminus \lbrace i \rbrace$
with $A \cup \lbrace i \rbrace$ and $B$ in $\mathcal{F}$.
Let us assume $A \notin \mathcal{F}$,
and let $A_1$, $A_2$
be two connected components of $A$
with $\sigma(A_1, i) < \sigma(A_2, i)$.
If $\sigma(B)<\sigma(A_2,i)$,
then
\begin{enumerate}
\item
\label{itemPropsIf|A1|=1AndIfSigma(B)<M,Then|A2|=1a}
$|A_2| = 1$.
Moreover,
there exists a unique edge $e_1$ in $\Sigma(A_1,i)$
and setting $e_1= \lbrace i, j_1 \rbrace$ with $j_1 \in A_1$ and $A_2 = \lbrace j_2 \rbrace$,
there exist a vertex $k \in B$,
and a cycle $C_4 = \lbrace e_1, e_2, e_3, e_4 \rbrace$
with $e_2 = \lbrace i, j_2 \rbrace$, $e_3= \lbrace j_2, k \rbrace$, $e_4 = \lbrace k, j_1 \rbrace$
and $w_2 = w_3 = \sigma(A_2,i)$, and $w_4 = \sigma(B)$.
\item
\label{itemPropA1AndA2AreInDistinctBlocksOfPmin(B)}
$\lbrace j_1 \rbrace$ and $A_2$ are in distinct blocks of $\mathcal{P}_{\min}(B)$.
Moreover,
any block of $\mathcal{P}_{\min}(A_1)$ belongs to a block of $\mathcal{P}_{\min}(B)$
different from the one containing $A_2$.
\item
\label{itemPropA1A2AndA3AreAllInDistinctBlocksOfPmin(B)}
If $A_3$ is a third connected component of $A$
and if the Reinforced adjacent cycles condition is satisfied,
then $|A_3|=1$,
$\sigma(A_2,i) = \sigma(A_3,i)$,
and $A_2$ and $A_3$ are in distinct blocks of $\mathcal{P}_{\min}(B)$.
\end{enumerate}
\end{proposition}

\begin{proof}
\textbf{\ref{itemPropsIf|A1|=1AndIfSigma(B)<M,Then|A2|=1a}.}
Let $e_1 = \lbrace i, j_1 \rbrace$
(resp. $e_2 = \lbrace i, j_2 \rbrace$)
be an  edge in $\Sigma(A_1,i)$
(resp. $\Sigma(A_2,i)$).
As $\sigma(A_1,i) < \sigma(A_2,i)$,
$e_1$ is necessarily unique by the Star condition.
As $B \in \mathcal{F}$,
there exists at least one path in $G_B$ linking $j_1$ to $j_2$.

Let us first assume that there is no path in $G_B$ linking $j_1$ to $j_2$
and containing at least one edge in $\Sigma(B)$.
Let $\gamma$ be a shortest path in $G_B$ linking $j_1$ to $j_2$.
$\lbrace e_1 \rbrace \cup \gamma \cup \lbrace e_2 \rbrace$ induces a cycle
$C_m= \lbrace e_1, e_2, \ldots, e_m \rbrace$ with $m \geq 4$.
Let $\tilde{e}_1$ be an edge in $\Sigma(B)$.
As $E(\gamma) \cap \Sigma(B) = \emptyset$,
$\tilde{e}_1$ cannot belong to $E(C_m)$.
Moreover,
$\tilde{e}_1$ cannot be a chord of $C_m$,
otherwise there would be a path linking $j_1$ to $j_2$ 
and containing an edge in $\Sigma(B)$.
Let $P$ be a shortest path in $G_B$ linking $\tilde{e}_1$ to a vertex $j^*$ in $\gamma$
($P$ may be reduced to $j^*$).
We select $\tilde{e}_1$ such that $P$ is as short as possible.
As $w_1 = \sigma(A_1,i) < \sigma(A_2,i) = w_2$,
the Cycle condition implies
\begin{equation}
\label{eqProofu1<u3=M=maxeinE(Cm)u(e)-2-2}
w_{1} < w_{2} \leq w_{3} = \cdots = w_{m-1} = \hat{M} = \max_{e \in \hat{E}(C_{m})} w(e),
\end{equation}
and
\begin{equation}
\label{eqEitherwm<HatMAndw2=HatM,Orwm=HatMAndw2<=HatM}
\begin{array}{cc}
\textrm{either} & w_m < \hat{M} \textrm{ and } w_2 = \hat{M},\\
\textrm{or} & w_m = \hat{M} \textrm{ and } w_2 \leq \hat{M}.
\end{array}
\end{equation}
If $j^* \not= j_1$ as represented in Figure~\ref{figLemmaCyclePanj*Not=j1-a},
then the Path condition implies $w_2 \leq \max (w_1, w(\tilde{e}_1))$.
\begin{figure}[!h]
\centering
\subfloat[$j^* \not= j_1$.]{
\label{figLemmaCyclePanj*Not=j1-a}
\begin{pspicture}(0,-.4)(1.5,2)
\tiny
\begin{psmatrix}[mnode=circle,colsep=0.4,rowsep=0.3]
  & & & {$i$}\\
 & & {$j_1$}	& & {$j_2$} \\
{} & {} &  {$j^*$}	&	{} & {}
\psset{arrows=-, shortput=nab,labelsep={0.05}}
\tiny
\ncline{3,4}{3,5}_{$\hat{M}$}
\ncline{3,3}{3,4}_{$\hat{M}$}
\ncline{2,3}{3,3}_{$w_m$}^{$e_m$}
\ncline{2,3}{1,4}^{$w_1$}_{$e_1$}
\ncline{1,4}{2,5}^{$w_2$}_{$e_2$}
\ncline{2,5}{3,5}^{$\hat{M}$}_{$e_3$}
\ncline{3,1}{3,2}^{$\tilde{e}_1$}
\ncline{3,2}{3,3}
\normalsize
\uput[0](.4,1.6){\textcolor{blue}{$A_{2}$}}
\pspolygon[framearc=1,linestyle=dashed,linecolor=blue,linearc=.3]
(-.4,-.2)(-.4,1.3)(1.1,1.3)(1.1,-.2)
\uput[0](-2.7,1.4){\textcolor{blue}{$A_{1}$}}
\pspolygon[framearc=1,linestyle=dashed,linecolor=blue,linearc=.3]
(-2.2,.6)(-2.2,1.3)(-1.45,1.3)(-1.45,.6)
\end{psmatrix}
\end{pspicture}
}
\subfloat[$j^* = j_1$.]{
\label{figLemmaCyclePanj*Not=j1-b}
\begin{pspicture}(-.5,-.4)(.5,2)
\tiny
\begin{psmatrix}[mnode=circle,colsep=0.4,rowsep=0.3]
   & & {$i$}\\
 {} & {$j_1$}	& & {$j_2$} \\
   &  {}	&	{} & {}
\psset{arrows=-, shortput=nab,labelsep={0.05}}
\tiny
\ncline{3,3}{3,4}_{$\hat{M}$}
\ncline{3,2}{3,3}_{$\hat{M}$}
\ncline{2,2}{3,2}_{$w_m$}^{$e_m$}
\ncline{2,2}{1,3}^{$w_1$}_{$e_1$}
\ncline{1,3}{2,4}^{$w_2$}_{$e_2$}
\ncline{2,4}{3,4}^{$\hat{M}$}_{$e_3$}
\ncline{2,1}{2,2}_{$\tilde{e}_1$}
\normalsize
\uput[0](.4,1.4){\textcolor{blue}{$A_{2}$}}
\pspolygon[framearc=1,linestyle=dashed,linecolor=blue,linearc=.3]
(-.4,-.2)(-.4,1.1)(1.1,1.1)(1.1,-.2)
\uput[0](-2.6,1.3){\textcolor{blue}{$A_{1}$}}
\pspolygon[framearc=1,linestyle=dashed,linecolor=blue,linearc=.3]
(-2.2,.4)(-2.2,1.1)(-1.45,1.1)(-1.45,.4)
\end{psmatrix}
\end{pspicture}
}
\caption{$w_1 < w_2 \leq \hat{M}$ and $w_m \leq \hat{M}$.}
\label{figLemmaCyclePanj*Not=j1}
\end{figure}
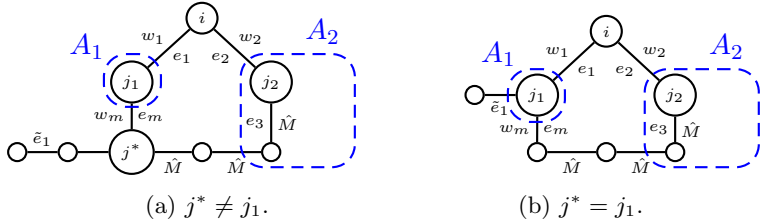
As $w_1 < w_2$ and as $w(\tilde{e}_1) = \sigma(B) < \sigma(A_2,i) = w_2$,
we get a contradiction.
We henceforth assume $j^* = j_1$.

Let us first assume $\tilde{e}_1$ incident to $j_1$
as represented in Figure~\ref{figLemmaCyclePanj*Not=j1-b}.
Any edge in $\gamma$ has weight strictly greater than $\sigma(B)$.
As $w(\tilde{e}_1) = \sigma(B)$,
the Star condition implies
$w(\tilde{e}_1) < w_1 = w_m$.
As $w_1 < \hat{M}$,
we get $w_m < \hat{M}$,
and (\ref{eqEitherwm<HatMAndw2=HatM,Orwm=HatMAndw2<=HatM}) implies $w_2 = \hat{M}$.
As $w(\tilde{e}_1) < w_1 = w_m < \hat{M}$,
the Reinforced pan condition implies that
$j_2$ is linked to $j_1$ by an edge,
a contradiction.

Let us now assume $\tilde{e}_1$ non-incident to $j_1$.
Let $e'_1$ be the edge of $P$ incident to $j_1$
as represented in Figure~\ref{figLemmaCyclePanTildee1Non-IncidentToj1}.
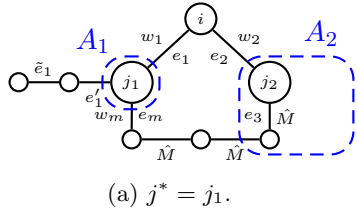
\begin{figure}[!h]
\centering
\subfloat[$j^* = j_1$.]{
\label{figLemmaCyclePanTildee1Non-IncidentToj1-b}
\begin{pspicture}(0,-.4)(.5,2)
\tiny
\begin{psmatrix}[mnode=circle,colsep=0.4,rowsep=0.3]
  & & & {$i$}\\
{} & {} & {$j_1$}	& & {$j_2$} \\
 &  &  {}	&	{} & {}
\psset{arrows=-, shortput=nab,labelsep={0.05}}
\tiny
\ncline{3,4}{3,5}_{$\hat{M}$}
\ncline{3,3}{3,4}_{$\hat{M}$}
\ncline{2,3}{3,3}_{$w_m$}^{$e_m$}
\ncline{2,3}{1,4}^{$w_1$}_{$e_1$}
\ncline{1,4}{2,5}^{$w_2$}_{$e_2$}
\ncline{2,5}{3,5}^{$\hat{M}$}_{$e_3$}
\ncline{2,1}{2,2}^{$\tilde{e}_1$}
\ncline{2,2}{2,3}_{$e'_1$}
\normalsize
\uput[0](.4,1.4){\textcolor{blue}{$A_{2}$}}
\pspolygon[framearc=1,linestyle=dashed,linecolor=blue,linearc=.3]
(-.4,-.2)(-.4,1.1)(1.1,1.1)(1.1,-.2)
\uput[0](-2.6,1.3){\textcolor{blue}{$A_{1}$}}
\pspolygon[framearc=1,linestyle=dashed,linecolor=blue,linearc=.3]
(-2.2,.4)(-2.2,1.1)(-1.45,1.1)(-1.45,.4)
\end{psmatrix}
\end{pspicture}
}
\caption{$w_1 < w_2 \leq \hat{M}$ and $w_m \leq \hat{M}$.}
\label{figLemmaCyclePanTildee1Non-IncidentToj1}
\end{figure}
Let us assume $w(\tilde{e}_1) \geq w_1$.
Then the Path condition implies $w(e'_1) \leq \max (w(\tilde{e}_1),w_1) = w(\tilde{e}_1)$.
As $w(\tilde{e}_1) = \sigma(B)$,
we necessarily have $w(e'_1) = \sigma(B)$.
This contradicts the choice of $\tilde{e}_1$.
Let us assume $w(\tilde{e}_1) < w_1$.
As any edge in $\gamma$ has weight strictly greater than $\sigma(B)$,
we also have $w(\tilde{e}_1) < w_m$.
If $w_m < \hat{M}$ and $w_2 = \hat{M}$,
then by the Reinforced pan condition
$j_2$ is linked to $j_1$ by an edge,
a contradiction.
If $w_m = \hat{M}$ and $w_2 \leq \hat{M}$,
then the Pan condition implies
$w_1 = w_2$ (and $V(C_m) \cap V(P) = \lbrace 2 \rbrace$),
a contradiction.

Let us now assume that there exists a path $\gamma$ in $G_B$ linking $j_1$ to $j_2$
and containing at least one edge in $\Sigma(B)$.
We select $\gamma$ as short as possible.
$\lbrace e_1 \rbrace \cup \gamma \cup \lbrace e_2 \rbrace$ induces a cycle
$C_m= \lbrace e_1, e_2, \ldots, e_m \rbrace$ with $m \geq 4$.
Let $\tilde{e}_1$ be an edge of $\gamma$ with $w(\tilde{e}_1) = \sigma(B)$.
If $\tilde{e}_1$ is non-incident to $j_1$,
then the Path condition implies
$\sigma(A_2, i) = w_2  \leq \max(w_1,w(\tilde{e}_1)) = \max(\sigma(A_1,i), \sigma(B))$,
a contradiction.
If $\tilde{e}_1$ is incident to $j_1$,
then $\tilde{e}_1 = e_m$.
As $w_1 = \sigma(A_1,i) < \sigma(A_2,i) = w_2$
and as $w(\tilde{e}_1) = \sigma(B) < \sigma(A_2,i) = w_2$,
the Cycle condition implies
\begin{equation}
\label{eqProofu1<u3=M=maxeinE(Cm)u(e)-2-2}
\max(w_{1},w_m) < w_{2} = \cdots = w_{m-1} = \hat{M} = \max_{e \in \hat{E}(C_{m})} w(e).
\end{equation}
Let us assume $|A_2| \geq 2$.
Then,
there is at least one edge $\tilde{e}_2$ in $E(A_2)$ incident to $j_2$
($\tilde{e}_2$ may coincide with $e_3$),
as represented in~Figure~\ref{figLemmaCyclePanb-1-a}.
\begin{figure}[!h]
\centering
\centering
\subfloat[$|A_2| \geq 2$.]{
\label{figLemmaCyclePanb-1-a}
\begin{pspicture}(-.5,-.4)(1.5,2)
\tiny
\begin{psmatrix}[mnode=circle,colsep=0.4,rowsep=0.3]
& {$i$}\\
{$j_1$}	& & {$j_2$} & {}\\
 {}	& {}	& {}
\psset{arrows=-, shortput=nab,labelsep={0.05}}
\tiny
\ncline{3,2}{3,3}_{$\hat{M}$}
\ncline{3,1}{3,2}_{$\hat{M}$}
\ncline{2,1}{3,1}_{$w_m$}^{$e_m$}
\ncline{2,1}{1,2}^{$w_1$}_{$e_1$}
\ncline{1,2}{2,3}^{$\hat{M}$}_{$e_2$}
\ncline{2,3}{3,3}^{$\hat{M}$}_{$e_3$}
\ncline{2,3}{2,4}^{$\tilde{e}_2$}
\normalsize
\uput[0](.4,1.5){\textcolor{blue}{$A_{2}$}}
\pspolygon[framearc=1,linestyle=dashed,linecolor=blue,linearc=.3]
(-.4,-.2)(-.4,1.2)(1.1,1.2)(1.1,-.2)
\uput[0](-2.7,1.2){\textcolor{blue}{$A_{1}$}}
\pspolygon[framearc=1,linestyle=dashed,linecolor=blue,linearc=.3]
(-2.2,.4)(-2.2,1.1)(-1.45,1.1)(-1.45,.4)
\end{psmatrix}
\end{pspicture}
}
\subfloat[$|A_2| = 1$.]{
\label{figLemmaCyclePanb-1-b}
\begin{pspicture}(-.5,-.4)(1.5,2)
\tiny
\begin{psmatrix}[mnode=circle,colsep=0.4,rowsep=0.3]
& {$i$}\\
{$j_1$}	& & {$j_2$}\\
 {}	& {}	& {}
\psset{arrows=-, shortput=nab,labelsep={0.05}}
\tiny
\ncline{3,2}{3,3}_{$\hat{M}$}
\ncline{3,1}{3,2}_{$\hat{M}$}
\ncline{2,1}{3,1}_{$w_m$}^{$e_m$}
\ncline{2,1}{1,2}^{$w_1$}_{$e_1$}
\ncline{1,2}{2,3}^{$\hat{M}$}_{$e_2$}
\ncline{2,3}{3,3}^{$\hat{M}$}_{$e_3$}
\ncline{2,1}{3,3}^{$\tilde{e}$}
\normalsize
\uput[0](.3,1.2){\textcolor{blue}{$A_{2}$}}
\pspolygon[framearc=1,linestyle=dashed,linecolor=blue,linearc=.3]
(-.4,.4)(-.4,1.1)(.4,1.1)(.4,.4)
\uput[0](-2.7,1.2){\textcolor{blue}{$A_{1}$}}
\pspolygon[framearc=1,linestyle=dashed,linecolor=blue,linearc=.3]
(-2.2,.4)(-2.2,1.1)(-1.45,1.1)(-1.45,.4)
\end{psmatrix}
\end{pspicture}
}
\caption{$w_1 < \hat{M} = w_2$ and $w_m < \hat{M}$.}
\label{figLemmaCyclePana-2}
\end{figure}
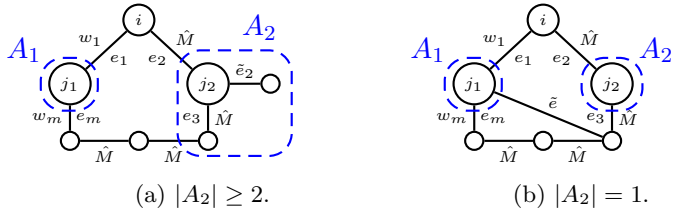
By the Reinforced cycle condition,
$\tilde{e}_2$ is linked to $j_1$ by an edge,
and therefore $A_2$ is linked to $A_1$,
a contradiction.
Let us now assume $|A_2| = 1$.
Then,
by the Reinforced cycle condition,
$e_3$ is linked to $j_1$ by an edge $\tilde{e}$.
If $m \geq 5$ as represented in Figure~\ref{figLemmaCyclePanb-1-b},
then it contradicts the minimality of $\gamma$.
Therefore,
we necessarily have $m=4$.

\textbf{\ref{itemPropA1AndA2AreInDistinctBlocksOfPmin(B)}.}
By Claim~\ref{itemPropsIf|A1|=1AndIfSigma(B)<M,Then|A2|=1a},
we have $|A_2| = 1$.
Moreover,
there exists a unique edge $e_1 \in \Sigma(A_1,i)$
and setting $e_1 = \lbrace i, j_1 \rbrace$ with $j_1 \in A_1$
and $A_2 = \lbrace j_2 \rbrace$,
there exists a vertex $k \in B$
and a cycle $C_4 = \lbrace j_1, e_1, i, e_2, j_2, e_3, k, e_4, j_1 \rbrace$
with $w_2 = w_3 = \sigma(A_2,i)$, and $w_4 = \sigma(B)$.
We set $\sigma(A_2,i) = M$.

Let us first assume that $\lbrace j_1 \rbrace$ and $A_2$ belong to the same block $B_j$ of $\mathcal{P}_{\min}(B)$.
Then,
there exists a path $\gamma$ in $G_{B_j}$ linking $j_1$ to $j_2$
as represented in Figure~\ref{figIf|A1|=1,ThenA1AndA2AreInDistinctBlocksOfPmin(B)}
and such that $w(e) > \sigma(B)$ for every edge $e$ in $\gamma$.
We select $\gamma$ as short as possible.
Let us assume
$k \in \gamma$
(resp. $k \notin \gamma$)
as represented in Figure~\ref{figIf|A1|=1,ThenA1AndA2AreInDistinctBlocksOfPmin(B)a}
(resp. Figure~\ref{figIf|A1|=1,ThenA1AndA2AreInDistinctBlocksOfPmin(B)b}).
\begin{figure}[!h]
\centering
\subfloat[$k \in \gamma$.]{
\begin{pspicture}(-.5,-.7)(.5,2.2)
\tiny
\begin{psmatrix}[mnode=circle,colsep=0.7,rowsep=0.35]
& {$i$}\\
{$j_1$}	& & {$j_2$}\\
{} & {$k$}
\psset{arrows=-, shortput=nab,labelsep={0.05}}
\tiny
\ncline{2,1}{1,2}^{$w_1$}_{$e_1$}
\ncline{1,2}{2,3}_{$e_2$}^{$M$}
\ncline[linecolor=blue]{2,3}{3,2}_{$e_3$}^{$M$}
\ncline{2,1}{3,2}_{$w_4$}^{$e_4$}
\ncarc[arcangle=-25,linecolor=blue]{2,1}{3,1}_{$w_1$}^{$\tilde{e}$}
\ncarc[arcangle=-25,linecolor=blue]{3,1}{3,2}_{$M$}
\end{psmatrix}
\end{pspicture}
\label{figIf|A1|=1,ThenA1AndA2AreInDistinctBlocksOfPmin(B)a}
}
\subfloat[$k \notin \gamma$.]{
\begin{pspicture}(-.5,-.7)(.5,2.2)
\tiny
\begin{psmatrix}[mnode=circle,colsep=0.7,rowsep=0.35]
& {$i$}\\
{$j_1$}	& & {$j_2$}\\
{} & {$k$}	&	 {}
\psset{arrows=-, shortput=nab,labelsep={0.05}}
\tiny
\ncline{2,1}{1,2}^{$w_1$}_{$e_1$}
\ncarc[arcangle=-25,linecolor=blue]{2,1}{3,1}_{$w_1$}^{$\tilde{e}$}
\ncarc[arcangle=-35,linecolor=blue]{3,1}{3,3}_{$M$}
\ncline{1,2}{2,3}^{$M$}_{$e_2$}
\ncline{2,3}{3,2}^{$M$}_{$e_3$}
\ncline{2,1}{3,2}_{$w_4$}^{$e_4$}
\ncline[linecolor=blue]{3,3}{2,3}_{$M$}
\end{psmatrix}
\end{pspicture}
\label{figIf|A1|=1,ThenA1AndA2AreInDistinctBlocksOfPmin(B)b}
}
\caption{$w_4 = \sigma(B)$, $w_1 = \sigma(A_1,i) < M$ and $w(\tilde{e}) > \sigma(B) $.}
\label{figIf|A1|=1,ThenA1AndA2AreInDistinctBlocksOfPmin(B)}
\end{figure}
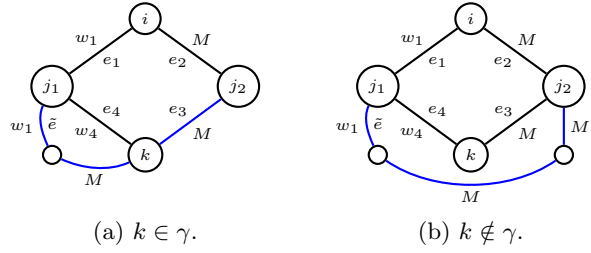
Let $\tilde{e}$ be the edge of $\gamma$ incident to $j_1$.
As $w(\tilde{e}) > \sigma(B) = w_4$,
the Star condition applied to $\lbrace \tilde{e}, e_1, e_4 \rbrace$,
implies $w(\tilde{e}) = w_1 > w_4$.
We get
\begin{equation}
\label{eqwA<w_1=w(Tildee)<w_2=M}
w_4 < w_1 = w(\tilde{e}) < w_2 = M.
\end{equation}
Let $\tilde{C}$ be the cycle formed by $\lbrace e_1 \rbrace \cup \lbrace e_2 \rbrace \cup \gamma$.
By~(\ref{eqwA<w_1=w(Tildee)<w_2=M}),
the Cycle condition implies $w(e) = M$ for every edge $e$ in $\gamma \setminus \lbrace \tilde{e} \rbrace$
and $w(e) = w_1$ for any chord $e$ of $\tilde{C}$ incident to $j_1$.
In particular,
$e_4$ cannot be a chord of $\tilde{C}$
and the situation represented in Figure~\ref{figIf|A1|=1,ThenA1AndA2AreInDistinctBlocksOfPmin(B)a} is not possible.
Then,
by~(\ref{eqwA<w_1=w(Tildee)<w_2=M}) and the Reinforced pan condition,
$j_2$ is linked to $j_1$ by an edge,
a contradiction.

Let $A_{1,1}, A_{1,2}, \ldots, A_{1,p}$ be the blocks of $\mathcal{P}_{\min}(A_1)$.
If $\sigma(B) < \sigma(A_1)$,
then there exists a block $B_j$ of $\mathcal{P}_{\min}(B)$
such that $A_{1,l} \subseteq B_j$ for all $l$,
$1 \leq l \leq p$.
As $j_1$ belongs to one block of $\mathcal{P}_{\min}(A_1)$,
we get by the previous reasoning that $j_2$ belongs to a block of $\mathcal{P}_{\min}(B)$
distinct from $B_j$.
We henceforth assume $\sigma(B) = \sigma(A_1)$.
Let us assume the existence of a block $A_{1,l}$ of $\mathcal{P}_{\min}(A_1)$
such that $A_{1,l}$ and $A_2$ belong to the same block $B_j$ of $\mathcal{P}_{\min}(B)$.
By the previous reasoning,
$j_1$ cannot belong to $A_{1,l}$.
Let $\gamma$ be a shortest path in $G_{B_j}$ linking $j_2$ to a vertex $\tilde{l}$ in $A_{1,l}$
and such that $w(e) > \sigma(B)$ for every edge $e$ in $\gamma$.
Let $\gamma'$ be a shortest path in $G_{A_1}$
linking $\tilde{l}$ to $j_1$.
We select $A_{1,l}$ and $\tilde{l} \in A_{1,l}$
such that $\gamma$ and $\gamma'$ are as short as possible.
Let us note that $\gamma$ cannot contain $j_1$,
otherwise $j_1$ and $j_2$ belong to the same block of $\mathcal{P}_{\min}(B)$,
a contradiction.
Moreover,
$\tilde{l}$ is the unique vertex common to $\gamma$ and $\gamma'$,
otherwise it contradicts the choice of $\tilde{l}$ or $A_{1,l}$.
Let us assume
$k \in \gamma$
(resp. $k \notin \gamma$)
as represented in Figure~\ref{figIf|A1|=1,ThenA1AndA2AreInDistinctBlocksOfPmin(B)a-bis}
(resp. Figure~\ref{figIf|A1|=1,ThenA1AndA2AreInDistinctBlocksOfPmin(B)b-bis}).
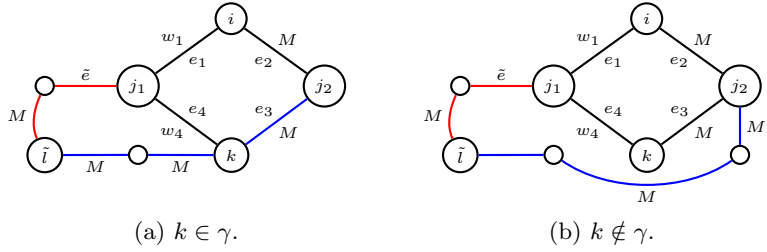
\begin{figure}[!h]
\centering
\subfloat[$k \in \gamma$.]{
\begin{pspicture}(-.5,-.7)(.5,2.2)
\tiny
\begin{psmatrix}[mnode=circle,colsep=0.7,rowsep=0.35]
& & {$i$}\\
{} & {$j_1$}	& & {$j_2$}\\
{$\tilde{l}$} & {} & {$k$}
\psset{arrows=-, shortput=nab,labelsep={0.05}}
\tiny
\ncline{2,2}{1,3}^{$w_1$}_{$e_1$}
\ncline{1,3}{2,4}_{$e_2$}^{$M$}
\ncline[linecolor=blue]{2,4}{3,3}_{$e_3$}^{$M$}
\ncline{2,2}{3,3}_{$w_4$}^{$e_4$}
\ncline[linecolor=red]{2,1}{2,2}^{$\tilde{e}$}
\ncarc[arcangle=-25,linecolor=red]{2,1}{3,1}_{$M$}
\ncline[linecolor=blue]{3,1}{3,2}_{$M$}
\ncline[linecolor=blue]{3,2}{3,3}_{$M$}
\end{psmatrix}
\end{pspicture}
\label{figIf|A1|=1,ThenA1AndA2AreInDistinctBlocksOfPmin(B)a-bis}
}
\subfloat[$k \notin \gamma$.]{
\begin{pspicture}(-.5,-.7)(.5,2.2)
\tiny
\begin{psmatrix}[mnode=circle,colsep=0.7,rowsep=0.35]
& & {$i$}\\
{} & {$j_1$}	& & {$j_2$}\\
{$\tilde{l}$} & {} & {$k$}	&	 {}
\psset{arrows=-, shortput=nab,labelsep={0.05}}
\tiny
\ncline{2,2}{1,3}^{$w_1$}_{$e_1$}
\ncline[linecolor=blue]{3,1}{3,2}
\ncarc[arcangle=-35,linecolor=blue]{3,2}{3,4}_{$M$}
\ncline{1,3}{2,4}^{$M$}_{$e_2$}
\ncline{2,4}{3,3}^{$M$}_{$e_3$}
\ncline{2,2}{3,3}_{$w_4$}^{$e_4$}
\ncline[linecolor=red]{2,1}{2,2}^{$\tilde{e}$}
\ncarc[arcangle=-25,linecolor=red]{2,1}{3,1}_{$M$}
\ncline[linecolor=blue]{3,4}{2,4}_{$M$}
\end{psmatrix}
\end{pspicture}
\label{figIf|A1|=1,ThenA1AndA2AreInDistinctBlocksOfPmin(B)b-bis}
}
\caption{$w_1 = \sigma(A_1,i) < M$ and $w(\tilde{e}) = \sigma(A_1) = \sigma(B) = w_4$.}
\label{figIf|A1|=1,ThenA1AndA2AreInDistinctBlocksOfPmin(B)-bis}
\end{figure}
As $j_1 \notin A_{1,l}$,
$\gamma'$ contains at least one edge $\tilde{e}$ with weight $\sigma(A_1)$.
As $w_1 < w_2 = M$
and as $w(e) > \sigma(B) = \sigma(A_1) = w(\tilde{e})$ for all edge $e$ in $\gamma$,
the Cycle condition applied to $\lbrace e_1 \rbrace \cup \lbrace e_2 \rbrace \cup \gamma \cup \gamma'$
implies that $\tilde{e}$ is incident to $j_1$
with
$w(\tilde{e})<M$
and $w(e) = M$ for any edge $e$ in $\gamma$ and in $\gamma' \setminus \lbrace \tilde{e} \rbrace$.
As $\max(w_1, w(\tilde{e})) < M$,
the Reinforced Cycle condition implies that $j_1$ and $j_2$ are linked by an edge,
a contradiction.

\textbf{\ref{itemPropA1A2AndA3AreAllInDistinctBlocksOfPmin(B)}.}
As $\sigma(A_1,i) < \sigma(A_2,i)$,
the Star condition implies $\sigma(A_2,i) = \sigma(A_3,i)$
and there is a unique edge $e_1$ in $\Sigma(A_1,i)$.
By Claim~\ref{itemPropsIf|A1|=1AndIfSigma(B)<M,Then|A2|=1a},
we have $|A_2| = |A_3| = 1$.
Let us set $e_1 = \lbrace i, j_1 \rbrace$ with $j_1 \in A_1$,
$A_2 = \lbrace j_2 \rbrace$,
and $A_3 = \lbrace j_3 \rbrace$.
By Claim~\ref{itemPropA1AndA2AreInDistinctBlocksOfPmin(B)},
$\lbrace j_1 \rbrace$ and $A_2$
(resp. $A_3$)
are in distinct blocks of $\mathcal{P}_{\min}(B)$.
Let us assume that $A_2$ and $A_3$ belong to the same block $B_j$ of $\mathcal{P}_{\min}(B)$.
By Claim~\ref{itemPropsIf|A1|=1AndIfSigma(B)<M,Then|A2|=1a},
there exists a vertex $k \in B$
and a cycle $C_4 = \lbrace j_1, e_1, i, e_2, j_2, e_3, k, e_4, j_1 \rbrace$
with $w_2 = w_3 = \sigma(A_2,i)$, and $w_4 = \sigma(B)$.
Let us set $e'_2 = \lbrace i, j_3 \rbrace$.
As $w_1 = \sigma(A_1,i) < \sigma(A_2,i) = w_2$,
the Star condition applied to $\lbrace e_1, e_2, e'_2 \rbrace$
implies $w(e'_2)=\sigma(A_2,i)$.
We set $\sigma(A_2,i) = M$.
As $B_j \in \mathcal{P}_{\min}(B)$,
there exists a path $\gamma$ linking $j_2$ to $j_3$ in $G_{B_j}$
with $w(e) > \sigma(B)$ for all edge $e$ in~$\gamma$.
We select $\gamma$ as short as possible.
$\gamma$ cannot contain $j_1$,
otherwise $j_1$ and $j_2$ (resp. $j_3$) belong to the same block of $\mathcal{P}_{\min}(B)$,
a contradiction.
Let us assume
$k \in \gamma$
(resp. $k \notin \gamma$)
as represented in Figure~\ref{figReinforcedCycleConditionConstantCycle4-2a}
(resp. Figure~\ref{figReinforcedCycleConditionConstantCycle4-2b}).
\begin{figure}[!h]
\centering
\subfloat[$k \in \gamma$.]{
\begin{pspicture}(-.5,-.7)(.5,2.2)
\tiny
\begin{psmatrix}[mnode=circle,colsep=0.7,rowsep=0.35]
& {$i$}\\
{$j_1$}	& & {$j_2$} & {$j_3$}\\
& {$k$} & {} & {}
\psset{arrows=-, shortput=nab,labelsep={0.05}}
\tiny
\ncline{2,1}{1,2}^{$w_1$}_{$e_1$}
\ncline{1,2}{2,3}_{$e_2$}^{$M$}
\ncline[linecolor=blue]{2,3}{3,2}_{$e_3$}^{$M$}
\ncline{2,1}{3,2}_{$w_4$}^{$e_4$}
\ncarc[arcangle=30]{1,2}{2,4}_{$e'_2$}^{$M$}
\ncline[linecolor=blue]{3,4}{2,4}_{$M$}
\ncline[linecolor=blue]{3,2}{3,3}_{$M$}
\ncline[linecolor=blue]{3,3}{3,4}_{$M$}
\end{psmatrix}
\end{pspicture}
\label{figReinforcedCycleConditionConstantCycle4-2a}
}
\subfloat[$k \notin \gamma$.]{
\begin{pspicture}(-.5,-.7)(.5,2.2)
\tiny
\begin{psmatrix}[mnode=circle,colsep=0.7,rowsep=0.35]
& {$i$}\\
{$j_1$}	& & {$j_2$} & {$j_3$}\\
& {$k$}	&	 {} & {}
\psset{arrows=-, shortput=nab,labelsep={0.05}}
\tiny
\ncline{2,1}{1,2}^{$w_1$}_{$e_1$}
\ncline{1,2}{2,3}^{$M$}_{$e_2$}
\ncline{2,3}{3,2}^{$M$}_{$e_3$}
\ncline{2,1}{3,2}_{$w_4$}^{$e_4$}
\ncarc[arcangle=30]{1,2}{2,4}_{$e'_2$}^{$M$}
\ncline[linecolor=blue]{3,3}{2,3}_{$M$}
\ncline[linecolor=blue]{3,3}{3,4}_{$M$}
\ncline[linecolor=blue]{3,4}{2,4}_{$M$}
\end{psmatrix}
\end{pspicture}
\label{figReinforcedCycleConditionConstantCycle4-2b}
}
\caption{$w_1 = \sigma(A_1,i) < M$ and $w_4 = \sigma(B) < M$.}
\label{figReinforcedCycleConditionConstantCycle4-2}
\end{figure}
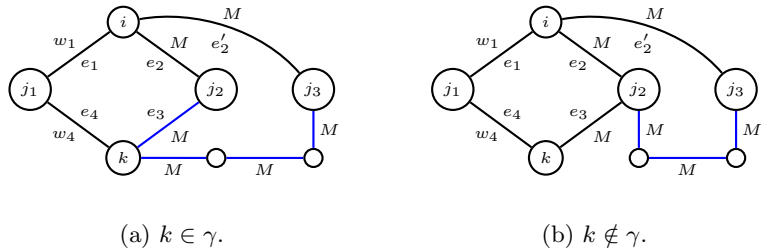
As $w_1 < M$ and as $w_4 = \sigma(B) < M = w(e'_2)$,
the Cycle condition applied to
$\lbrace e_4 \rbrace \cup \lbrace e_1 \rbrace \cup \lbrace e'_2 \rbrace \cup (\gamma \setminus \lbrace e_3 \rbrace)$
(resp. $\lbrace e_4 \rbrace \cup \lbrace e_1 \rbrace \cup \lbrace e'_2 \rbrace \cup \gamma \cup \lbrace e_3 \rbrace$)
implies $w(e) = M$ for all $e$ in $\gamma$.
As $w_4 < M$,
the Reinforced pan condition
applied to the cycle formed by $\lbrace e_2 \rbrace \cup \lbrace e'_2 \rbrace \cup \gamma$
and the path formed by $e_4$
(resp. $\lbrace e_4 \rbrace \cup \lbrace e_3 \rbrace$)
implies that
$j_3$ is linked to a vertex in
$\lbrace j_1, k \rbrace$
(resp. $\lbrace j_1, k, j_2 \rbrace$)
by an edge $e$.
$e$ is necessarily incident to~$k$ as represented in Figure~\ref{figReinforcedCycleConditionConstantCycle5},
otherwise it would link $j_1$ or $A_2$ to $A_3$.
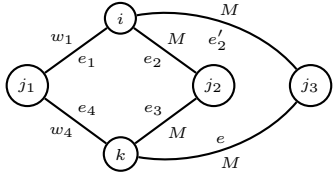
\begin{figure}[!h]
\centering
\begin{pspicture}(-.5,-.2)(.5,2.2)
\tiny
\begin{psmatrix}[mnode=circle,colsep=0.7,rowsep=0.35]
& {$i$}\\
{$j_1$}	& & {$j_2$} & {$j_3$}\\
& {$k$}
\psset{arrows=-, shortput=nab,labelsep={0.05}}
\tiny
\ncline{2,1}{1,2}^{$w_1$}_{$e_1$}
\ncline{1,2}{2,3}_{$e_2$}^{$M$}
\ncline{2,3}{3,2}_{$e_3$}^{$M$}
\ncline{2,1}{3,2}_{$w_4$}^{$e_4$}
\ncarc[arcangle=30]{1,2}{2,4}_{$e'_2$}^{$M$}
\ncarc[arcangle=-30]{3,2}{2,4}_{$M$}^{$e$}
\end{psmatrix}
\end{pspicture}
\caption{$w_1 < M$ and $w_4 < M$.}
\label{figReinforcedCycleConditionConstantCycle5}
\end{figure}
By the Star condition applied to $\lbrace e, e_3, e_4 \rbrace$,
$e$ has weight $M$.
Then,
by the Reinforced adjacent cycles condition,
there exists an edge linking $j_2$ to $j_3$,
a contradiction.
\end{proof}

\begin{proposition}
\label{Propositionpartitionin2set1-bis}
Let us assume that the Path, Star, Cycle,
Reinforced cycle and Reinforced pan conditions are satisfied.
Let us consider $i \in N$
and $A \subseteq B \subseteq N \setminus \lbrace i \rbrace$
with $A \cup \lbrace i \rbrace$ and $B$ in $\mathcal{F}$.
Let us assume $A \notin \mathcal{F}$,
and let $A_1$, $A_2$
be two connected components of $A$
with $\sigma(A_1, i) = \sigma(A_2, i) = M$.
If $|A_1| \geq 2$ and if $\sigma(A_1)<M$,
then
\begin{enumerate}
\item
\label{itemPropsIf|A1|=1AndIfSigma(B)<M,Then|A2|=1a-bis}
$|A_2| = 1$
and $\sigma(B) = \sigma(A_1)$.
Moreover,
there exists a unique edge $\tilde{e}_1$ in $\Sigma(B)$
and setting $\tilde{e}_1= \lbrace j_1, k_1 \rbrace$
with $j_1$ and $k_1$ in $A_1$ and $A_2 = \lbrace j_2 \rbrace$,
there exist a vertex $k_2 \in B$,
and a cycle $C_5 = \lbrace k_1, \tilde{e}_1, j_1, e_2, i, e_3, j_2, e_4,\\ k_2, e_5, k_1 \rbrace$
with $w(\tilde{e}_1) = \sigma(B)$ and $w_2 = w_3 = w_4 = w_5 = M$.
\item
\label{itemPropA1AndA2AreInDistinctBlocksOfPmin(B)-bis}
$\lbrace j_1 \rbrace$ and $A_2$ are in distinct blocks of $\mathcal{P}_{\min}(B)$.
Moreover,
there is a unique block in $\mathcal{P}_{\min}(A_1)$ linked to $i$
and this block contains $j_1$ and belongs to a block of $\mathcal{P}_{\min}(B)$
different from the one containing $A_2$.
\item
\label{itemPropA1A2AndA3AreAllInDistinctBlocksOfPmin(B)-bis}
If the Reinforced adjacent cycles condition is satisfied,
then $A_2$ is unique.
\end{enumerate}
\end{proposition}

\begin{proof}
\textbf{\ref{itemPropsIf|A1|=1AndIfSigma(B)<M,Then|A2|=1a-bis}.}
Let $\tilde{e}_1 = \lbrace s, t \rbrace$ be an edge in $E(A_1)$
with weight $w(\tilde{e}_1) < M$
(such an edge exists as $\sigma(A_1)<M$).
Let $P$ be a shortest path in $G_{A_1}$ connecting $\tilde{e}_1$ to a vertex $j_1$ in $A_1$ such that
$e'_1 = \lbrace i,j_1 \rbrace$ belongs to $\Sigma(A_1,i)$.
We select $\tilde{e}_1$ such that $P$ is as short as possible
($P$ may be reduced to $j_1$).
We can assume $t \in P$. 
We have $w(e) \geq M$ for any edge $e$ in $P$,
otherwise we can change $\tilde{e}_1$.
As $w(\tilde{e}_1) < M = w(e'_1)$,
the Path condition implies
$w(e) \leq \max (w(\tilde{e}_1),w(e'_1))=M$ for any edge $e$ in $P$.
Therefore,
we have
\begin{equation}
\label{eqw(e)=MForAnyEdgeeInP}
w(e) = M \textrm{ for any edge } e \textrm{ in } P.
\end{equation}
Let $e_2' = \lbrace i, j_2 \rbrace$ be an edge in $\Sigma(A_2,i)$.
Let $\gamma$ be a shortest path in $G_{B}$ linking $j_2$
to a vertex $j^*$ in $V(P) \cup \lbrace s \rbrace$.
Let us denote by $P'$ the path formed by
$P \cup \lbrace \tilde{e}_1 \rbrace$
and by $P'_{j_1,j^*}$
the subpath of $P'$ linking $j_1$
to $j^*$.
Then $e'_1$, $e'_2$, $\gamma$, and $P'_{j_1,j^*}$
form a cycle $C_m = \lbrace e_1, e_2, \ldots, e_m \rbrace$ with $m \geq 4$.
We select $j_1$, $P$, and $\gamma$ such that
$C_m$ is as short as possible.

Let us first assume $j^* \in V(P)$.
Let us denote by $P_{j_1,j^*}$
(resp. $P_{t,j^*}$)
the subpath of $P$ linking $j_1$
(resp. $t$)
to $j^*$.
The Path condition applied to
$\lbrace \tilde{e}_1 \rbrace \cup P_{t,j^*} \cup \gamma \cup \lbrace e'_2 \rbrace $
implies $w(e) \leq \max  (w(\tilde{e}_1),w(e'_2))=M$ for any edge $e$ in $\gamma$.
Let us assume that there exists an edge $e$ in $\gamma$
with $w(e)<M$.
If $P_{t,j^*}$ contains at least one edge,
then (\ref{eqw(e)=MForAnyEdgeeInP}) and the Path condition imply $M \leq \max (w(\tilde{e}_1), w(e))$,
a contradiction.
Otherwise, $P_{t,j^*}$ reduces to~$j^*$
and $\tilde{e}_1$ is incident to $j^*$.
Then,
by the Star condition,
the edge in $\gamma$
incident to $j^*$
has weight $M$,
and the Path condition also implies $M \leq \max(w(\tilde{e}_1), w(e))$,
a contradiction.
Hence,
every edge in $\gamma$ has weight $M$ and $C_m$ is a constant cycle.
We can assume $j^* = 2$ and $2 \leq j_1 < i < j_2 \leq m$
as represented in Figure~\ref{figReinforcedCycleConditionConstantCycle+Pan1a},
after renumbering if necessary.
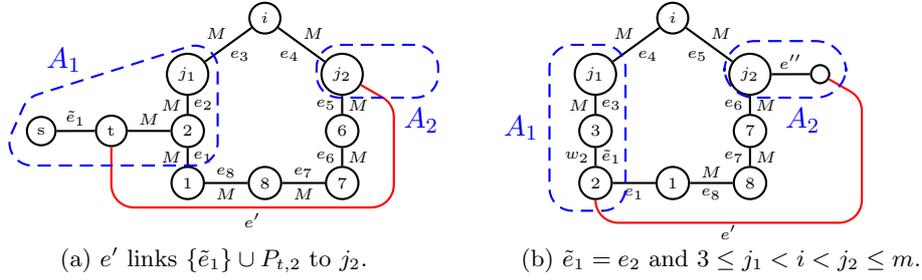
\begin{figure}[!h]
\centering
\subfloat[$e'$ links $\lbrace \tilde{e}_1 \rbrace \cup P_{t,2}$ to $j_2$.]{
\label{figReinforcedCycleConditionConstantCycle+Pan1a}
\begin{pspicture}(-.5,-.7)(1,2)
\tiny
\begin{psmatrix}[mnode=circle,colsep=0.5,rowsep=0.2]
  & & & {$i$}\\
 & & {$j_1$}	& & {$j_2$}\\
{s} & {t} &  {$2$}	&	& {$6$}\\
& 	&{$1$} & {$8$} & {$7$}
\psset{arrows=-, shortput=nab,labelsep={0.05}}
\tiny
\ncline{4,4}{4,5}_{$M$}^{$e_7$}
\ncline{4,3}{4,4}_{$M$}^{$e_8$}
\ncline{3,3}{4,3}_{$M$}^{$e_1$}
\ncline{2,3}{3,3}_{$M$}^{$e_2$}
\ncline{2,3}{1,4}^{$M$}_{$e_3$}
\ncline{1,4}{2,5}^{$M$}_{$e_4$}
\ncline{2,5}{3,5}^{$M$}_{$e_5$}
\ncline{4,5}{3,5}_{$M$}^{$e_6$}
\ncline{3,1}{3,2}^{$\tilde{e}_1$}
\ncline{3,2}{3,3}^{$M$}
\ncangles[angleA=270,angleB=-30,armA=.8,armB=.5,linearc=.3,linecolor=red]{3,2}{2,5}_{$e'$}
\normalsize
\uput[0](.7,.9){\textcolor{blue}{$A_{2}$}}
\pspolygon[framearc=1,linestyle=dashed,linecolor=blue,linearc=.3]
(-.4,1.2)(-.4,1.9)(1.2,1.9)(1.2,1.2)
\uput[0](-4,1.7){\textcolor{blue}{$A_{1}$}}
\pspolygon[framearc=1,linestyle=dashed,linecolor=blue,linearc=.3]
(-4.6,.3)(-4.2,1.2)(-1.7,2)(-1.7,.3)
\end{psmatrix}
\end{pspicture}
}
\subfloat[$\tilde{e}_1 = e_2$ and $3 \leq j_1 < i < j_2 \leq m$.]{
\label{figReinforcedCycleConditionConstantCycle+Pan1b}
\begin{pspicture}(-1.5,-.7)(2,2)
\tiny
\begin{psmatrix}[mnode=circle,colsep=0.5,rowsep=0.2]
& {$i$}\\
{$j_1$}	& & {$j_2$} & {}\\
 {$3$}	&	& {$7$}\\
{$2$} & {$1$} & {$8$}
\psset{arrows=-, shortput=nab,labelsep={0.15}}
\tiny
\psset{arrows=-, shortput=nab,labelsep={0.05}}
\ncline{3,1}{4,1}_{$w_2$}^{$\tilde{e}_1$}
\ncline{1,2}{2,3}^{$M$}_{$e_5$}
\ncline{2,3}{3,3}^{$M$}_{$e_6$}
\ncline{2,3}{2,4}^{$e''$}
\ncline{2,1}{1,2}^{$M$}_{$e_4$}
\ncline{2,1}{3,1}_{$M$}^{$e_3$}
\ncline{4,2}{4,3}^{$M$}_{$e_8$}
\ncline{4,1}{4,2}_{$e_1$}
\ncline{4,3}{3,3}_{$M$}^{$e_7$}
\color{black}
\psset{arrows=-, shortput=nab,labelsep={0.0}}
\ncangles[angleA=270,angleB=-30,armA=.3,armB=.5,linearc=.3,linecolor=red]{4,1}{2,4}_{$e'$}
\normalsize
\uput[0](.4,.9){\textcolor{blue}{$A_{2}$}}
\pspolygon[framearc=1,linestyle=dashed,linecolor=blue,linearc=.4]
(-.4,1.2)(-.4,1.95)(1.2,1.95)(1.2,1.2)
\uput[0](-3.3,.8){\textcolor{blue}{$A_{1}$}}
\pspolygon[framearc=1,linestyle=dashed,linecolor=blue,linearc=.3]
(-2.7,-.3)(-2.7,1.9)(-1.7,1.9)(-1.7,-.3)
\end{psmatrix}
\end{pspicture}
}
\caption{$w(\tilde{e}_1) < M$.}
\label{figReinforcedCycleConditionConstantCycle+Pan1}
\end{figure}
By the Reinforced pan condition,
there exists an edge $e'$ in $E$
linking $\lbrace \tilde{e}_1 \rbrace \cup P_{t,2}$ to $j_2$,
and therefore linking $A_1$ to $A_2$,
a contradiction.

Let us now assume $j^* = s$.
Then,
we necessarily have $m \geq 5$.
We can assume
$\tilde{e}_1 = e_2$ with $s=2$ and $t=3$,
and $3 \leq j_1 < i < j_2 \leq m$ as represented in Figure~\ref{figReinforcedCycleConditionConstantCycle+Pan1b}.
Moreover,
by the Cycle condition,
we have $w_2 < w_3 = \cdots = w_m= M$ and $w_1 \leq M$.
If $|A_2| \geq 2$,
then there exists an edge $e''$ in $E(A_2)$
incident to $j_2$.
By the Reinforced cycle condition,
there exists an edge linking $e_2$ to $e''$,
and therefore linking $A_1$ to $A_2$,
a contradiction.
Thus,
we necessarily have $|A_2| =1$
with $A_2 = \lbrace j_2 \rbrace$.
If $w_1 < M$,
then we have $\max(w_1, w_2) < M$
and by the Reinforced Cycle condition,
the edge $\lbrace i, j_2 \rbrace$ is linked to $2$ by an edge $e'$.
If $e' = \lbrace j_2, 2 \rbrace$,
then $e'$ links $A_1$ to $A_2$,
a contradiction.
If $e' = \lbrace i, 2 \rbrace$,
then it contradicts the choice of $j_1$
and the minimality of $C_m$.
Hence,
we have $w_1 = M$.
As $w_2 < M$,
the Reinforced Cycle condition implies the existence of an edge $e'$
(resp. $e''$)
linking the edge $\lbrace i, j_2 \rbrace$
(resp.  $\lbrace j_2, j_2 +1 \rbrace$)
to $\tilde{e}_1$.
$e'$
(resp. $e''$)
is necessarily incident to $i$
(resp. $j_2 +1$),
otherwise it would link $A_2$ to $A_1$.
If $e'$ is incident to $2$ as represented in Figure~\ref{figSmallerCyclef},
then it contradicts the choice of $j_1$ (and the minimality of $C_m$).
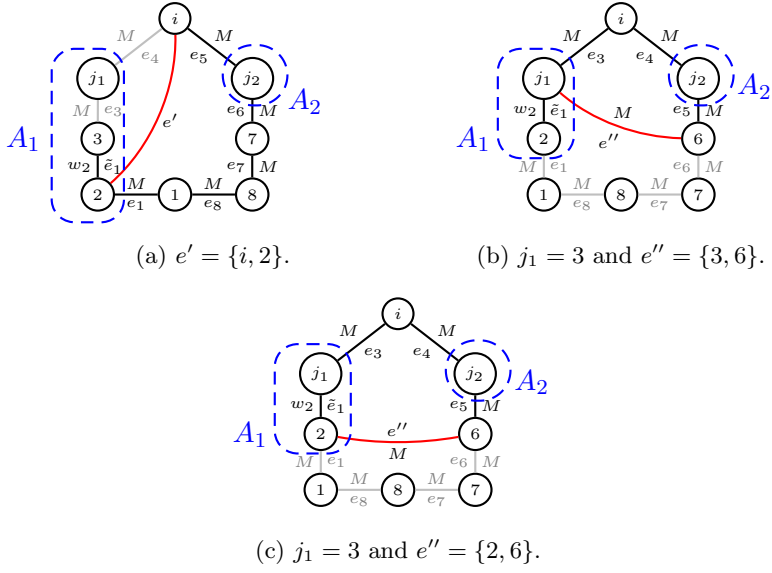
\begin{figure}[!h]
\centering
\subfloat[$e'=\lbrace i, 2 \rbrace$.]{
\label{figSmallerCyclef}
\begin{pspicture}(-1,-.5)(2,2.5)
\tiny
\begin{psmatrix}[mnode=circle,colsep=0.5,rowsep=0.25]
& {$i$}\\
{$j_1$}	& & {$j_2$}\\
 {$3$}	&	& {$7$}\\
{$2$} & {$1$} & {$8$}
\psset{arrows=-, shortput=nab,labelsep={0.15}}
\tiny
\psset{arrows=-, shortput=nab,labelsep={0.05}}
\ncline{3,1}{4,1}_{$w_2$}^{$\tilde{e}_1$}
\ncline{1,2}{2,3}^{$M$}_{$e_5$}
\ncline{2,3}{3,3}^{$M$}_{$e_6$}
\ncline{4,2}{4,3}^{$M$}_{$e_8$}
\ncline{4,1}{4,2}_{$e_1$}^{$M$}
\ncline{4,3}{3,3}_{$M$}^{$e_7$}
\color{gray}
\ncline[linecolor=lightgray]{2,1}{1,2}^{$M$}_{$e_4$}
\ncline[linecolor=lightgray]{2,1}{3,1}_{$M$}^{$e_3$}
\color{black}
\psset{arrows=-, shortput=nab,labelsep={0.0}}
\ncarc[arcangle=-25,linecolor=red]{4,1}{1,2}_{$e'$}
\normalsize
\uput[0](.4,1.3){\textcolor{blue}{$A_{2}$}}
\pspolygon[framearc=1,linestyle=dashed,linecolor=blue,linearc=.4]
(-.45,1.25)(-.45,2.05)(.4,2.05)(.4,1.25)
\uput[0](-3.3,.8){\textcolor{blue}{$A_{1}$}}
\pspolygon[framearc=1,linestyle=dashed,linecolor=blue,linearc=.3]
(-2.7,-.3)(-2.7,2)(-1.75,2)(-1.75,-.3)
\end{psmatrix}
\end{pspicture}
}
\subfloat[$j_1=3$ and $e''=\lbrace 3,6 \rbrace$.]{
\label{figSmallerCycled}
\begin{pspicture}(-1,-.5)(1,2.5)
\tiny
\begin{psmatrix}[mnode=circle,colsep=0.5,rowsep=0.25]
& {$i$}\\
{$j_1$}	& & {$j_2$}\\
 {$2$}	&	& {$6$}\\
{$1$} & {$8$} & {$7$}
\psset{arrows=-, shortput=nab,labelsep={0.15}}
\tiny
\psset{arrows=-, shortput=nab,labelsep={0.05}}
\ncline{2,1}{1,2}^{$M$}_{$e_3$}
\ncline{2,1}{3,1}_{$w_2$}^{$\tilde{e}_1$}
\ncline{1,2}{2,3}^{$M$}_{$e_4$}
\ncline{2,3}{3,3}^{$M$}_{$e_5$}
\ncarc[arcangle=-20,linecolor=red]{2,1}{3,3}_{$e''$}^{$M$}
\color{gray}
\ncline[linecolor=lightgray]{4,2}{4,3}^{$M$}_{$e_7$}
\ncline[linecolor=lightgray]{4,1}{4,2}^{$M$}_{$e_8$}
\ncline[linecolor=lightgray]{3,1}{4,1}^{$e_1$}_{$M$}
\ncline[linecolor=lightgray]{4,3}{3,3}_{$M$}^{$e_6$}
\normalsize
\uput[0](.4,1.5){\textcolor{blue}{$A_{2}$}}
\pspolygon[framearc=1,linestyle=dashed,linecolor=blue,linearc=.4]
(-.45,1.25)(-.45,2.05)(.4,2.05)(.4,1.25)
\uput[0](-3.3,.8){\textcolor{blue}{$A_{1}$}}
\pspolygon[framearc=1,linestyle=dashed,linecolor=blue,linearc=.3]
(-2.7,.55)(-2.7,2)(-1.7,2)(-1.7,.55)
\end{psmatrix}
\end{pspicture}
}\\
\subfloat[$j_1=3$ and $e''=\lbrace 2,6 \rbrace$.]{
\label{figSmallerCyclec}
\begin{pspicture}(-1,-.5)(1,2.5)
\tiny
\begin{psmatrix}[mnode=circle,colsep=0.5,rowsep=0.25]
& {$i$}\\
{$j_1$}	& & {$j_2$}\\
 {$2$}	&	& {$6$}\\
{$1$} & {$8$} & {$7$}
\psset{arrows=-, shortput=nab,labelsep={0.15}}
\tiny
\psset{arrows=-, shortput=nab,labelsep={0.05}}
\ncline{2,1}{1,2}^{$M$}_{$e_3$}
\ncline{2,1}{3,1}_{$w_2$}^{$\tilde{e}_1$}
\ncline{1,2}{2,3}^{$M$}_{$e_4$}
\ncline{2,3}{3,3}^{$M$}_{$e_5$}
\ncarc[arcangle=-10,linecolor=red]{3,1}{3,3}^{$e''$}_{$M$}
\color{gray}
\ncline[linecolor=lightgray]{4,2}{4,3}^{$M$}_{$e_7$}
\ncline[linecolor=lightgray]{4,1}{4,2}^{$M$}_{$e_8$}
\ncline[linecolor=lightgray]{3,1}{4,1}^{$e_1$}_{$M$}
\ncline[linecolor=lightgray]{4,3}{3,3}_{$M$}^{$e_6$}
\normalsize
\uput[0](.4,1.5){\textcolor{blue}{$A_{2}$}}
\pspolygon[framearc=1,linestyle=dashed,linecolor=blue,linearc=.4]
(-.45,1.25)(-.45,2.05)(.4,2.05)(.4,1.25)
\uput[0](-3.3,.8){\textcolor{blue}{$A_{1}$}}
\pspolygon[framearc=1,linestyle=dashed,linecolor=blue,linearc=.3]
(-2.7,.55)(-2.7,2)(-1.7,2)(-1.7,.55)
\end{psmatrix}
\end{pspicture}
}
\caption{$w(\tilde{e}_1) = w_2 < M$ and $w_1 = M$.}
\label{figSmallerCycle}
\end{figure}
Hence,
we have $e' = \lbrace i, 3 \rbrace$ and this implies $j_1 = 3$
(otherwise it still contradicts the choice of~$j_1$).
If $e'' = \lbrace 3, j_2 +1 \rbrace$ as represented in Figure~\ref{figSmallerCycled},
then it contradicts the choice of $\gamma$.
Thus,
we necessarily have $e'' = \lbrace 2, j_2 +1 \rbrace$ as represented in Figure~\ref{figSmallerCyclec}.
If $m \geq 6$,
then it contradicts the minimality of $\gamma$.
Therefore,
we necessarily have $m=5$.

Let us assume that there exist an edge $\tilde{e} \not= \tilde{e}_1$
in $E(B)$ with weight $w(\tilde{e}) \leq w(\tilde{e}_1)$.
By the Cycle condition,
any chord of $C_5$ has weight $M$.
Thus,
$\tilde{e}$ cannot belong to $\hat{E}(C_5)$.
Let $\tilde{P}$ be a shortest path in $G_B$ linking $\tilde{e}$ to $C_5$.
As $w(\tilde{e}) \leq w(\tilde{e}_1) < M$,
the Pan condition applied to the pan formed by $C_5$ and $\tilde{P} \cup \lbrace \tilde{e} \rbrace$
implies that one of the edges in $E(C_5)$ adjacent to $\tilde{e}_1$
should have weight $w(\tilde{e}_1)$,
a contradiction.
Therefore,
we necessarily have $w(\tilde{e}_1) = \sigma(B) = \sigma(A_1)$
and $\Sigma(B) = \Sigma(A_1) = \lbrace \tilde{e}_1 \rbrace$.

\textbf{\ref{itemPropA1AndA2AreInDistinctBlocksOfPmin(B)-bis}.}
By Claim~\ref{itemPropsIf|A1|=1AndIfSigma(B)<M,Then|A2|=1a-bis},
we have $|A_2| = 1$ and $\sigma(B) = \sigma(A_1)$.
Moreover,
there exists a unique edge $e_1$ in $\Sigma(B) = \Sigma(A_1)$
and setting $e_1 = \lbrace j_1, k_1 \rbrace$
with $j_1$ and $k_1$ in $A_1$ and $A_2 = \lbrace j_2 \rbrace$,
there exists a vertex $k_2$ in $B$ and a cycle $C_5 = \lbrace k_1, e_1, j_1, e_2, i, e_3, j_2, e_4, k_2, e_5, k_1\rbrace$
with $w_1 = \sigma(A_1)$ and $w_2 = \cdots = w_5 = M$.

Let us first assume that $\lbrace j_1 \rbrace$ and $A_2$ belong to the same block $B_j$ of $\mathcal{P}_{\min}(B)$.
Then,
there exists a path $\gamma$ in $G_{B_j}$ linking $j_1$ to $j_2$
as represented in Figure~\ref{figC5Path}
and such that $w(e) > \sigma(B)$ for every edge $e$ in $\gamma$.
\begin{figure}[!h]
\centering
\subfloat[$k_2 \in \gamma$.]{
\label{figC5Patha}
\begin{pspicture}(-1,-.2)(1,2.6)
\tiny
\begin{psmatrix}[mnode=circle,colsep=0.5,rowsep=0.25]
& {$i$}\\
{$j_1$}	& & {$j_2$}\\
{$k_1$}	&	& {$k_2$}\\
{} &  & {}
\psset{arrows=-, shortput=nab,labelsep={0.15}}
\tiny
\psset{arrows=-, shortput=nab,labelsep={0.05}}
\ncline{2,1}{1,2}^{$M$}_{$e_2$}
\ncline{2,1}{3,1}_{$w_1$}^{$e_1$}
\ncline{1,2}{2,3}^{$M$}_{$e_3$}
\ncline[linecolor=blue]{2,3}{3,3}^{$M$}_{$e_4$}
\ncline{3,1}{3,3}_{$M$}^{$e_5$}
\ncarc[arcangle=-60,linecolor=blue]{2,1}{4,1}_{$e'$}
\ncarc[arcangle=-10,linecolor=blue]{4,1}{4,3}
\ncarc[arcangle=-20,linecolor=blue]{4,3}{3,3}
\normalsize
\uput[0](.4,1.5){\textcolor{blue}{$A_{2}$}}
\pspolygon[framearc=1,linestyle=dashed,linecolor=blue,linearc=.4]
(-.45,1.25)(-.45,2.05)(.4,2.05)(.4,1.25)
\uput[0](-3.5,1.4){\textcolor{blue}{$A_{1}$}}
\pspolygon[framearc=1,linestyle=dashed,linecolor=blue,linearc=.3]
(-2.9,.3)(-2.9,2)(-1.6,2)(-1.6,.3)
\end{psmatrix}
\end{pspicture}
}
\subfloat[$k_2 \notin \gamma$.]{
\label{figC5Pathb}
\begin{pspicture}(-1,-.2)(1,2.6)
\tiny
\begin{psmatrix}[mnode=circle,colsep=0.5,rowsep=0.25]
& {$i$}\\
{$j_1$}	& & {$j_2$}\\
{$k_1$}	&	& {$k_2$}\\
{} &  & {}
\psset{arrows=-, shortput=nab,labelsep={0.15}}
\tiny
\psset{arrows=-, shortput=nab,labelsep={0.05}}
\ncline{2,1}{1,2}^{$M$}_{$e_2$}
\ncline{2,1}{3,1}_{$w_1$}^{$e_1$}
\ncline{1,2}{2,3}^{$M$}_{$e_3$}
\ncline{2,3}{3,3}^{$M$}_{$e_4$}
\ncline{3,1}{3,3}_{$M$}^{$e_5$}
\ncarc[arcangle=-60,linecolor=blue]{2,1}{4,1}_{$e'$}
\ncarc[arcangle=-10,linecolor=blue]{4,1}{4,3}
\ncarc[arcangle=-60,linecolor=blue]{4,3}{2,3}
\normalsize
\uput[0](.4,1.5){\textcolor{blue}{$A_{2}$}}
\pspolygon[framearc=1,linestyle=dashed,linecolor=blue,linearc=.4]
(-.45,1.25)(-.45,2.05)(.4,2.05)(.4,1.25)
\uput[0](-3.5,1.4){\textcolor{blue}{$A_{1}$}}
\pspolygon[framearc=1,linestyle=dashed,linecolor=blue,linearc=.3]
(-2.9,.3)(-2.9,2)(-1.6,2)(-1.6,.3)
\end{psmatrix}
\end{pspicture}
}
\caption{$w_1=\sigma(A_1) < M = w_2 = \cdots = w_5$.}
\label{figC5Path}
\end{figure}
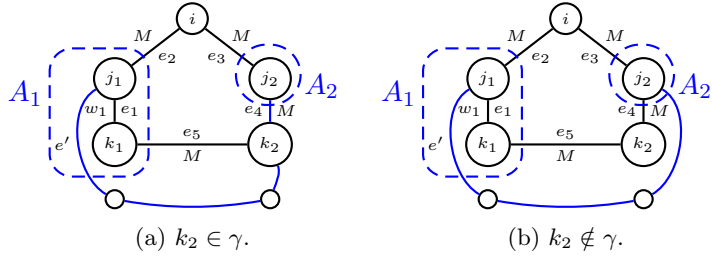
We select $\gamma$ as short as possible.
Let $\tilde{C}$ be the cycle formed by $\lbrace e_2 \rbrace \cup \lbrace e_3 \rbrace \cup \gamma$.
Any edge in $E(\tilde{C})$ has a weight strictly greater than $\sigma(A_1)$.
By the Cycle condition,
$e_1$ cannot be a chord of $\tilde{C}$.
Thus,
$\gamma$ cannot contain $k_1$ but may contain~$k_2$.
Let us assume $k_2 \in \gamma$
(resp. $k_2 \notin \gamma$)
as represented in Figure~\ref{figC5Patha}
(resp. Figure~\ref{figC5Pathb}).
Let $e'$ be the edge of $\gamma$ incident to $j_1$.
As $w_1 < M = w_2$,
the Star condition applied to $\lbrace e', e_1, e_2 \rbrace$ implies $w(e') = M$.
Then,
the Cycle condition applied to the cycle formed by $\lbrace e_1 \rbrace \cup \lbrace e_5 \rbrace \cup \gamma$
(resp. $\lbrace e_1 \rbrace \cup \lbrace e_5 \rbrace \cup \lbrace e_4 \rbrace \cup \gamma$)
implies $w(e) = M$ for any edge $e$ in $\gamma$.
Hence,
$\tilde{C}$ is a constant cycle.
As $w_1 = \sigma(A_1) < M$,
the Reinforced pan condition implies the existence of an edge linking
$j_2$ to $e_1$,
a contradiction.

As $\lbrace j_1, k_1 \rbrace$ is the unique edge in $\Sigma(A_1)$,
$\mathcal{P}_{\min}(A_1)$ contains at most two blocks.
Let $A_{1,1}, \ldots, A_{1,p}$ with $p \leq 2$ be the blocks of $\mathcal{P}_{\min}(A_1)$ linked to $i$
and let us assume $j_1 \in A_{1,1}$.
Let us assume $p=2$.
Then,
$A_{1,2}$ necessarily contains $k_1$.
Let $\tilde{e} = \lbrace i, \tilde{l} \rbrace$  be an edge in $E(A_1,i)$
linking $i$ to~$A_{1,2}$.
If $\tilde{l} = k_1$ as represented in Figure~\ref{figIf|A1|=1,ThenA1AndA2AreInDistinctBlocksOfPmin(B)a-ter},
then
as $w_1 < M = w_5$
the Star condition applied to $\lbrace \tilde{e}, e_1, e_5 \rbrace$
implies $w(\tilde{e}) = M$.
\begin{figure}[!h]
\centering
\subfloat[$\tilde{l} = k_1$.]{
\begin{pspicture}(-.5,-.5)(1,2)
\tiny
\begin{psmatrix}[mnode=circle,colsep=0.55,rowsep=0.35]
 & {$i$}\\
{$j_1$}	& & {$j_2$}\\
{$k_1$} & & {$k_2$}
\psset{arrows=-, shortput=nab,labelsep={0.05}}
\tiny
\ncline{2,1}{1,2}^{$M$}_{$e_2$}
\ncline{2,1}{3,1}_{$w_1$}^{$e_1$}
\ncline{1,2}{2,3}_{$e_3$}^{$M$}
\ncline{2,3}{3,3}_{$e_4$}^{$M$}
\ncarc[arcangle=-20,linecolor=red]{3,1}{1,2}_{$\tilde{e}$}
\ncline{3,1}{3,3}_{$M$}^{$e_5$}
\normalsize
\uput[0](.1,1.5){\textcolor{blue}{$A_{2}$}}
\pspolygon[framearc=1,linestyle=dashed,linecolor=blue,linearc=.4]
(-.5,.65)(-.5,1.4)(.3,1.4)(.3,.65)
\uput[0](-3,1.7){\textcolor{blue}{$A_{1}$}}
\pspolygon[framearc=1,linestyle=dashed,linecolor=blue,linearc=.3]
(-3,-.4)(-3,1.4)(-1.85,1.4)(-1.85,-.4)
\end{psmatrix}
\end{pspicture}
\label{figIf|A1|=1,ThenA1AndA2AreInDistinctBlocksOfPmin(B)a-ter}
}
\subfloat[$\tilde{l} \not= k_1$.]{
\begin{pspicture}(-1,-.5)(1,2)
\tiny
\begin{psmatrix}[mnode=circle,colsep=0.55,rowsep=0.35]
 & & & {$i$}\\
&  & {$j_1$}	& & {$j_2$}\\
{$\tilde{l}$} & {} & {$k_1$} & & {$k_2$}
\psset{arrows=-, shortput=nab,labelsep={0.05}}
\tiny
\ncline{2,3}{1,4}^{$M$}_{$e_2$}
\ncline{2,3}{3,3}_{$w_1$}^{$e_1$}
\ncline{1,4}{2,5}_{$e_3$}^{$M$}
\ncline{2,5}{3,5}_{$e_4$}^{$M$}
\ncarc[arcangle=40,linecolor=red]{3,1}{1,4}^{$\tilde{e}$}
\ncline[linecolor=blue]{3,1}{3,2}_{$M$}
\ncline[linecolor=blue]{3,2}{3,3}_{$M$}^{$e'$}
\ncline{3,3}{3,5}_{$M$}^{$e_5$}
\normalsize
\uput[0](.1,1.5){\textcolor{blue}{$A_{2}$}}
\pspolygon[framearc=1,linestyle=dashed,linecolor=blue,linearc=.4]
(-.5,.65)(-.5,1.4)(.3,1.4)(.3,.65)
\uput[0](-5,1.7){\textcolor{blue}{$A_{1}$}}
\pspolygon[framearc=1,linestyle=dashed,linecolor=blue,linearc=.3]
(-4.9,-.45)(-4.9,1.45)(-1.8,1.45)(-1.8,-.45)
\uput[0](-4.5,1.2){\textcolor{cyan}{$A_{1,1}$}}
\pspolygon[framearc=1,linestyle=dashed,linewidth=.02,linecolor=cyan,linearc=.3]
(-3.7,.7)(-3.7,1.35)(-1.9,1.35)(-1.9,.7)
\uput[0](-4.9,.6){\textcolor{cyan}{$A_{1,2}$}}
\pspolygon[framearc=1,linestyle=dashed,linewidth=.02,linecolor=cyan,linearc=.3]
(-4.65,-.35)(-4.65,.4)(-1.9,.4)(-1.9,-.35)
\end{psmatrix}
\end{pspicture}
\label{figIf|A1|=1,ThenA1AndA2AreInDistinctBlocksOfPmin(B)b-ter}
}
\caption{$w_1 = \sigma(A_1) < M$.}
\label{figIf|A1|=1,ThenA1AndA2AreInDistinctBlocksOfPmin(B)-ter}
\end{figure}
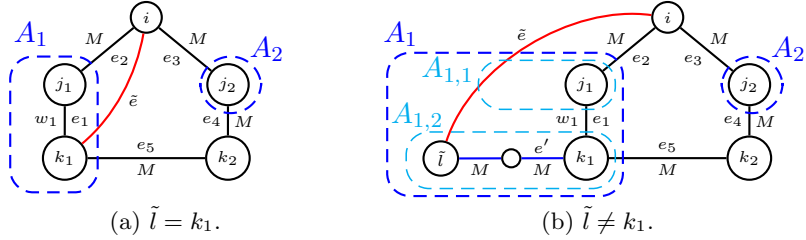
As $w_1 < M$,
the Reinforced pan condition applied to the pan formed by $\lbrace e_3, e_4, e_5, \tilde{e} \rbrace$
and $e_1$ implies that
$j_2$ is linked to $j_1$ or $k_1$,
a contradiction.
We henceforth assume $\tilde{l} \not= k_1 $.
Let $\gamma'$ be a shortest path in $G_{A_{1,2}}$
linking $\tilde{l}$ to $k_1$
as represented in Figure~\ref{figIf|A1|=1,ThenA1AndA2AreInDistinctBlocksOfPmin(B)b-ter}.
Let~$e'$ be the edge of $\gamma'$ incident to $k_1$.
As $w_1 < M = w_5$,
the Star condition applied to $\lbrace e_1, e_5, e' \rbrace$ implies $w(e') = M$.
Then,
by the Cycle condition applied to the cycle formed by $\lbrace \tilde{e}, e_2, e_1 \rbrace \cup \gamma'$,
we get $w(\tilde{e}) = M$ and $w(e) = M $ for every edge in $\gamma'$.
Finally,
by the Reinforced pan condition applied to the pan formed by $\lbrace \tilde{e}, e_3, e_4, e_5 \rbrace \cup \gamma'$ and $e_1$,
there is an edge linking $j_2$ to $j_1$ or~$k_1$,
a contradiction.
Hence,
we necessarily have $p=1$
and $A_{1,1}$ is the unique block of $\mathcal{P}_{\min}(A_1)$ linked to $i$.
As $j_1 \in A_{1,1}$,
$A_{1,1}$ and $A_2$ are in distinct blocks of $\mathcal{P}_{\min}(B)$ by the previous reasoning.

\textbf{\ref{itemPropA1A2AndA3AreAllInDistinctBlocksOfPmin(B)-bis}.}
By contradiction,
let $A_3$ be a third connected component of $A$ satisfying $\sigma(A_3,i) = M$.
By Claim~\ref{itemPropsIf|A1|=1AndIfSigma(B)<M,Then|A2|=1a-bis},
we have $|A_3| = 1$.
Moreover,
there exists a unique edge $e_1$ in $\Sigma(B) = \Sigma(A_1)$
and setting $e_1 = \lbrace j_1, k_1 \rbrace$
with $j_1$ and $k_1$ in $A_1$,
and $A_2 = \lbrace j_2 \rbrace$
(resp. $A_3 = \lbrace j_3 \rbrace$),
there exists a vertex $k_2$
(resp. $k_3$)
in $B$ and a cycle $C_5 = \lbrace k_1, e_1, j_1, e_2, i, e_3, j_2, e_4, k_2, e_5, k_1 \rbrace$
(resp. $\tilde{C}_5 = \lbrace k_1, e_1, j_1, e_2, i, \tilde{e}_3, j_3, \tilde{e}_4, k_3, \tilde{e}_5, k_1 \rbrace$)
with $w_1 = \sigma(A_1)$ and $w_2 = w_3 = \cdots = w_5 = M$
(resp. $w_2 = w(\tilde{e}_3) = \cdots = w(\tilde{e}_5) = M$).
We may have $k_2 = k_3$ and $e_5 = \tilde{e}_5$.
Let us assume $k_2 \not= k_3$ as represented in Figure~\ref{figIf|A1|=1,ThenA1AndA2AreInDistinctBlocksOfPmin(B)-5b}.
\begin{figure}[!h]
\centering
\subfloat[$k_2 = k_3$.]{
\begin{pspicture}(-.5,-.8)(.5,2)
\tiny
\begin{psmatrix}[mnode=circle,colsep=0.55,rowsep=0.35]
 & {$i$}\\
{$j_1$}	& & {$j_2$} & {$j_3$}\\
{$k_1$} & & {$k_2$}
\psset{arrows=-, shortput=nab,labelsep={0.05}}
\tiny
\ncline{2,1}{1,2}^{$M$}_{$e_2$}
\ncline{2,1}{3,1}_{$w_1$}^{$e_1$}
\ncline{1,2}{2,3}_{$e_3$}^{$M$}
\ncline{2,3}{3,3}_{$e_4$}^{$M$}
\ncline{3,1}{3,3}_{$M$}^{$e_5$}
\ncarc[arcangle=30]{1,2}{2,4}_{$\tilde{e}_3$}^{$M$}
\ncarc[arcangle=30]{2,4}{3,3}_{$\tilde{e}_4$}^{$M$}
\normalsize
\uput[0](-1.1,1){\textcolor{blue}{$A_{2}$}}
\pspolygon[framearc=1,linestyle=dashed,linecolor=blue,linearc=.4]
(-.5,.65)(-.5,1.4)(.3,1.4)(.3,.65)
\uput[0](1,1.6){\textcolor{blue}{$A_{3}$}}
\pspolygon[framearc=1,linestyle=dashed,linecolor=blue,linearc=.4]
(.6,.65)(.6,1.4)(1.4,1.4)(1.4,.65)
\uput[0](-2.9,1.7){\textcolor{blue}{$A_{1}$}}
\pspolygon[framearc=1,linestyle=dashed,linecolor=blue,linearc=.3]
(-2.9,-.45)(-2.9,1.45)(-1.8,1.45)(-1.8,-.45)
\end{psmatrix}
\end{pspicture}
\label{figIf|A1|=1,ThenA1AndA2AreInDistinctBlocksOfPmin(B)-5a}
}
\subfloat[$k_2 \not= k_3$.]{
\begin{pspicture}(-1,-.8)(1,2)
\tiny
\begin{psmatrix}[mnode=circle,colsep=0.55,rowsep=0.35]
 & {$i$}\\
{$j_1$}	& & {$j_2$} & {$j_3$}\\
{$k_1$} & & {$k_2$} & {$k_3$}
\psset{arrows=-, shortput=nab,labelsep={0.05}}
\tiny
\ncline{2,1}{1,2}^{$M$}_{$e_2$}
\ncline{2,1}{3,1}_{$w_1$}^{$e_1$}
\ncline{1,2}{2,3}_{$e_3$}^{$M$}
\ncline{2,3}{3,3}_{$e_4$}^{$M$}
\ncline{3,1}{3,3}_{$M$}^{$e_5$}
\ncarc[arcangle=30]{1,2}{2,4}_{$\tilde{e}_3$}^{$M$}
\ncline{2,4}{3,4}_{$\tilde{e}_4$}^{$M$}
\ncarc[arcangle=-30]{3,1}{3,4}^{$\tilde{e}_5$}_{$M$}
\normalsize
\uput[0](-2.3,1){\textcolor{blue}{$A_{2}$}}
\pspolygon[framearc=1,linestyle=dashed,linecolor=blue,linearc=.4]
(-1.7,.65)(-1.7,1.4)(-.9,1.4)(-.9,.65)
\uput[0](-.1,1.6){\textcolor{blue}{$A_{3}$}}
\pspolygon[framearc=1,linestyle=dashed,linecolor=blue,linearc=.4]
(-.5,.65)(-.5,1.4)(.3,1.4)(.3,.65)
\uput[0](-4.1,1.7){\textcolor{blue}{$A_{1}$}}
\pspolygon[framearc=1,linestyle=dashed,linecolor=blue,linearc=.3]
(-4.1,-.45)(-4.1,1.45)(-3,1.45)(-3,-.45)
\end{psmatrix}
\end{pspicture}
\label{figIf|A1|=1,ThenA1AndA2AreInDistinctBlocksOfPmin(B)-5b}
}
\caption{$w_1 = \sigma(A_1) < M$.}
\label{figIf|A1|=1,ThenA1AndA2AreInDistinctBlocksOfPmin(B)-5}
\end{figure}
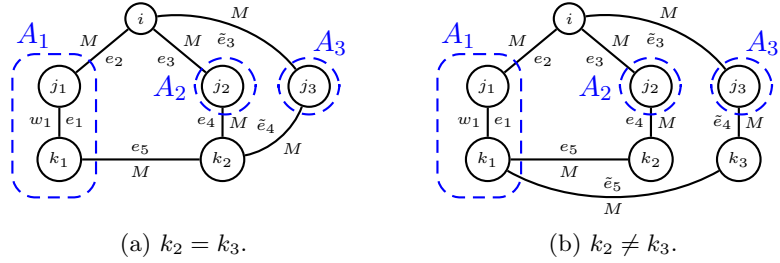
As $w_1 < M$,
the Reinforced pan condition applied to
$\lbrace e_5, e_4, e_3, \tilde{e}_3, \tilde{e}_4, \tilde{e}_5 \rbrace$
and $e_1$
implies that $j_2$ is linked to $j_1$ or $k_1$,
a contradiction.
Let us now assume $k_2 = k_3$ as represented in Figure~\ref{figIf|A1|=1,ThenA1AndA2AreInDistinctBlocksOfPmin(B)-5a}.
$\lbrace j_1, j_2 \rbrace$ and $\lbrace k_1, j_2 \rbrace$
(resp. $\lbrace j_1, j_3 \rbrace$  and $\lbrace k_1, j_3 \rbrace$)
cannot exist,
otherwise $A_1$ and $A_2$
(resp. $A_3$)
are not distinct components.
Then,
by the Reinforced adjacent cycles condition,
there exists an edge linking $j_2$ to $j_3$,
a contradiction.
\end{proof}

\begin{lemma}
\label{Lemmapartitionin2set2}
Let us assume that the Path, Star, Cycle, Pan, Adjacent cycles,
Reinforced cycle, Reinforced pan,
and Reinforced adjacent cycles conditions are satisfied.
Let us consider $i \in N$
and $A \subseteq B \subseteq N \setminus \lbrace i \rbrace$
with $A \cup \lbrace i \rbrace$ and $B$ in $\mathcal{F}$.
Let us assume $A \notin \mathcal{F}$,
and let $A_1, A_2, \ldots, A_p$ with $p \geq 2$ be the connected components of $A$.
We set $K_A = \lbrace 2, \ldots, p \rbrace$.
Let us assume that one of the following conditions is satisfied:
\begin{enumerate}
\item
$|A_1| = 1$ and $\sigma(B) \leq \sigma(A_1, i) < \sigma(A_2, i) \leq \cdots \leq \sigma(A_p,i)$.
\item
$|A_1| \geq 2$ and either 
\begin{equation}
\label{eqSigma(B)<=Sigma(A1)<=Sigma(A1,i)<Sigma(A2,i)<=Sigma(Ap,i)}
\sigma(B) \leq \sigma(A_1) \leq \sigma(A_1, i) < \sigma(A_2, i) \leq \cdots \leq \sigma(A_p,i),
\end{equation}
or
\begin{equation}
\label{eqSigma(B)<=Sigma(A1)<Sigma(A1,i)=Sigma(A2,i)=Sigma(Ap,i)}
\sigma(B) \leq \sigma(A_1) < \sigma(A_1, i) = \sigma(A_2, i) = \cdots = \sigma(A_p,i).
\end{equation}
\end{enumerate}
Let $A_{1,1}, A_{1,2}, \ldots, A_{1,k}$
(resp. $B_{1}, B_{2}, \ldots, B_{q}$)
be the blocks of $\mathcal{P}_{\min}(A_1)$
(resp. $\mathcal{P}_{\min}(B)$).
Let $J_{A_1}$
(resp. $J_B$)
be the set of indices $j \in \lbrace 1, \ldots, k \rbrace$
(resp. $j \in \lbrace 1, \ldots, q \rbrace$)
such that $A_{1,j}$
(resp. $B_j$)
is linked to $i$ by an edge $e$ in $E(A_1,i)$
(resp. $E(B,i)$)
with weight $w(e) > \sigma(A_1)$
(resp. $w(e) > \sigma(B)$).
We set $J_{A_1} = \emptyset$ if $|A_1| = 1$.
Then,
we have
\begin{equation}
\label{itemA1jSubseteqBjForAlljWith1<=j<=k}
A_{1,j} \subseteq B_j,\, \forall j \in J_{A_1},
\end{equation}
\begin{equation}
\label{itemAl=1AndAlSubseteqBk+l-1ForAlllWith2<=l<=p}
|A_j| = 1  \textrm{ and } A_j \subseteq B_{j}, \, \forall j \in K_A,
\end{equation}
with $J_{A_1} \cup K_A \subseteq J_B$ and $J_{A_1} \cap K_A = \emptyset$,
after renumbering if necessary.
Moreover,
if (\ref{eqSigma(B)<=Sigma(A1)<Sigma(A1,i)=Sigma(A2,i)=Sigma(Ap,i)}) is satisfied,
then $|J_{A_1}|=1$ and $|K_A|=1$.
\end{lemma}

\begin{proof}
Let us first assume $|A_1| = 1$.
Then,
$J_{A_1} = \emptyset$ and (\ref{itemA1jSubseteqBjForAlljWith1<=j<=k}) is satisfied.
Claims~\ref{itemPropsIf|A1|=1AndIfSigma(B)<M,Then|A2|=1a}
and~\ref{itemPropA1A2AndA3AreAllInDistinctBlocksOfPmin(B)} of Proposition~\ref{Propositionpartitionin2set1} applied to
$A = A_1 \cup \bigcup_{k \in K_A} A_k$ and $B$
imply~(\ref{itemAl=1AndAlSubseteqBk+l-1ForAlllWith2<=l<=p}).
Let us now assume $|A_1| \geq 2$.
As $\sigma(A_1,i) \geq \sigma(A_1)$,
we have $\sigma(A_1 \cup \lbrace i \rbrace) = \sigma(A_1)$.
Then $A' = \bigcup_{j \in J_{A_1}} A_{1,j}$
is a block of $\mathcal{P}_{\min}(A_1 \cup \lbrace i \rbrace)_{|A_1}$
and $\mathcal{P}_{\min}(A_1)_{|A'} = \bigcup_{j \in J_{A_1}} \lbrace A_{1,j} \rbrace$.
By Theorem~\ref{thPathBranchCyclePanAdjacentCond},
there is inheritance of $\mathcal{F}$-convexity for super\-additive games
with the correspondence~$\mathcal{P}_{\min}$.
In particular, 
there is inheritance of $\mathcal{F}$-convexity for unanimity games.
By~Corollary~\ref{corIGNEiafttsihfv},
there is inheritance of superadditivity
with the correspondence~$\mathcal{P}_{\min}$.
Then,
Theorem~\ref{thLGNuPNFNFSNuSSNNuSFABFABF} implies
$\mathcal{P}_{\min}(A_1)_{|A'} = \mathcal{P}_{\min}(B)_{|A'}$
and therefore (\ref{itemA1jSubseteqBjForAlljWith1<=j<=k}) is satisfied.
If (\ref{eqSigma(B)<=Sigma(A1)<=Sigma(A1,i)<Sigma(A2,i)<=Sigma(Ap,i)}) is satisfied,
then Claims~\ref{itemPropsIf|A1|=1AndIfSigma(B)<M,Then|A2|=1a},
\ref{itemPropA1AndA2AreInDistinctBlocksOfPmin(B)},
and~\ref{itemPropA1A2AndA3AreAllInDistinctBlocksOfPmin(B)} of Proposition~\ref{Propositionpartitionin2set1}
imply~(\ref{itemAl=1AndAlSubseteqBk+l-1ForAlllWith2<=l<=p}).
If~(\ref{eqSigma(B)<=Sigma(A1)<Sigma(A1,i)=Sigma(A2,i)=Sigma(Ap,i)}) is satisfied,
then Claims~\ref{itemPropsIf|A1|=1AndIfSigma(B)<M,Then|A2|=1a-bis},
\ref{itemPropA1AndA2AreInDistinctBlocksOfPmin(B)-bis},
and~\ref{itemPropA1A2AndA3AreAllInDistinctBlocksOfPmin(B)-bis} of Proposition~\ref{Propositionpartitionin2set1-bis}
imply~(\ref{itemAl=1AndAlSubseteqBk+l-1ForAlllWith2<=l<=p})
and $|J_{A_1}|=|K_A|=1$.
\end{proof}

\begin{lemma}
\label{lemmaPathStarCyclePanSum_FInP_Min(B)v(ACapF)=Sum_k=1^mv(A_k)=v(A)-1}
Let us assume that the Path, Star, and Cycle conditions are satisfied.
Let us consider $i \in N$
and $A \subseteq B \subseteq N \setminus \lbrace i \rbrace$
with $A \cup \lbrace i \rbrace$ and $B$ in $\mathcal{F}$.
Let $A_1, A_2, \ldots, A_p$ with $p \geq 1$ be the connected components of $A$.
Let us assume $\sigma(A_k, i) \leq \sigma(B)$ for all $k$, $1 \leq k \leq p$,
and $ \sigma(A_k)=\sigma(B)$
for all $k$ with $|A_k| \geq 2$.
Then,
each block of $\mathcal{P}_{\min}(B)$ meets at most one subset $A_k$ of $A$.
\end{lemma}

\begin{proof}
If $p = 1$,
then the result is trivially satisfied.
Let us assume $p \geq 2$.

Let us consider a connected component $A_k$ with $1 \leq k \leq p$.
Let $e_1 = \lbrace i, j \rbrace$
be an edge in $\Sigma(A_k,i)$.
Let $e_m$ be an edge in $\Sigma(B)$
and let $\gamma$ be a shortest path in $G_{B}$
linking $e_{m}$ to $j$
as represented in Figure~\ref{figProofLemmaImplicationPathStarConditions}.
\begin{figure}[!h]
\centering
\begin{pspicture}(0,-.3)(0,1.2)
\tiny
\begin{psmatrix}[mnode=circle,colsep=0.4,rowsep=0]
{$i$} & & {}\\
& {$j$}  &   {} &  {} & {}
\psset{arrows=-, shortput=nab,labelsep={0.1}}
\tiny
\ncline{2,2}{1,3}^{$e'_1$}
\ncline{1,1}{2,2}_{$e_1$}
\ncline{2,2}{2,3}
\ncline{2,3}{2,4}
\ncline{2,4}{2,5}_{$e_{m}$}
\normalsize
\uput[0](.3,.2){\textcolor{blue}{$B$}}
\pspolygon[framearc=1,linestyle=dashed,linecolor=blue,linearc=.4]
(-2.5,-.4)(-2.5,.8)(.3,.8)(.3,-.4)
\end{psmatrix}
\end{pspicture}
\caption{$w_1 = \sigma(A_k, i) \leq \sigma(B) = w_{m}$.}
\label{figProofLemmaImplicationPathStarConditions}
\end{figure}
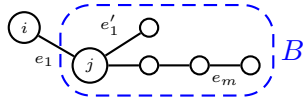
We select $e_m$ such that $\gamma$ is as short as possible.
Then,
$\gamma$ necessarily reduces to $j$,
otherwise the Path condition applied to $\lbrace e_1 \rbrace \cup \gamma \cup \lbrace e_m \rbrace$
implies $w(e) \leq \max(w_1, w_m) = \max(\sigma(A_k,i),\sigma(B)) = \sigma(B)$
and therefore $w(e) = \sigma(B)$
for any edge $e$ in $\gamma$,
contradicting the choice of~$e_m$.
If there exists an edge $e'_1 \not= e_m$ in $E(B)$
incident to $j$,
then the Star condition applied to $\lbrace e_{1}, e'_1, e_m\rbrace$
implies $w(e'_1) \leq w_m = \sigma(B)$.
We get that
any edge in $E(B)$ incident to $j$ has weight $\sigma(B)$.
In particular,
$\lbrace j \rbrace$ corresponds to a block of $\mathcal{P}_{\min}(B)$.
Thus,
if $|A_k| = 1$,
then $A_k = \lbrace j \rbrace$ is a block of $\mathcal{P}_{\min}(B)$.

Let us now consider two connected components $A_j$ and $A_k$
with $1 \leq j < k \leq p$, 
and let us assume that there is 
a block $F \in \mathcal{P}_{\min}(B)$
such that $F \cap A_j \not= \emptyset$ and $F \cap A_k \not= \emptyset$.
Let $e_1 = \lbrace i, j_1 \rbrace$
(resp. $e_2 = \lbrace i, j_2 \rbrace$)
be an edge in $\Sigma(A_j,i)$
(resp. $\Sigma(A_k,i)$).
By the previous reasoning,
we necessarily have $|A_j| \geq 2$ and $|A_k| \geq 2$.
Moreover,
any edge in $E(B)$
incident to $j_1$
(resp. $j_2$)
has weight $\sigma(B)$
and $\lbrace j_1 \rbrace$
(resp. $\lbrace j_2 \rbrace$)
is a block of $\mathcal{P}_{\min}(B)$
distinct from $F$.
Let $k_1$
(resp. $k_2$)
be a vertex of $A_j \cap F$
(resp. $A_k \cap F$)
and let $\gamma_1$
(resp. $\gamma_2$)
be a shortest path 
linking $j_1$ to $k_1$ in $G_{A_j}$
(resp. $j_2$ to $k_2$ in $G_{A_k}$).
We select $k_1$ and $k_2$
such that $\gamma_1$ and $\gamma_2$ are as short as possible.
As $F \in \mathcal{P}_{\min}(B)$
and as $\lbrace j_1 \rbrace$ and $\lbrace j_2 \rbrace$
are also blocks of $\mathcal{P}_{\min}(B)$,
$\gamma_1$
(resp. $\gamma_2$)
contains at least one edge.
Let $\tilde{e}_1$
(resp. $\tilde{e}_2$)
be the edge of $\gamma_1$
(resp. $\gamma_2$)
incident to $j_1$
(resp. $j_2$).
As $F \in \mathcal{P}_{\min}(B)$,
there exists a path $\tilde{\gamma}$ in $G_F$
from $k_1$ to $k_2$
with $w(e) > \sigma(B)$ for each edge $e$ in $\tilde{\gamma}$.
We select $\tilde{\gamma}$ as short as possible.
We obtain a simple cycle
$C_m = \lbrace e_1 \rbrace \cup \gamma_1 \cup \tilde \gamma \cup \gamma_2 \cup \lbrace e_2 \rbrace$
as represented in Figure~\ref{figcycleu1<u2andu1<u3=um-1=M-2-2}.
\begin{figure}[!h]
\centering
\begin{pspicture}(0,-.3)(0,2.2)
\tiny
\begin{psmatrix}[mnode=circle,colsep=0.6,rowsep=-.05]
	& {$j_1$}	& {$k_1$}  	&  {}\\
	&		&	&	&  {} \\
{$i$}\\
	& 		&	&  	&{}\\
	& {$j_2$} 	& {}	&  {$k_2$}
\psset{arrows=-, shortput=nab,labelsep={0.02}}
\tiny
\ncline{3,1}{1,2}^{$e_{1}$}
\ncline{3,1}{5,2}_{$e_{2}$}
\ncline{1,2}{1,3}^{$\tilde{e}_1$}
\ncline{1,3}{1,4}
\ncline{5,2}{5,3}_{$\tilde{e}_2$}
\ncline{5,3}{5,4}
\ncline{1,4}{2,5}
\ncline{5,4}{4,5}
\ncline{2,5}{4,5}
\normalsize
\uput[0](-2.2,.9){\textcolor{blue}{$A_{j}$}}
\pspolygon[framearc=1,linestyle=dashed,linecolor=blue,linearc=.4]
(-2.9,1.2)(-2.9,2)(-.8,2)(-.8,1.2)
\uput[0](-.8,.8){\textcolor{blue}{$A_{k}$}}
\pspolygon[framearc=1,linestyle=dashed,linecolor=blue,linearc=.4]
(-2.9,-.3)(-2.9,.5)(.3,.5)(.3,-.3)
\uput[0](1.7,.9){\textcolor{blue}{$F$}}
\pspolygon[framearc=1,linestyle=dashed,linecolor=blue,linearc=.4]
(-1.7,-.4)(-1.7,.8)(-1.7,2.1)(.4,2.1)(1.6,1.4)(1.6,.2)(.4,-.4)
\end{psmatrix}
\end{pspicture}
\caption{$C_m = \lbrace e_1 \rbrace \cup \gamma_j \cup \tilde \gamma \cup \gamma_k \cup \lbrace e_2 \rbrace$.}
\label{figcycleu1<u2andu1<u3=um-1=M-2-2}
\end{figure}
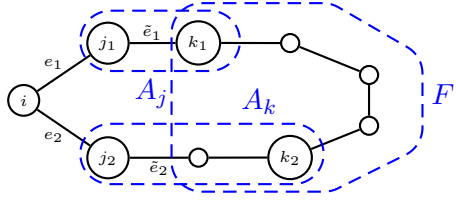
Then,
$e_1$, $e_2$, $\tilde{e}_1$, and $\tilde{e}_2$
are four edges with weights strictly smaller than the edges of $\tilde{\gamma}$,
contradicting the Cycle condition.
Therefore,
each block $F \in \mathcal{P}_{\min}(B)$
meets at most one subset $A_k$ of $A$.
\end{proof}

\begin{lemma}
\label{lemmaPathStarCyclePanSum_FInP_Min(B)v(ACapF)=Sum_k=1^mv(A_k)=v(A)}
Let us assume that the Path, Star, Cycle, Pan, Reinforced pan,
and Reinforced cycle conditions are satisfied.
Let us consider $i \in N$
and $A \subseteq B \subseteq N \setminus \lbrace i \rbrace$
with $A \cup \lbrace i \rbrace$ and $B$ in $\mathcal{F}$.
Let $A_1, A_2, \ldots, A_p$ with $p \geq 2$ be the connected components of $A$.
Let us assume that one of the following conditions is satisfied
\begin{equation}
\label{eqSigma(Ak,i)=M<Sigma(B),Forallk,1<=k<=p}
\sigma(A_k,i) = M < \sigma(B),\, \forall k,\, 1 \leq k \leq p,
\end{equation}
or
\begin{equation}
\label{eqSigma(A1,i)<Sigma(B)AndSigma(A1,i)<Sigma(Ak,i)=M,Forallk,2<=k<=p}
\sigma(A_1,i) < \sigma(B) \textrm{ and } \sigma(A_1,i) < \sigma(A_k,i) = M, \, \forall k,\,  2 \leq k \leq p.
\end{equation}
Then,
for every convex game $(N,v)$, we have
\begin{equation}
\label{maxlbracksigma(A1,i)2}
\sum_{F \in \mathcal {P}_{\min}(B)_{|A}}\; v(F)=
\sum_{k=1}^{p}\; \overline{v}(A_k).
\end{equation}
\end{lemma}

\begin{proof}
Let us first assume (\ref{eqSigma(Ak,i)=M<Sigma(B),Forallk,1<=k<=p}) satisfied.
Then,
Lemma~\ref{lemminsAi>=sA>=sBorsA=sB>sAi} implies
\begin{equation}
\label{eqsigma(A1,i)<Sigma(Ak)=Sigma(B)}
\sigma(A_k,i) < \sigma(A_k) = \sigma(B),\, \forall k,\, 1 \leq k \leq p,\, \textrm{ s.t. } |A_k| \geq 2.
\end{equation}
By~(\ref{eqSigma(Ak,i)=M<Sigma(B),Forallk,1<=k<=p}),
(\ref{eqsigma(A1,i)<Sigma(Ak)=Sigma(B)}),
and Lemma~\ref{lemmaPathStarCyclePanSum_FInP_Min(B)v(ACapF)=Sum_k=1^mv(A_k)=v(A)-1},
we get
\begin{equation}
\label{eqPmin(B)|A=Pmin(B)|A1,Pmin(B)|A2,Pmin(B)|Ap}
\mathcal{P}_{\min}(B)_{|A}=
\lbrace \mathcal{P}_{\min}(B)_{|A_1}, \mathcal{P}_{\min}(B)_{|A_2}, \ldots, \mathcal{P}_{\min}(B)_{|A_p}\rbrace.
\end{equation}
Let us consider a subset $A_k$ with $1 \leq k \leq p$.
As $\sigma(A_1,i) < \sigma(B)$,
there exists an edge $e_0$ in $E(A_1,i)$ connected to $B$
with weight $w(e_0) < \sigma(B)$.
If $|A_k| \geq 2$,
then Lemma~\ref{lemPminvB-vB>=vA-vA2} applied to $A_k \subseteq B$
implies
$\mathcal{P}_{\min}(B)_{|A_k} = \mathcal{P}_{\min}(A_k)$.
If $|A_k| = 1$,
then we obviously have $\mathcal{P}_{\min}(B)_{|A_k} = \mathcal{P}_{\min}(A_k)$.
Therefore,
(\ref{eqPmin(B)|A=Pmin(B)|A1,Pmin(B)|A2,Pmin(B)|Ap}) implies (\ref{maxlbracksigma(A1,i)2}).

Let us now assume (\ref{eqSigma(A1,i)<Sigma(B)AndSigma(A1,i)<Sigma(Ak,i)=M,Forallk,2<=k<=p}) satisfied.
As $\sigma(A_1,i) < \sigma(B)$,
the previous reasoning applied to $A_1 \subseteq B$ gives
\begin{equation}
\label{eqSumFInPMin(B)|A1v(F)=Overlinev(A1)-1}
\sum_{F \in \mathcal{P}_{\min}(B)_{|A_1}} v(F) = \overline{v}(A_1).
\end{equation}

Let us assume $\sigma(B) < M$.
By~(\ref{eqSigma(A1,i)<Sigma(B)AndSigma(A1,i)<Sigma(Ak,i)=M,Forallk,2<=k<=p}),
we get
\begin{equation}
\label{eqSigma(A_1,i)<Sigma(B1)<Sigma(Ak,i),Forallk,2<=k<=p-1}
\sigma(A_1,i) < \sigma(B) < M = \sigma(A_k,i),\, \forall k,\, 2 \leq k \leq p.
\end{equation}
Let $K_A$ be the set of indices $k \in \lbrace 2, \ldots, p \rbrace$.
By~(\ref{eqSigma(A_1,i)<Sigma(B1)<Sigma(Ak,i),Forallk,2<=k<=p-1}),
Claims~\ref{itemPropsIf|A1|=1AndIfSigma(B)<M,Then|A2|=1a},
\ref{itemPropA1AndA2AreInDistinctBlocksOfPmin(B)},
and \ref{itemPropA1A2AndA3AreAllInDistinctBlocksOfPmin(B)} of Proposition~\ref{Propositionpartitionin2set1} 
applied to $A_1 \cup \bigcup_{k \in K_A} A_k$ and $B$ imply
\begin{equation}
\label{eq|Ak|=1ForAllkInKA-1}
|A_k| = 1 \textrm{ for all } k \in K_A,
\end{equation}
and each subset $A_k$ is included in a distinct block of $\mathcal{P}_{\min}(B)$.
Moreover,
the blocks of $\mathcal{P}_{\min}(A_1)$ belong to blocks of $\mathcal{P}_{\min}(B)$
different from the ones containing the subsets $A_k$ with $k \in K_A$.
In particular,
this implies that $A_1$ has always an empty intersection with 
the block of $\mathcal{P}_{\min}(B)$ containing a subset $A_k$ with $k \in K_A$.
We get
\begin{equation}
\label{eqProofSumFInPMin(Bl)TildeA1v(F)=SumFInPMin(B1)A1v(F)+Sumk=2pv(Ak)-1}
\sum_{F \in \mathcal{P}_{\min}(B)_{|A}} v(F)=
\sum_{F \in \mathcal{P}_{\min}(B)_{|A_1}} v(F) +
\sum_{k \in K_A} v(A_k).
\end{equation}
By~(\ref{eq|Ak|=1ForAllkInKA-1}),
we have $v(A_k) = \overline{v}(A_k)$ for all $k \in K_A$.
Then,
(\ref{eqSumFInPMin(B)|A1v(F)=Overlinev(A1)-1})
and~(\ref{eqProofSumFInPMin(Bl)TildeA1v(F)=SumFInPMin(B1)A1v(F)+Sumk=2pv(Ak)-1}) give~(\ref{maxlbracksigma(A1,i)2}).

Let us now assume $M \leq \sigma(B)$.
By~(\ref{eqSigma(A1,i)<Sigma(B)AndSigma(A1,i)<Sigma(Ak,i)=M,Forallk,2<=k<=p}),
we get
\begin{equation}
\label{eqSigma(A1,i)<Sigma(Ak,i)<=Sigma(B1),Forallk,2<=k<=p-1}
\sigma(A_1,i) < \sigma(A_k,i) \leq M \leq \sigma(B),\, \forall k,\, 2 \leq k \leq p.
\end{equation}
(\ref{eqSigma(A1,i)<Sigma(Ak,i)<=Sigma(B1),Forallk,2<=k<=p-1})
and Lemma~\ref{lemminsAi>=sA>=sBorsA=sB>sAi} imply
\begin{equation}
\label{eqSigma(Ak,i)<=Sigma(Ak)=Sigma(B1),ForallAkSubseteqB1With|Ak|>=2-1}
\sigma(A_k,i) \leq \sigma(A_k) = \sigma(B),\, \forall k, \, 1 \leq k \leq p, \textrm{ s.t. } |A_k| \geq 2.
\end{equation}
Then,
(\ref{eqSigma(A1,i)<Sigma(Ak,i)<=Sigma(B1),Forallk,2<=k<=p-1}),
(\ref{eqSigma(Ak,i)<=Sigma(Ak)=Sigma(B1),ForallAkSubseteqB1With|Ak|>=2-1}),
and Lemma~\ref{lemmaPathStarCyclePanSum_FInP_Min(B)v(ACapF)=Sum_k=1^mv(A_k)=v(A)-1}
applied to $\bigcup_{k=1}^{p} A_k$ and $B$
imply (\ref{eqPmin(B)|A=Pmin(B)|A1,Pmin(B)|A2,Pmin(B)|Ap})
and we can conclude as in the first case.
\end{proof}

For any $i \in N$ and for any subset $A \subseteq N \setminus \lbrace i \rbrace$,
we define
$M(A,i)=\max_{e \in E(A,i)}\; w(e)$.
\begin{lemma}
\label{Ifsigma(B)<igma(B1,thenforall4}
Let us assume that the Star and Path conditions are satisfied.
Let us consider $i \in N$
and $A \subseteq N \setminus \lbrace i \rbrace$
with $A \notin \mathcal{F}$ but $A \cup \lbrace i \rbrace$ in $\mathcal{F}$.
Let $A_1, A_2,\ldots, A_p$ be the connected components of $A$.
We assume
\begin{equation}
\label{eqLemmaSigma(B,i)=Sigma(B1,i)<=Sigma(B2,i)=Sigma(Bq,i)=M(B,i)}
\sigma(A,i) = \sigma(A_1,i) \leq \sigma(A_2,i) = \ldots = \sigma(A_p,i) = M(A,i).
\end{equation}
\begin{enumerate}
\item
\label{itemLemmaIfSigma(B,i)<M(B,i)ThenM(B,i)<=Sigma(B2)<=Sigma(B3)}
If $\sigma(A,i) < M(A,i)$,
then we have
\begin{equation}
\label{eqLemmaM<=sigma(B2)<=sigma(B3)<=sigma(Bq)}
M(A,i) \leq \sigma(A_k),\, \forall k,\, 2 \leq k \leq p, \textrm{ s.t. } |A_k| \geq 2.
\end{equation}
\item
\label{itemLemmaIfSigma(B,i)=M(B,i)And|B1|>=2,Thenmax(Sigma(B1),M(B,i))<=Sigma(B2)}
If $\sigma(A,i) = M(A,i)$ and if $|A_k| \geq 2$ for at least one index $k$ with $1 \leq k \leq p$,
then we have
\begin{equation}
\label{eqLemmaM<=sigma(B2)<=sigma(B3)<=sigma(Bq)-2}
\max(\sigma(A_1),M(A,i)) \leq \sigma(A_k),\, \forall k,\, 2 \leq k \leq p, \textrm{ s.t. } |A_k| \geq 2,
\end{equation}
after renumbering if necessary.
\end{enumerate}
\end{lemma}

\begin{proof}
Let us first assume $\sigma(A,i) < M(A,i)$.
By the Star condition,
there is a unique edge $e_0$ in $E(A,i)$
with $w(e_0) = \sigma(A,i)$
and $w(e) = M(A,i)$
for all $e \in E(A,i) \setminus \lbrace e_0 \rbrace$.
By~(\ref{eqLemmaSigma(B,i)=Sigma(B1,i)<=Sigma(B2,i)=Sigma(Bq,i)=M(B,i)}),
we have $e_0 \in E(A_1,i)$.
Let us consider a subset $A_k$ with $2 \leq k \leq p$ and $|A_k| \geq 2$.
Let $\tilde e = \lbrace i, j_k \rbrace$ be an edge in $E(A_k,i)$,
and let $e_k$ be an edge in $\Sigma(A_k)$.
Let $\gamma_k$ be a path in $G_{A_k}$ linking $e_k$ to $j_k$
as represented in Figure~\ref{figcycleu1<u2andu1<u3=um-1=M-2-2-3}
($\gamma_k$ may be reduced to $j_k$).
\begin{figure}[!h]
\centering
\begin{pspicture}(0,-.3)(0,1.8)
\tiny
\begin{psmatrix}[mnode=circle,colsep=0.55,rowsep=0.1]
	& {$j_1$} \\
{$i$}\\
	& {$j_k$} & {} & {}
\psset{arrows=-, shortput=nab,labelsep={0.05}}
\tiny
\ncline{2,1}{1,2}_{$e_{0}$}^{$\sigma(A,i)$}
\ncline{2,1}{3,2}^{$\tilde{e}$}_{$M(A,i)$}
\ncline[linestyle=dotted]{3,2}{3,3}
\ncline{3,3}{3,4}^{$e_k$}_{$\sigma(A_k)$}
\normalsize
\uput[0](.4,1.3){\textcolor{blue}{$A_{1}$}}
\pspolygon[framearc=1,linestyle=dashed,linecolor=blue,linearc=.4]
(-2.2,.9)(-2.2,1.7)(.3,1.7)(.3,.9)
\uput[0](.4,0){\textcolor{blue}{$A_{k}$}}
\pspolygon[framearc=1,linestyle=dashed,linecolor=blue,linearc=.4]
(-2.2,-.4)(-2.2,.4)(.3,.4)(.3,-.4)
\end{psmatrix}
\end{pspicture}
\caption{Path induced by $e_0 \cup \tilde{e} \cup \gamma_k \cup e_k$.}
\label{figcycleu1<u2andu1<u3=um-1=M-2-2-3}
\end{figure}
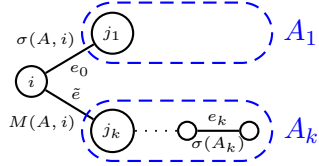
The Path condition applied to $e_0 \cup \tilde{e} \cup \gamma_k \cup e_k$ implies
$w(\tilde{e}) \leq \max(w(e_0),w(e_k))$.
As $w(e_0) = \sigma(A,i) < M(A,i) = w(\tilde{e})$
and as $w(e_k) = \sigma(A_k)$,
we get 
$M(A,i) \leq \sigma(A_k)$.
This implies (\ref{eqLemmaM<=sigma(B2)<=sigma(B3)<=sigma(Bq)}).

Let us now assume $\sigma(A,i) = M(A,i)$.
By~(\ref{eqLemmaSigma(B,i)=Sigma(B1,i)<=Sigma(B2,i)=Sigma(Bq,i)=M(B,i)}),
we get $\sigma(A_k,i) = M(A,i)$ for all~$k$, $1 \leq k \leq p$.
We can assume $\sigma(A_1) \leq \sigma(A_k)$ for all $k$ with $|A_k| \geq 2$,
after renumbering if necessary.
Let us consider a subset $A_k$ with $2 \leq k \leq p$ and $|A_k| \geq 2$.
Let $e_1$ be an edge in $\Sigma(A_1)$
and let $e_k$ be an edge in $\Sigma(A_k)$.
There is a path linking $e_1$ to $e_k$ and passing through $i$.
As the edges incident to $i$ have weight $M(A,i)$,
the Path condition implies $M(A,i) \leq \max(\sigma(A_1), \sigma(A_k))$
and (\ref{eqLemmaM<=sigma(B2)<=sigma(B3)<=sigma(Bq)-2}) is satisfied.
\end{proof}

\begin{proof}[Proof of Theorem~\ref{thADDLG=NEbacfgaluEbawfvhat}]
Let us consider a convex game $(N,v)$,
$i \in N$ and subsets $A \subseteq B \subseteq N \setminus \lbrace i \rbrace$.
We have to prove that the following inequality is satisfied:
\begin{equation}
\label{maxlbracksigma(A1,i)10hat}
\Delta_i \hat{v}(B) \geq
\Delta_i \hat{v}(A).
\end{equation}

We first show that we can w.l.o.g. make several restrictions on the sets $A$ and $B$.
By Corollary~\ref{corIGNEiafttsihfv} and Theorem~\ref{thPathBranchCyclePanAdjacentCond},
$(N, \overline{v})$ and $(N, \hat{v})$ are superadditive and $\mathcal{F}$-convex.
Let $A_1, A_2, \ldots, A_p$
(resp. $B_{1}, B_{2}, \ldots, B_{q}$)
with $p \geq 1$
(resp. $q \geq 1$)
be the connected components of 
$A$
(resp. $B$).
Let $J_A$
(resp. $\overline{J_A}$)
be the set of indices $j \in \lbrace 1, \ldots, p \rbrace$
such that $A_j$ is linked 
(resp. non-linked)
to~$i$ by an edge.
If $J_A = \emptyset$,
then $\Delta_{i} \hat{v}(A) = \hat{v}(\lbrace i \rbrace)$
and (\ref{maxlbracksigma(A1,i)10hat}) is satisfied by the superadditivity of $(N, \hat{v})$.
If $J_A \not= \emptyset$,
then setting $A' = \bigcup_{j \in J_A} A_{j}$,
we get
\begin{equation}
\label{maxlbracksigma(A1,i)4}
\Delta_i \hat{v}(A)=
\overline{v}(A' \cup \lbrace i \rbrace) - \sum_{j \in J_A} \overline{v}(A_{j})=
\hat{v}(A' \cup \lbrace i \rbrace) - \hat{v}(A')=
\Delta_i \hat{v}(A').
\end{equation}
Therefore we can substitute $A$ for $A'$
and assume $A \cup \lbrace i \rbrace$ connected.
Similarly,
we can assume $B \cup \lbrace i \rbrace$ connected.
Moreover,
we can assume $p \geq 2$,
otherwise the superadditivity and $\mathcal{F}$-convexity
of $(N, \hat{v})$
implies (\ref{maxlbracksigma(A1,i)10hat}).
We can also w.l.o.g. assume
\begin{equation}
\label{eqACapB_jNeqEmptyset,Forallj,1<=j<=q}
A \cap B_j \neq \emptyset,\;
\forall j,\, 1 \leq j \leq q.
\end{equation}
Indeed,
let us assume
$A \cap B_j \neq \emptyset$
for all $j$,
$1 \leq j \leq m$
and $A \cap B_j = \emptyset$
for all~$j$,
$m+1 \leq j \leq q$
for some index $m$ with $1 \leq m \leq q$.
Setting $B'= \bigcup_{j=1}^{m} B_j$
and $B''= \bigcup_{j=m+1}^{q} B_{j}$,
we have
\begin{equation}
\label{maxlbracksigma(A1,i)7}
\Delta_i \hat{v}(B) =
\overline{v}(B' \cup B'' \cup \lbrace i \rbrace) - \sum_{j=1}^{q} \overline{v}(B_{j}).
\end{equation}
By the superadditivity of $(N,\overline{v})$,
we have
$\overline{v}(B' \cup B'' \cup \lbrace i \rbrace) \geq
\overline{v}(B' \cup \lbrace i \rbrace) + \sum_{j=m+1}^{q} \overline{v}(B_{j})$.
Then (\ref{maxlbracksigma(A1,i)7}) implies
\begin{equation}
\label{maxlbracksigma(A1,i)9}
\Delta_i \hat{v}(B)
\geq \overline{v}(B' \cup \lbrace i \rbrace)- \sum_{j=1}^{m} \overline{v}(B_{j})=
\hat{v}(B' \cup \lbrace i \rbrace) - \hat{v}(B') = \Delta_i \hat{v}(B').
\end{equation}
By (\ref{maxlbracksigma(A1,i)9})
and as $A \subseteq B'$,
(\ref{maxlbracksigma(A1,i)10hat})
is satisfied if the following inequality is satisfied
\begin{equation}
\label{maxlbracksigma(A1,i)10}
\Delta_i \hat{v}(B')
\geq
 \Delta_i \hat{v}(A).
\end{equation}
Therefore we can replace $B$ by $B'$.

Let us first assume $|B \setminus A| = 1$.
Then,
we have $B = A \cup \lbrace j \rbrace$ with $j \in N \setminus (A \cup \lbrace i \rbrace)$
and (\ref{maxlbracksigma(A1,i)10hat}) is equivalent to 
\begin{equation}
\label{eqHatv(AUij)-Hatv(AUj)>=Hatv(AUi)-Hatv(A)}
\hat{v}(A \cup \lbrace i, j \rbrace) - \hat{v}(A \cup \lbrace j \rbrace)
\geq
\hat{v}(A \cup \lbrace i \rbrace) - \hat{v}(A).
\end{equation}
We can assume $j$ linked to $A$ by at least one edge,
otherwise (\ref{eqHatv(AUij)-Hatv(AUj)>=Hatv(AUi)-Hatv(A)}) is satisfied by superadditivity of $(N, \hat{v})$.
Let $A_1, A_2, \ldots, A_{p'}$ with $1 \leq p' \leq p$ be
the components of $A$ linked to $j$ by an edge.
Then,
the connected components of $B$ are $B_1, B_2, \ldots, B_{p-p'+1}$
with
$B_1 = \bigcup_{k=1}^{p'} A_k \cup \lbrace j \rbrace$,
and 
$B_l = A_{p'+l-1}$
for all $l$, $2 \leq l \leq p-p'+1$.
(\ref{eqHatv(AUij)-Hatv(AUj)>=Hatv(AUi)-Hatv(A)}) is equivalent to
\begin{equation}
\label{eqOverlinev(BUi)-Overlinev(B1)>=Overlinev(AUi)-Sumk=1p'Overlinev(Ak)}
\overline{v}(B \cup \lbrace i \rbrace) - \overline{v}(B_1)
\geq
\overline{v}(A \cup \lbrace i \rbrace) - \sum_{k=1}^{p'} \overline{v}(A_k).
\end{equation}
As $\bigcup_{k=1}^{p'} A_k \cup \lbrace i \rbrace$,
$A \cup \lbrace i \rbrace$,
$B_1 \cup \lbrace i \rbrace$,
and $B \cup \lbrace i \rbrace$ are connected,
the $\mathcal{F}$-convexity of $(N,\overline{v})$ implies
\begin{equation}
\label{v(BUi)-v(B1Ui)>=v(AUi)-v}
\overline{v}(B \cup \lbrace i \rbrace) + \overline{v}\left(\bigcup_{k=1}^{p'} A_k \cup \lbrace i \rbrace\right)
\geq
\overline{v}(A \cup \lbrace i \rbrace) + \overline{v}(B_1 \cup \lbrace i \rbrace).
\end{equation}
By~(\ref{v(BUi)-v(B1Ui)>=v(AUi)-v}),
(\ref{eqOverlinev(BUi)-Overlinev(B1)>=Overlinev(AUi)-Sumk=1p'Overlinev(Ak)}) is satisfied if
the following inequality is satisfied:
\begin{equation}
\label{eqOverlinev(B1Ui)-Overlinev(B1)>=Overlinev(Uk=1p'AkUi)-Sumk=1p'Overlinev(Ak)}
\overline{v}(B_1 \cup \lbrace i \rbrace) - \overline{v}(B_1)
\geq
\overline{v}\left(\bigcup_{k=1}^{p'} A_k \cup \lbrace i \rbrace\right) - \sum_{k=1}^{p'} \overline{v}(A_k).
\end{equation}
(\ref{eqOverlinev(B1Ui)-Overlinev(B1)>=Overlinev(Uk=1p'AkUi)-Sumk=1p'Overlinev(Ak)}) is equivalent to
\begin{equation}
\label{eqHatv(B1Ui)-Hatv(B1)>=Hatv(Uk=1p'AkUi)-Hatv(Uk=1p'Ak)}
\hat{v}(B_1 \cup \lbrace i \rbrace) - \hat{v}(B_1)
\geq
\hat{v}\left(\bigcup_{k=1}^{p'} A_k \cup \lbrace i \rbrace\right) - \hat{v} \left(\bigcup_{k=1}^{p'} A_k \right).
\end{equation}
Thus,
it is sufficient to prove that (\ref{eqHatv(B1Ui)-Hatv(B1)>=Hatv(Uk=1p'AkUi)-Hatv(Uk=1p'Ak)}) is satisfied
and we can assume each component of $A$ linked to $j$
and replace $p'$ by~$p$ and $B_1$ by $B$.

If $\sigma(B,i)< M(B,i)$,
then the Star condition implies the existence of a unique edge $e_0 \in E(B,i)$
with $w(e_0) = \sigma(B,i)$
and $w(e)=M(B,i)$ for all $e\in E(B,i) \setminus \lbrace e_0 \rbrace$.
Otherwise,
we have $w(e)=\sigma(B,i)=M(B,i)$ for all $e \in E(B,i)$.
As $p \geq 2$,
$E(A,i)$ contains at least two edges.
Hence,
we have $M(A,i) = M(B,i)$ in both cases.
Setting $M = M(A,i) = M(B,i)$,
we can always assume
\begin{equation}
\label{maxlbracksigma(A1,i)5}
\begin{gathered}
\sigma(B,i) \leq \sigma(A,i) = \sigma(A_1,i) \leq \sigma(A_2,i) = \ldots = \sigma(A_p,i) = M,
\end{gathered}
\end{equation}
after renumbering if necessary.
If $\sigma(A,i) < M$,
then we have $\sigma(A,i) = \sigma(B,i)$ and  $e_0 \in E(A_1,i)$.

We consider several cases.

\vspace{\baselineskip}
\noindent
\textbf{
Case 1
}
We assume $\sigma(A,i) < M$.
Then,
we also have $\sigma(B,i) < M$
and (\ref{maxlbracksigma(A1,i)5}) implies
\begin{equation}
\label{eqProofCase1-2Sigma(B1,i)<Sigma(Bk,i)=M(B,i)Forallk2<=k<=q}
\sigma(A,i) = \sigma(A_1,i) < \sigma(A_k,i) = M,\; \forall k,\; 2 \leq k \leq p.
\end{equation}
Moreover,
$e_0$ links $i$ to $A_1$ and
we necessarily have
\begin{equation}
\label{eqProofCase1Sigma(B,i)=Sigma(B1,i)=Sigma(A1,i)=Sigma(A,i)}
\sigma(B,i) = \sigma(A_1,i) = \sigma(A,i).
\end{equation}
As $\sigma(A,i) < M$,
Claim~\ref{itemLemmaIfSigma(B,i)<M(B,i)ThenM(B,i)<=Sigma(B2)<=Sigma(B3)} of Lemma~\ref{Ifsigma(B)<igma(B1,thenforall4} implies
\begin{equation}
\label{eqProofCase1-2Sigma(B)=Sigma(B1)<Sigma(B2)<=Sigma(B3)<=Sigma(Bq)}
M \leq  \sigma(A_k),\, \forall k,\, 2 \leq k \leq p, \textrm{ s.t. } |A_k| \geq 2.
\end{equation}

\noindent
\textbf{Subcase 1.1}
We assume
$\sigma(B) \leq \sigma(B,i)$.
Then,
we have
\begin{equation}
\label{eqProofSubcase1.1Sigma(BUi)=Sigma(B1)IfSigma(B)<Sigma(B,i),Sigma(B1)UeOOtherwise}
\Sigma(B \cup \lbrace i \rbrace) =
\left\lbrace 
\begin{array}{cl}
\Sigma(B) & \textrm{if } \sigma(B) < \sigma(B,i),\\
\Sigma(B) \cup \lbrace e_0 \rbrace & \textrm{if } \sigma(B) = \sigma(B,i).
\end{array}
\right.
\end{equation}
By~(\ref{eqProofCase1-2Sigma(B1,i)<Sigma(Bk,i)=M(B,i)Forallk2<=k<=q}) and (\ref{eqProofCase1Sigma(B,i)=Sigma(B1,i)=Sigma(A1,i)=Sigma(A,i)}),
we obtain 
\begin{equation}
\label{eqSigma(B1)<=Sigma(A1,i)<Sigma(Ak,i)=M,Forallk,2<=k<=p}
\sigma(B) \leq \sigma(A_1,i)< \sigma(A_k,i) = M,\; \forall k,\; 2 \leq k \leq p.
\end{equation}
As $E(A_1,i) \not= \emptyset$,
(\ref{eqSigma(B1)<=Sigma(A1,i)<Sigma(Ak,i)=M,Forallk,2<=k<=p})
and Lemma~\ref{lemminsAi>=sA>=sBorsA=sB>sAi} applied to $A_1 \subseteq B$ imply
\begin{equation}
\label{eqSigma(B1)<=Sigma(A1)<=Sigma(A1,i)<Sigma(Ak,i)=M,Forallk,2<=k<=p}
\sigma(B) \leq \sigma(A_1) \leq \sigma(A_1,i)< \sigma(A_k,i) = M,\, \forall k,\, 2 \leq k \leq p,\, \textrm{if}\, |A_1| \geq 2.
\end{equation}
Then,
(\ref{eqProofCase1-2Sigma(B)=Sigma(B1)<Sigma(B2)<=Sigma(B3)<=Sigma(Bq)}),
(\ref{eqSigma(B1)<=Sigma(A1,i)<Sigma(Ak,i)=M,Forallk,2<=k<=p}),
and (\ref{eqSigma(B1)<=Sigma(A1)<=Sigma(A1,i)<Sigma(Ak,i)=M,Forallk,2<=k<=p})
imply
\begin{equation}
\label{eqProofSubcase1.1Sigma(AUi)=Sigma(A1)IfSigma(A)<Sigma(A,i),Sigma(A1)UeOOtherwise}
\Sigma(A \cup \lbrace i \rbrace) =
\left\lbrace 
\begin{array}{cl}
\Sigma(A_1) & \textrm{if } |A_1| \geq 2 \textrm{ and } \sigma(A_1) < \sigma(A_1,i),\\
\Sigma(A_1) \cup \lbrace e_0 \rbrace & \textrm{if } |A_1| \geq 2 \textrm{ and }  \sigma(A_1) = \sigma(A_1,i),\\
\lbrace e_0 \rbrace & \textrm{if } |A_1| = 1.
\end{array}
\right.
\end{equation}
Let us assume $|A_1| \geq 2$.
Let $A_{1,1}, A_{1,2}, \ldots, A_{1,k_1}$
(resp. $B_{1}, B_{2}, \ldots, B_{l_1}$)
be the blocks of $\mathcal{P}_{\min}(A_1)$
(resp. $\mathcal{P}_{\min}(B)$)
connected to $i$ by an edge in $E(A_1, i)$
(resp. $E(B,i)$).
Let $J_{A_1}$
(resp. $J_{B}$)
be the set of indices $j \in \lbrace 1, \ldots, k_1 \rbrace$
(resp. $j \in \lbrace 1, \ldots, l_1 \rbrace$)
such that $E(A_{1,j},i)$
(resp.  $E(B_{j},i)$)
contains at least one edge
with weight strictly greater than $\sigma(A_1)$
(resp. $\sigma(B)$).
$J_{A_1}$
(resp. $J_B$)
may be empty.
Let $K_A$ be the set of indices $k \in \lbrace 2, \ldots, p \rbrace$.
Setting $\tilde{A} = \bigcup_{j \in J_{A_1}} A_{1,j} \cup \bigcup_{k \in K_A} A_k$
and $\tilde B = \bigcup_{j \in J_B} B_{j}$,
we have
\begin{equation}
\label{eqProofSubCase1.1Hatv(AUi)-Hatv(A)=v(TildeAUi)-Sum_i=1^l1v(A_1,i)-Sum_l=2^pOverlinev(Al)}
\Delta_i \hat{v}(A) =
v(\tilde A \cup \lbrace i \rbrace)
- \sum_{j \in J_{A_1}} v(A_{1,j})
- \sum_{k \in K_A} \overline{v}(A_k).
\end{equation}
and
\begin{equation}
\label{eqProofSubCase1.1Hatv(BUi)-Hatv(B)=v(TildeBUi)-Sum_i=1^l1v(B_1,i)-Sum_l=2^qOverlinev(Bl)}
\Delta_i \hat{v}(B) =
v(\tilde B \cup \lbrace i \rbrace) - \sum_{j \in J_B} v(B_{j}).
\end{equation}
Let us now assume $|A_1| = 1$.
Then,
we have $v(A_1) = \overline{v}(A_1) =  0$ and 
(\ref{eqProofSubCase1.1Hatv(AUi)-Hatv(A)=v(TildeAUi)-Sum_i=1^l1v(A_1,i)-Sum_l=2^pOverlinev(Al)}) is also satisfied
setting $J_{A_1} = \emptyset$.
By~(\ref{eqProofSubCase1.1Hatv(AUi)-Hatv(A)=v(TildeAUi)-Sum_i=1^l1v(A_1,i)-Sum_l=2^pOverlinev(Al)})
and (\ref{eqProofSubCase1.1Hatv(BUi)-Hatv(B)=v(TildeBUi)-Sum_i=1^l1v(B_1,i)-Sum_l=2^qOverlinev(Bl)}),
(\ref{maxlbracksigma(A1,i)10hat}) is equivalent to
\begin{multline}
\label{eqProofCase1-2v(BUi)-Sumi=1l1v(B1i)-Suml=2qv(Bl)>=v(AUi)-SumjInJv(A1j)-Suml=2qSumAkSubseteqBlv(Ak)}
v(\tilde B \cup \lbrace i \rbrace)
- \sum_{j \in J_B} v(B_{j})
\geq
v(\tilde A \cup \lbrace i \rbrace)
- \sum_{j \in J_{A_1}} v(A_{1,j})
- \sum_{k \in K_A} \overline{v}(A_k).
\end{multline}
By~(\ref{eqSigma(B1)<=Sigma(A1,i)<Sigma(Ak,i)=M,Forallk,2<=k<=p})
and (\ref{eqSigma(B1)<=Sigma(A1)<=Sigma(A1,i)<Sigma(Ak,i)=M,Forallk,2<=k<=p}),
Lemma~\ref{Lemmapartitionin2set2} applied to
$A_1 \cup \bigcup_{k \in K_A} A_k$ and $B$ implies
\begin{equation}
\label{eqA1,jSubseteqB1j,ForalljInJA}
A_{1,j} \subseteq B_{j},\, \forall j \in J_{A_1},
\end{equation}
and
\begin{equation}
\label{eq|Ak|=1AndAkSubsetqB1k,ForammkInKA}
|A_j| = 1 \textrm{ and } A_j \subseteq B_{j},\, \forall j \in K_A,
\end{equation}
with $J_{A_1} \cup K_A \subseteq J_B$ and $J_{A_1} \cap K_A = \emptyset$,
after renumbering if necessary.
(\ref{eqA1,jSubseteqB1j,ForalljInJA}) and (\ref{eq|Ak|=1AndAkSubsetqB1k,ForammkInKA}) imply
$\tilde{A} \subseteq \tilde{B}$ and
\begin{equation}
\label{eqACapB1j=A1jForalljInJA,AjForalljInKA,EmptysetForalljInOverlineKA}
\tilde A \cap B_{j}=
\left \lbrace
\begin{array}{cl}
A_{1,j} & \forall j \in J_{A_1},\\
A_j & \forall j \in K_A,\\
\emptyset & \forall j \in J_B \setminus (J_{A_1} \cup K_A).
\end{array}
\right.
\end{equation}
As $\tilde A \subseteq \tilde B$, the convexity of $(N,v)$ implies
\begin{equation}
\label{eqProofCase1-2v(TildeBUi)-v(TildeB)>=v(TildeAUi)-v(TildeA)}
v(\tilde B \cup \lbrace i \rbrace) - v(\tilde B)
\geq
v(\tilde A \cup\lbrace i \rbrace) - v (\tilde A).
\end{equation}
Applying Lemma~\ref{lemmavB+i=1pvBiA}
with $\mathcal{F}= 2^{N} \setminus \lbrace \emptyset \rbrace$
to the subsets $\tilde{A}$, $\tilde{B}$,
and to the partition 
$
\lbrace \lbrace B_{j} \mid j \in J_B \rbrace \rbrace
$
of $\tilde B$,
we obtain
\begin{equation}
\label{eqProofCase1-2v(B)-Sumj=1l1v(B_1j)-Suml=2qSum_FInPmin(Bl)v(F)>=v(A)-Sumj=1l1v(ACapB_1j)-Suml=2qSumFInPmin(Bl)v(ACapF)}
v(\tilde B)
- \sum_{j \in J_B} v(B_{j})
\geq
v(\tilde A)
- \sum_{j \in J_B} v(\tilde A \cap B_{j}).
\end{equation}
By~(\ref{eq|Ak|=1AndAkSubsetqB1k,ForammkInKA}),
we have $|A_j| = 1$ and therefore $v(A_j) = \overline{v}(A_j)$ for all $j$ in $K_A$.
Then,
(\ref{eqACapB1j=A1jForalljInJA,AjForalljInKA,EmptysetForalljInOverlineKA}) implies
\begin{equation}
\label{eqProofCase1-2Sum_j=1^l1v(TildeACapB_1,j)=Sum_jInJv(TildeA_1,j)}
\sum_{j \in J_B} v(\tilde A \cap B_{j})=
\sum_{j \in J_{A_1}} v(A_{1,j}) + \sum_{k \in K_A} \overline{v}(A_k).
\end{equation}
(\ref{eqProofCase1-2v(TildeBUi)-v(TildeB)>=v(TildeAUi)-v(TildeA)}),
(\ref{eqProofCase1-2v(B)-Sumj=1l1v(B_1j)-Suml=2qSum_FInPmin(Bl)v(F)>=v(A)-Sumj=1l1v(ACapB_1j)-Suml=2qSumFInPmin(Bl)v(ACapF)}),
and (\ref{eqProofCase1-2Sum_j=1^l1v(TildeACapB_1,j)=Sum_jInJv(TildeA_1,j)})
imply (\ref{eqProofCase1-2v(BUi)-Sumi=1l1v(B1i)-Suml=2qv(Bl)>=v(AUi)-SumjInJv(A1j)-Suml=2qSumAkSubseteqBlv(Ak)}).

\vspace{\baselineskip}
\noindent
\textbf{Subcase 1.2}
We now assume
$\sigma(B,i)<\sigma(B)$.
(\ref{eqProofCase1Sigma(B,i)=Sigma(B1,i)=Sigma(A1,i)=Sigma(A,i)}) implies
\begin{equation}
\label{eqSigma(B,i)=Sigma(B1,i)=Sigma(A1,i)=Sigma(A,i)<Sigma(B)<=Sigma(A)}
\sigma(B,i) = \sigma(A_1,i) = \sigma(A,i) < \sigma(B).
\end{equation}
We get $\Sigma(B \cup \lbrace i \rbrace) = \Sigma(A \cup \lbrace i \rbrace) = \lbrace e_0 \rbrace$.
As $p \geq 2$,
$E(B,i) \setminus \lbrace e_0 \rbrace$ contains at least one edge.
We obtain
\begin{equation}
\label{eqProofCase1v(BCupi)-v(B)=v(BCupi)-Sum_l=1^qv(B_l)}
\Delta_i \hat{v}(B) =
v(B \cup \lbrace i \rbrace) - \overline{v}(B).
\end{equation}
Setting $\tilde A = \bigcup_{j=2}^{p} A_j$,
we have
\begin{equation}
\label{maxlbracksigma(A1,ihatB)10hat65}
\Delta_i \hat{v}(A) =
\left \lbrace
\begin{array}{ll}
v(A \cup \lbrace i \rbrace) - \sum\limits_{k=1}^{p} \overline{v}(A_{k}) & \textrm{if}\ E(A_1,i) \setminus \lbrace e_0 \rbrace \not= \emptyset,\\
v(A_1) + v(\tilde A \cup \lbrace i \rbrace) - \sum\limits_{k=1}^{p} \overline{v}(A_{k}) & \textrm{otherwise}.
\end{array}
\right.
\end{equation}
The superadditivity of $(N,v)$ implies
\begin{equation}
\label{eqv(AUi)>=v(A1)+v(TildeACupi)}
v(A \cup \lbrace i \rbrace) \geq v(A_1) + v(\tilde{A} \cup \lbrace i \rbrace).
\end{equation}
By~(\ref{eqProofCase1v(BCupi)-v(B)=v(BCupi)-Sum_l=1^qv(B_l)}), (\ref{maxlbracksigma(A1,ihatB)10hat65}),
and~(\ref{eqv(AUi)>=v(A1)+v(TildeACupi)}),
(\ref{maxlbracksigma(A1,i)10hat}) is satisfied if the following inequality is satisfied:
\begin{equation}
\label{eqProofCase4v(BCupi)-Sum_l=1^qv(B_l)>=v(ACupi)-Sum_k=1^pv(A_k)}
v(B \cup \lbrace i \rbrace) - \overline{v}(B) \geq
v(A \cup \lbrace i \rbrace) - \sum_{k=1}^{p} \overline{v}(A_{k}).
\end{equation}
By~(\ref{eqProofCase1-2Sigma(B1,i)<Sigma(Bk,i)=M(B,i)Forallk2<=k<=q})
and~(\ref{eqSigma(B,i)=Sigma(B1,i)=Sigma(A1,i)=Sigma(A,i)<Sigma(B)<=Sigma(A)}),
Lemma~\ref{lemmaPathStarCyclePanSum_FInP_Min(B)v(ACapF)=Sum_k=1^mv(A_k)=v(A)} implies
\begin{equation}
\label{eqSumFInPMin(B1)|TildeA1v(F)=SumAkSubseteqB1Overlinev(Ak)}
\sum_{F \in \mathcal{P}_{\min}(B)_{|A}} v(F)=
\sum_{k=1}^{p} \overline{v}(A_k).
\end{equation}
The convexity of $(N,v)$ implies
\begin{equation}
\label{eqProofCase4v(BCupi)-v(B)>=v(ACupi)-v(A)}
v(B \cup \lbrace i \rbrace) - v(B) \geq v(A \cup \lbrace i \rbrace) - v(A). 
\end{equation}
Applying Lemma~\ref{lemmavB+i=1pvBiA}
with $\mathcal{F}= 2^{N} \setminus \lbrace \emptyset \rbrace$
to $A$, $B$,
and to the partition 
$
\lbrace
\mathcal P_{\min}(B)
\rbrace
$
of $B$,
we obtain
\begin{equation}
\label{eqProofCase4v(B)-Sum_l=1^qv(B_l)>=v(A)-Sum_l=1^qSum_FInP_Min(B_l)v(ACapF)}
v(B) - \overline{v}(B)
\geq
v(A) - \sum_{F \in \mathcal P_{\min}(B)_{|A}} v(F).
\end{equation}
By~(\ref{eqSumFInPMin(B1)|TildeA1v(F)=SumAkSubseteqB1Overlinev(Ak)}),
(\ref{eqProofCase4v(B)-Sum_l=1^qv(B_l)>=v(A)-Sum_l=1^qSum_FInP_Min(B_l)v(ACapF)})
is equivalent to
\begin{equation}
\label{eqProofCase4v(B)-Sum_l=1^qv(B_l)>=v(A)-Sum_l=1^qSum_A_kSubseteqB_lv(A_k)}
v(B) - \overline{v}(B)
\geq
v(A) - \sum_{k=1}^{p} \overline{v}(A_{k}).
\end{equation}
(\ref{eqProofCase4v(BCupi)-v(B)>=v(ACupi)-v(A)})
and (\ref{eqProofCase4v(B)-Sum_l=1^qv(B_l)>=v(A)-Sum_l=1^qSum_A_kSubseteqB_lv(A_k)})
imply (\ref{eqProofCase4v(BCupi)-Sum_l=1^qv(B_l)>=v(ACupi)-Sum_k=1^pv(A_k)}).

\vspace{\baselineskip}
\noindent
\textbf{
Case 2
}
We assume  $\sigma(A,i) = M$.
If $|A_k|=1$ for all $k$ with $1 \leq k \leq p$,
then we have $\Sigma(A \cup \lbrace i \rbrace) = E(A,i)$.
We get $\Delta_i \hat{v}(A) = \hat{v}(\lbrace i \rbrace)$
and (\ref{maxlbracksigma(A1,i)10hat}) is satisfied by the superadditivity of $(N,\hat{v})$.
We henceforth assume that there is at least one index $k$
such that $|A_k| \geq 2$.

\vspace{\baselineskip}
\noindent
\textbf{
Subcase 2.1
}
We assume $\sigma(A,i) = M \leq \sigma(A)$.
By~(\ref{maxlbracksigma(A1,i)5}),
we get
\begin{equation}
\label{eqSigma(Ak,i)=M,Forallk,1<=k<=p}
\sigma(A_k,i) = M,\; \forall k,\; 1 \leq k \leq p,
\end{equation}
and
\begin{equation}
\label{eqSigma(Ak,i)=M<=Sigma(Ak),Forallk,1<=k<=p,s.t.|Ak|>=2}
\sigma(A_k,i) = M \leq \sigma(A_k),\; \forall k,\; 1 \leq k \leq p, \, \textrm{ s.t. } |A_k| \geq 2.
\end{equation}
(\ref{eqSigma(Ak,i)=M,Forallk,1<=k<=p}) and (\ref{eqSigma(Ak,i)=M<=Sigma(Ak),Forallk,1<=k<=p,s.t.|Ak|>=2}) imply
$\Sigma(A \cup \lbrace i \rbrace)=
E(A,i) \cup
\bigcup_{k \, |\, \sigma(A_k) = M}
\Sigma(A_k)$.
This implies
\begin{equation}
\label{eqDeltaiHatv(A)=Sumk:Sigma(Ak)>Mv(Ak)-Overlinev(Ak)}
\Delta_i \hat{v}(A) =
\sum_{\lbrace k \,:\, \sigma(A_k)>M \rbrace} v(A_k) - \overline{v}(A_k).
\end{equation}
We can assume that there is at least one index $k$ such that $\sigma(A_k) > M$,
otherwise we have $\Delta_i \hat{v}(A) = 0$ and 
(\ref{maxlbracksigma(A1,i)10hat}) is trivially satisfied by superadditivity of $(N,\hat{v})$.
By~Lemma~\ref{lemminsAi>=sA>=sBorsA=sB>sAi} and (\ref{eqSigma(Ak,i)=M<=Sigma(Ak),Forallk,1<=k<=p,s.t.|Ak|>=2}),
we get
\begin{equation}
\label{eqSigma(Ak,i)=M<Sigma(Ak)=Sigma(Bl),ForallAkSubseteqBlWith|Ak|>=2}
\sigma(A_k,i) = M < \sigma(A_k) = \sigma(B),\; \forall k,\; 1 \leq k \leq p, \, \textrm{ s.t. }  \sigma(A_k) > M.
\end{equation}
By~(\ref{eqSigma(Ak,i)=M<Sigma(Ak)=Sigma(Bl),ForallAkSubseteqBlWith|Ak|>=2}),
we can henceforth assume $\sigma(B) > M$.
If $\sigma(B,i)<M$
(resp. $\sigma(B,i)=M$),
then we have $\Sigma(B \cup \lbrace i \rbrace) = \lbrace e_0 \rbrace$
(resp. $\Sigma(B \cup \lbrace i \rbrace) = E(B,i)$).
As $\sigma(A,i)=M$,
we necessarily have $e_0 = \lbrace i, j \rbrace$.
We get
\begin{equation}
\label{eqDeltaiHatv(B)=v(BUi)-Overlinev(B)IfSigma(B,i)<Mv(B)-Overlinev(B)IfSigma(B,i)=M}
\Delta_i \hat{v}(B) =
\left \lbrace
\begin{array}{ll}
v(B \cup \lbrace i \rbrace) - \overline{v}(B) & \textrm{ if } \sigma(B,i)<M,\\
v(B) - \overline{v}(B) & \textrm{ if } \sigma(B,i)=M.
\end{array}
\right.
\end{equation}
The superadditivity of $(N,v)$ implies
\begin{equation}
\label{eqProofCase1-2Hatv(BUi)-Hatv(B)>=Sum_l:Sigma(Bl)>Mv(Bl)-Overlinev(Bl)}
v(B \cup \lbrace i \rbrace) - \overline{v}(B)
\geq
v(B) - \overline{v}(B).
\end{equation}
By~(\ref{eqDeltaiHatv(A)=Sumk:Sigma(Ak)>Mv(Ak)-Overlinev(Ak)}),
(\ref{eqDeltaiHatv(B)=v(BUi)-Overlinev(B)IfSigma(B,i)<Mv(B)-Overlinev(B)IfSigma(B,i)=M}),
and~(\ref{eqProofCase1-2Hatv(BUi)-Hatv(B)>=Sum_l:Sigma(Bl)>Mv(Bl)-Overlinev(Bl)}),
(\ref{maxlbracksigma(A1,i)10hat}) is satisfied if the following inequality is satisfied:
\begin{equation}
\label{eqProofCase1-2Sum_l:Sigma(Bl)>Mv(Bl)-Overlinev(Bl)>=Sum_l:Sigma(Bl)>MSum_AkSubseteqBlv(Ak)-Overlinev(Ak)}
v(B) - \overline{v}(B)
\geq
\sum_{\lbrace k \,:\, \sigma(A_k)>M \rbrace} v(A_k) - \overline{v}(A_k).
\end{equation}
Applying Lemma~\ref{lemmavB+i=1pvBiA}
with $\mathcal{F}= 2^{N} \setminus \lbrace \emptyset \rbrace$
to the subsets $\tilde A=\bigcup_{\lbrace k \,:\, \sigma(A_k)>M \rbrace} A_k$ and $B$,
and to the partition $\mathcal P_{\min}(B)$ of $B$,
we obtain
\begin{equation}
\label{eqProofCase1-2v(Bl)-Overlinev(Bl)>=v(TildeA1)-Sum_FInP_min(Bl)v(TildeAlCapF)}
v(B) - \overline{v}(B)
\geq
v(\tilde A) - \sum_{F\in\mathcal P_{\min}(B)_{|\tilde A}} v(F).
\end{equation}
By superadditivity of $(N,v)$,
(\ref{eqProofCase1-2v(Bl)-Overlinev(Bl)>=v(TildeA1)-Sum_FInP_min(Bl)v(TildeAlCapF)}) implies
\begin{equation}
\label{eqProofCase1-2v(Bl)-Overlinev(Bl)>=Sum_AkSubseteqBlv(Ak)-Sum_FInP-min(Bl)v(TildeAlCapF}
v(B) - \overline{v}(B)
\geq
\sum_{\lbrace k \,:\, \sigma(A_k)>M \rbrace} v(A_k) - \sum_{F\in\mathcal P_{\min}(B)_{|\tilde A}} v(F).
\end{equation}
As $\sigma(B) > M$,
there exists an edge $e \in E$ connected to $B$
with weight $w(e) < \sigma(B)$.
Then,
as (\ref{eqSigma(Ak,i)=M<Sigma(Ak)=Sigma(Bl),ForallAkSubseteqBlWith|Ak|>=2}) is satisfied,
Lemma~\ref{lemmaPathStarCyclePanSum_FInP_Min(B)v(ACapF)=Sum_k=1^mv(A_k)=v(A)}
applied to $\tilde{A}$ and $B$
implies
\begin{equation}
\label{eqSumFInPmin(Bl)|TildeAlv(F)=SumAkSubseteqBlOverlinev(Ak)-bis}
\sum_{F \in \mathcal{P}_{\min}(B)_{|\tilde{A}}} v(F)=
\sum_{\lbrace k \,:\, \sigma(A_k)>M \rbrace} \overline{v}(A_k).
\end{equation}
(\ref{eqProofCase1-2v(Bl)-Overlinev(Bl)>=Sum_AkSubseteqBlv(Ak)-Sum_FInP-min(Bl)v(TildeAlCapF})
and (\ref{eqSumFInPmin(Bl)|TildeAlv(F)=SumAkSubseteqBlOverlinev(Ak)-bis})
imply (\ref{eqProofCase1-2Sum_l:Sigma(Bl)>Mv(Bl)-Overlinev(Bl)>=Sum_l:Sigma(Bl)>MSum_AkSubseteqBlv(Ak)-Overlinev(Ak)})
and therefore (\ref{maxlbracksigma(A1,i)10hat}) is satisfied.

\vspace{\baselineskip}
\noindent
\textbf{
Subcase 2.2
}
We assume $\sigma(A) < \sigma(A,i) = M$.\\
Let us assume $\sigma(B,i) < M$.
Then,
$e_0 = \lbrace i, j \rbrace$ is the only edge in $E(B,i)$
with $w(e_0) = \sigma(B,i)$
and $w(e)=M$ for all $e\in E(B,i) \setminus \lbrace e_0 \rbrace$.
We can assume w.l.o.g. $\sigma(A_1) = \sigma(A)$.
Let us consider an edge $e_1 = \lbrace i, j_1 \rbrace$ with $j_1 \in A_1$ and an edge $\tilde{e}_1$ in $\Sigma(A_1)$. 
Let $\gamma$ be a shortest path in $G_{A_1}$ linking $j_1$ to an end-vertex of $\tilde{e}_1$
($\gamma$ may be reduced to $j_1$).
Then the Path condition applied to $\lbrace e_0 \rbrace \cup \lbrace e_1 \rbrace \cup \gamma \cup \lbrace \tilde{e}_1 \rbrace$
implies $M = w_1 \leq \max (w(e_0),w(\tilde{e}_1))$,
a contradiction.
Thus,
we necessarily have $\sigma(B,i) = \sigma(A,i) = M$
and (\ref{maxlbracksigma(A1,i)5}) implies
\begin{equation}
\label{eqCase3sigma(B,i)=sigma(B_1,i)=sigma(B_2,i)=sigma(B_q,i)=MIdemForA}
\sigma(B,i) = \sigma(A,i) = \sigma(A_1,i) = \cdots = \sigma(A_p,i) = M.
\end{equation}
Moreover,
as $\sigma(B) \leq \sigma(A) < M$,
Claim~\ref{itemLemmaIfSigma(B,i)=M(B,i)And|B1|>=2,Thenmax(Sigma(B1),M(B,i))<=Sigma(B2)}
of Lemma~\ref{Ifsigma(B)<igma(B1,thenforall4} implies
\begin{equation}
\label{eqProofCase3sigma(B)=sigma(B1)<M=sigmaB2<=sigma(Bq)IdemForA}
\sigma(A) = \sigma(A_1) < M \leq \sigma(A_k),\, \forall k,\, 2 \leq k \leq p, \textrm{ s.t. } |A_k| \geq 2,
\end{equation}
with $|A_1| \geq 2$,
after renumbering if necessary.
As $\sigma(A_1) < M$,
(\ref{eqCase3sigma(B,i)=sigma(B_1,i)=sigma(B_2,i)=sigma(B_q,i)=MIdemForA}) implies
\begin{equation}
\label{eqSigma(B)<=Sigma(A1)<M=Sigma(Ak,i),Forallk,1<=k<=p}
\sigma(B) \leq \sigma(A_1) < M = \sigma(A_k,i),\, \forall k,\, 1 \leq k \leq p.
\end{equation}
(\ref{eqProofCase3sigma(B)=sigma(B1)<M=sigmaB2<=sigma(Bq)IdemForA})
and (\ref{eqSigma(B)<=Sigma(A1)<M=Sigma(Ak,i),Forallk,1<=k<=p}) imply 
$\Sigma(A\cup \lbrace i \rbrace) = \Sigma(A_1)$
and $\Sigma(B\cup \lbrace i \rbrace) = \Sigma(B)$.
Let $A_{1,1}, A_{1,2}, \ldots, A_{1,k_1}$
(resp. $B_{1}, B_{2}, \ldots, B_{l_1}$) be the blocks of
$\mathcal{P}_{\min}(A_1)$
(resp. $\mathcal{P}_{\min}(B)$)
linked to $i$ by an edge in $E(A_1, i)$
(resp. $E(B, i)$).
Let us set $J_{A_1} = \lbrace 1, \ldots, k_1 \rbrace$ and $J_B= \lbrace 1, \ldots, l_1\rbrace$.
Let $K_A$ be the set of indices $k \in \lbrace 2, \ldots, p \rbrace$.
Setting
$\tilde A = \bigcup_{j \in J_{A_1}} A_{1,j} \cup \bigcup_{k \in K_A} A_k$
and
$\tilde B = \bigcup_{j \in J_B} B_{j}$,
we get
\begin{equation}
\label{EqProofCase7.2v(ACupi)-v(A)=v(TildeACupi)-Sum_i=1^k_1v(A_1,i)-Sum_k=2^pv(A_k)}
\Delta_i \hat{v}(A) =
v(\tilde A \cup\lbrace i \rbrace)
- \sum_{j \in J_{A_1}} v(A_{1,j})
- \sum_{k \in K_A} \overline{v}(A_{k}),
\end{equation}
and
\begin{equation}
\label{eqProofCase7.2Hatv(BCupi)-Hatv(B)=v(TildeBCupi)-Sum_i=1^l_1v(B_1,i)-Sum_l=2^qv(B_l)}
\Delta_i \hat{v}(B) =
v(\tilde B \cup \lbrace i \rbrace)
- \sum_{j \in J_B} v(B_{j}).
\end{equation}
By~(\ref{eqSigma(B)<=Sigma(A1)<M=Sigma(Ak,i),Forallk,1<=k<=p}),
Lemma~\ref{Lemmapartitionin2set2} applied to $A_1 \cup \bigcup_{k \in K_A} A_k$ and $B$
implies $|J_{A_1}| = |K_A| = 1$, and
\begin{equation}
\label{eqA1,1SubseteqB11}
A_{1,1} \subseteq B_{1},
\end{equation}
and
\begin{equation}
\label{eq|A2|=1AndA2SubseteqB12,IfKANot=Emptyset}
|A_2| = 1 \textrm{ and } A_2 \subseteq B_{2},
\end{equation}
after renumbering if necessary.
(\ref{eqA1,1SubseteqB11}) and (\ref{eq|A2|=1AndA2SubseteqB12,IfKANot=Emptyset}) correspond
exactly to (\ref{eqA1,jSubseteqB1j,ForalljInJA}) and (\ref{eq|Ak|=1AndAkSubsetqB1k,ForammkInKA})
and we can conclude as in Subcase~1.1.\\

We can finally prove that (\ref{maxlbracksigma(A1,i)10hat}) is satisfied by induction on $|B \setminus A|$.
By the previous reasoning,
it is true if $|B \setminus A| = 1$.
If $|B \setminus A| \geq 2$,
then we can select a vertex $j \in B \setminus A$
and set $B' = B \setminus \lbrace j \rbrace$.
Applying (\ref{maxlbracksigma(A1,i)10hat})
to $B' \subseteq B \subseteq N \setminus \lbrace i \rbrace$,
we get
\begin{equation}
\label{eqDeltaiHatv(B'Uj)>=DeltaiHatv(B')}
\Delta_i \hat{v}(B) \geq \Delta_i \hat{v}(B').
\end{equation}
By induction,
we also have
\begin{equation}
\label{eqDeltaiHatv(B')>=DeltaiHatv(A)}
\Delta_i \hat{v}(B') \geq \Delta_i \hat{v}(A).
\end{equation}
(\ref{eqDeltaiHatv(B'Uj)>=DeltaiHatv(B')}) and (\ref{eqDeltaiHatv(B')>=DeltaiHatv(A)})
imply (\ref{maxlbracksigma(A1,i)10hat}).
\end{proof}

\section{Conclusion}
\label{Conclusion}
Our main result gives a characterization of the weighted graphs
satisfying inheritance of convexity for the correspondence $\tilde{\mathcal{P}}_{\min}$.
This characterization implies strong restrictions on the structure of the underlying graph
and the relative positions of edge-weights.
But it also highlights
that
the class of graphs satisfying inheritance of convexity with $\tilde{\mathcal{P}}_{\min}$
is much larger
than the one satisfying inheritance of convexity with $\mathcal{P}_{\min}$.
Indeed,
the characterization of inheritance of convexity for the correspondence $\mathcal{P}_{\min}$
obtained by~\cite{Skoda2019b}
implies a restriction of the edge-weights to at most three different values
and cycle-completeness of large subgraphs is required.
In contrast,
inheritance of convexity for the correspondence $\tilde{\mathcal{P}}_{\min}$
allows an arbitrary number of edge-weights.
Moreover,
the reinforced conditions on cycles and pans
imply the existence of some specific chords in some cycles
but cycle-completeness does not come into play.
\cite{Skoda2019b} established that inheritance of convexity for $\mathcal{P}_{\min}$
can be verified in polynomial time.
It would be interesting to investigate the complexity of the problem
of deciding whether inheritance of convexity is satisfied for the correspondence $\tilde{\mathcal{P}}_{\min}$.
It is already known that inheritance of $\mathcal{F}$-convexity
can be checked in polynomial time (\cite{Skoda2016-3}).
Therefore,
the first five conditions of our characterization can also be checked in polynomial time.
It can be easily seen
that the Reinforced Adjacent Cycles condition can be checked in polynomial time
as we only have to consider cycles of size $4$ or $5$.
It remains to study the complexity of checking the Reinforced Cycle and the Reinforced Pan conditions.
It seems more difficult as there is no obvious limitation on the number or size of paths and cycles
involved in these last conditions.


\bibliographystyle{apalike}
\bibliography{biblio}

\end{document}